\documentclass{article}

\usepackage[preprint]{neurips_2026}
\usepackage[utf8]{inputenc}
\usepackage[T1]{fontenc}
\usepackage{graphicx} 
\usepackage[title]{appendix}
\usepackage{amsmath}
\usepackage{amsfonts}
\usepackage{amssymb}
\usepackage{amsthm}
\usepackage{booktabs}
\usepackage{multirow}
\usepackage{enumitem}
\usepackage{setspace}
\usepackage{algorithm}
\usepackage{algorithmic}
\usepackage{xcolor}
\usepackage{colortbl}
\usepackage{float}
\usepackage{hyperref}
\usepackage{tikz}
\usetikzlibrary{arrows.meta, calc, positioning, decorations.pathreplacing}
\usepackage[capitalize,noabbrev]{cleveref}

\newtheorem{theorem}{Theorem}
\newtheorem{assumption}{Assumption}
\newtheorem{corollary}{Corollary}
\newtheorem{lemma}{Lemma}
\newtheorem{proposition}{Proposition}
\newtheorem{remark}{Remark}

\title{TEA-Time: \textbf{T}ransporting \textbf{E}ffects \textbf{A}cross \textbf{Time}}

\author{
  Harsh Parikh\thanks{Corresponding author: \texttt{harsh.parikh@yale.edu}} \\
  Amazon SCOT, Seattle, USA\\
  Yale University, New Haven, USA
  \And
  Gabriel Levin-Konigsberg \\
  Amazon SCOT,
  Seattle, USA
  \And
  Dominique Perrault-Joncas \\
  Amazon SCOT,
  Seattle, USA
  \And
  Alexander Volfovsky\thanks{Corresponding author: \texttt{alexander.volfovsky@duke.edu}} \\
  Amazon SCOT, Seattle, USA\\
  Duke University, Durham, USA
}

\begin{document}

\maketitle

\begin{abstract}
Treatment effects estimated from a randomized controlled trial are local not only to the study population but also to the time at which the trial was conducted. The literature on generalizing experimental findings to new populations is extensive, yet transporting effects across time has received far less attention, and even defining the target estimand is nonobvious. We formalize the transported average treatment effect under a separable temporal effects assumption, derive two identification strategies---replicated trials and common arm---and develop doubly robust, semiparametrically efficient estimators for each. Applied to a large archive of headline A/B tests, the common arm strategy is substantially more precise but exhibits systematic bias when the temporal factor depends on the gap between intervention and measurement rather than on measurement time alone, while the replicated trials strategy, which allows this dependence, tracks the ground truth more faithfully. Simulation studies investigate when each strategy is reliable and when it silently fails.
\end{abstract}

\section{Introduction}
\label{sec:introduction}
Causal effects identified by randomized controlled trials (RCTs) are inherently local---not only to the population studied but also to the time at which the trial was conducted. A marketing intervention tested during summer may perform differently in the holiday season; a job training program evaluated during economic expansion may have different effects during recession; a drug's efficacy may vary with seasonal disease prevalence. Despite the ubiquity of such temporal variation, the causal inference literature has devoted considerably more attention to transporting effects across populations than across time.

This paper develops a framework for \emph{temporal transportation}: using data from trials conducted at one time to estimate treatment effects at different times. The setting we consider is one where an organization has access to multiple historical RCTs---potentially comparing different interventions---and seeks to predict what the effect of a specific treatment would be if administered at a new target time. Applications include e-commerce promotions, digital advertising, labor market policy, and clinical trials.

\begin{figure}[!h]
\centering
\begin{tikzpicture}[
    trial/.style={rectangle, rounded corners=1.5pt, minimum height=0.45cm, draw, thick, font=\scriptsize},
    primary/.style={trial, fill=blue!5, draw=blue!90},
    anchorA/.style={trial, fill=orange!50, draw=orange!20!black!50},
    anchorB/.style={trial, fill=green!50!black!10, draw=green!20!black!50},
    transported/.style={trial, fill=blue!15, draw=blue!60, dashed},
    infoarrow/.style={-{Stealth[length=2mm]}, thick, black!70},
    temporalarrow/.style={-{Stealth[length=2mm]}, thick, densely dashed, black!50},
]

\def\rowsep{0.75}
\def\colwidth{1.8}

\node[font=\scriptsize\bfseries, black!70] at (0.9, 0.5) {Source};
\node[font=\scriptsize\bfseries, gray!70] at (0.9+\colwidth, 0.5) {$\to \dots \to$};
\node[font=\scriptsize\bfseries, black!70] at (0.9+2*\colwidth, 0.5) {Target};

\node[font=\scriptsize, right] at (-1.6, 0.5*\rowsep) {\underline{Arms}};
\node[font=\scriptsize, right] at (-1.6, 0) {$a$ vs $b$};
\node[font=\scriptsize, right] at (-1.6, -\rowsep) {$c$ vs $d$};
\node[font=\scriptsize, right] at (-1.6, -2*\rowsep) {$e$ vs $b$};

\node[primary, minimum width=1.1cm] (P1) at (0.9, 0) {Trial $k^\star$};
\node[transported, minimum width=0.9cm] (T1) at (0.9+2*\colwidth, 0) {TATE};

\node[anchorA, minimum width=0.7cm] (A1) at (0.9, -\rowsep) {};
\node[anchorA, minimum width=1.6cm] (A2) at (0.9+\colwidth, -\rowsep) {};
\node[anchorA, minimum width=0.7cm] (A3) at (0.9+2*\colwidth, -\rowsep) {};

\node[anchorB, minimum width=1.7cm] (B1) at (1.2, -2*\rowsep) {};
\node[anchorB, minimum width=0.8cm] (B3) at (0.9+2*\colwidth, -2*\rowsep) {};

\draw[infoarrow] (P1.south) -- (A1.north);
\draw[infoarrow] ($(P1.south)-(0.1,0)$) -- ($(A1.north)-(0.1,0)$) -- ($(B1.north)-(0.1,0)$);

\draw[temporalarrow] (A1.east) -- (A2.west);
\draw[temporalarrow] (A2.east) -- (A3.west);
\draw[temporalarrow] (B1.east) -- (B3.west);

\draw[infoarrow] (A3.north) -- (T1.south);
\draw[infoarrow] ($(B3.north)+(0.1,0)$) -- ($(A3.south)+(0.1,0)$) -- ($(T1.south)+(0.1,0)$);

\draw[decorate, decoration={brace, amplitude=4pt, raise=2pt}, gray] 
    (0, -2*\rowsep-0.3) -- (0, -\rowsep+0.3);
\node[font=\tiny, gray, left] at (-0.15, -1.5*\rowsep) {Anchors};

\end{tikzpicture}
\caption{Temporal transportation uses anchor trials to extrapolate treatment effects across time. Rows are treatment comparisons; columns are times. Solid boxes are ATEs directly estimable from each trial's data. Solid arrows denote information flow within a time period (from the target trial to contemporary anchors). The target TATE is identified by chaining anchor ATEs across time without requiring the primary comparison at the target time.}
\label{fig:transport_schematic}
\end{figure}

The challenge of temporal transportation differs fundamentally from cross-population generalization. When transporting across populations, we typically observe covariates in both source and target and can reweight to adjust for distributional shifts \citep{dahabreh2019generalizing, stuart2011use}. When transporting across time, we cannot observe outcomes under the target timing---by definition, we are extrapolating to a period where the relevant trial was not conducted. Identification therefore requires structural assumptions about how treatment effects vary with time. Our key insight is that \emph{other} trials---comparing potentially different treatments---can serve as temporal anchors (Figure~\ref{fig:transport_schematic}): if anchor trials testing other treatment comparisons were conducted at both source and target times, we can use them to learn how outcomes or effects evolve temporally. To the best of our knowledge, this is the first work to formalize temporal transportation of treatment effects as an identification problem, distinct from cross-population transportability \citep{dahabreh2019generalizing, stuart2011use} (which assumes shared covariate support) and from meta-analysis and panel-data methods (which estimate effects where treatments \emph{were} administered).

\textbf{Contributions.} Our paper makes four contributions. First, we formalize the temporal transportation problem and introduce the \emph{transported average treatment effect} (TATE) as the target estimand. Under a \emph{separable temporal effects} assumption---a multiplicative decomposition of conditional-mean potential outcomes into treatment- and covariate-specific effects multiplied by a shared temporal effect---we show that the TATE equals the product of an observed average treatment effect and an identifiable temporal ratio. Second, we provide two identification strategies with distinct data requirements. The \emph{replicated trials} strategy uses pairs of trials comparing identical treatments at different times, permitting flexible temporal structures but requiring exact replication. The \emph{common arm} strategy identifies the temporal ratio from any treatment arm observed at multiple times---such as a control arm appearing across many trials---but imposes a stronger restriction that temporal effects operate only through measurement time. Third, since separability is structural and cannot be falsified from a single pair of time points, we derive \emph{partial identification bounds} on the TATE as a function of the violation magnitude; the resulting \emph{breakdown frontier} summarizes robustness that a practitioner can report alongside the point estimate. Fourth, we develop doubly robust estimators for the TATE under each strategy that are consistent if either the outcome model or the propensity scores are correctly specified and achieve the semiparametric efficiency bound when both are correct; for the common arm strategy with multiple anchors, we derive the optimally weighted combination. We evaluate both strategies through Monte Carlo simulations across three data generating processes and an application to over 32,000 A/B tests from the Upworthy Research Archive \citep{matias2021upworthy}: the common arm approach offers greater precision but incurs bias when the temporal effect depends on both intervention and measurement times, and the breakdown frontier quantifies exactly how fragile each conclusion is.

\section{Setup and Assumptions}
\label{sec:preliminaries}

\textbf{Notation.} Let $\mathcal{T}$ denote discrete time periods, $\mathcal{P}$ a population of units, and $\mathcal{A} \cup \{0\}$ the set of treatments with $0$ as control. We observe a collection of RCTs $\mathcal{K}$. Each trial $k \in \mathcal{K}$ draws a sample $\mathcal{S}_k \subset \mathcal{P}$ of $n_k$ units, compares treatment $a_k$ against baseline $b_k$, administers treatment at time $t_{0k}$, and measures outcomes at time $t_{1k} \geq t_{0k}$. Units in trial $k$ are randomly assigned to $A \in \{a_k, b_k\}$. Let $X$ denote pre-treatment covariates and $S \in \mathcal{K}$ denote trial membership.

Define $Y_{t_1}(a, t_0)$ as the potential outcome under treatment $a$ administered at $t_0$ and measured at $t_1$. The observed outcome for a unit in trial $k$ is $Y_{t_{1k}} = Y_{t_{1k}}(A, t_{0k})$.

\textbf{Standard assumptions.} We maintain: (i) \emph{random assignment} $A \perp \{Y_{t_1}(a, t_0)\} \mid S = k, X$ for each trial $k$\footnote{This conditional form accommodates both simple randomization (where $A \perp X \mid S = k$) and covariate-adaptive designs.}; (ii) \emph{SUTVA}: no interference, no hidden treatment versions; (iii) \emph{trial-timing exogeneity}: trial timing $(t_{0k}, t_{1k})$ is independent of the potential outcomes at the scheduling stage---automatically satisfied by predetermined cadences (Appendix~\ref{app:trial_timing}); and (iv) \emph{compositional stability}: $X \perp S$ at the super-population level, so the pooled mean $\mathbb{E}[\mu_{a,k}(X)]$ equals the trial-$k$ marginal $\mathbb{E}[Y\mid A=a, S=k]$. Under drift ($X \not\perp S$), the trial-$k^\star$ marginal requires a density-ratio reweighting of our doubly robust score (Appendix~\ref{app:assumption_discussion}). \emph{Individual non-stationarity}---$\boldsymbol{\theta}_a(X)$ changing over time for the same unit---violates separability at the unit level and cannot be resolved through reweighting.

\textbf{Separable temporal effects.} To enable temporal transportation, we impose multiplicative structure on how outcomes vary with timing.

\begin{assumption}[Separable Temporal Effects]
\label{ass:separable}
For a fixed rank $R \geq 1$, the conditional mean of potential outcomes factorizes as
\[
\mathbb{E}[Y_{t_1}(a, t_0) \mid X] \;=\; \boldsymbol{\theta}_a(X)^\top \boldsymbol{\Lambda}(t_0, t_1),
\]
where $\boldsymbol{\theta}_a(X) \in \mathbb{R}^R$ is a vector of treatment- and covariate-specific effects and $\boldsymbol{\Lambda}(t_0, t_1) \in \mathbb{R}^R$ is a vector of temporal effects common across units and treatments.
\end{assumption}

Assumption~\ref{ass:separable} restricts only the conditional-mean surface of $Y$; it is agnostic about the noise structure and implies nothing about higher moments. All identification, partial-identification, and efficient-estimation results in this paper use only this mean-level statement. We occasionally write the equivalent generative form $Y_{t_1}(a, t_0) = \boldsymbol{\theta}_a(X)^\top \boldsymbol{\Lambda}(t_0, t_1) + \epsilon_{t_1}(a)$ with $\mathbb{E}[\epsilon_{t_1}(a) \mid X, S] = 0$ when a unit-level expression clarifies an argument, but nothing beyond the mean-zero conditional noise is used.

This is a rank-$R$ factor model: outcomes are an inner product of $R$ (effect, temporal-effect) pairs. The shared temporal-effect vector $\boldsymbol{\Lambda}$ is the ingredient that couples treatment arms across time. Separability is a priori plausible when a small number of external drivers multiplicatively scale heterogeneous unit-level outcomes---e.g., seasonality in digital A/B testing \citep{bai2009panel, athey2021matrix}; disease prevalence in clinical effectiveness; macroeconomic cycles in labor-market policy \citep{abadie2010synthetic}---and less credible when treatment-specific temporal dynamics dominate (quantified by Proposition~\ref{prop:sensitivity}). Appendix~\ref{app:impossibility} shows this coupling is essential: without cross-treatment structural restriction the TATE is not point-identified, and local-smoothness substitutes yield identified-set widths that grow linearly with the temporal gap to the nearest observed trial of the target treatment. We treat $R$ as fixed and known; rank selection and the implied treatment-contrast identity are in Appendix~\ref{app:assumption_discussion}. The framework operates in discrete time and accommodates arbitrary structural breaks in $\boldsymbol{\Lambda}$, provided anchor trials bridge them.

\textbf{Target estimand.} The \emph{transported average treatment effect} (TATE) for trial $k^\star$ is $\tau_{k^\star}(\delta_0, \delta_1) \equiv \tau(a_{k^\star},b_{k^\star},t_{0k^\star},t_{1k^\star},\delta_0, \delta_1)$, defined as
\begin{align}
 \mathbb{E}\left[
\begin{bmatrix}
    Y_{t_{1k^\star} + \delta_1}(a_{k^\star}, t_{0k^\star} + \delta_0) - Y_{t_{1k^\star} + \delta_1}(b_{k^\star}, t_{0k^\star} + \delta_0) 
\end{bmatrix} \mid S = k^\star \right],
\end{align}
representing the ATE for the trial $k^\star$ population had treatment occurred at time $t_{0k^\star} + \delta_0$ with outcomes measured at $t_{1k^\star} + \delta_1$.

\section{Identification}
\label{sec:identification}

We identify the TATE in two steps. First, we decompose it into an observed ATE and a temporal ratio. Second, we show that two different auxiliary-data configurations identify the temporal ratio under different structural assumptions. All proofs are in Appendix~\ref{app:proofs}.

\begin{theorem}[TATE Decomposition]
\label{thm:decomposition}
Under random assignment, SUTVA, and Assumption~\ref{ass:separable}:
\begin{align}
\label{eq:tate_decomp}
\tau_{k^\star}(\delta_0, \delta_1) = \bar{\tilde{\boldsymbol{\theta}}}_{a_{k^\star}, b_{k^\star}, k^\star}^\top \boldsymbol{\Lambda}(t_{0k^\star} + \delta_0, t_{1k^\star} + \delta_1),
\end{align}
where $\bar{\tilde{\boldsymbol{\theta}}}_{a_{k^\star}, b_{k^\star}, k^\star} := \mathbb{E}[\tilde{\boldsymbol{\theta}}_{a_{k^\star}, b_{k^\star}}(X) \mid S = k^\star] \in \mathbb{R}^R$ is the mean effect contrast for trial $k^\star$, and $\tau_{k^\star}(0,0) = \bar{\tilde{\boldsymbol{\theta}}}_{a_{k^\star}, b_{k^\star}, k^\star}^\top \boldsymbol{\Lambda}(t_{0k^\star}, t_{1k^\star})$ is the observed source-time ATE.
\end{theorem}

Under separability, the observed and transported treatment effects share the same treatment- and covariate-specific effect contrast and differ only through the temporal effect. Identification therefore reduces to recovering $\boldsymbol{\Lambda}$ at the target time from auxiliary trials. Without a cross-treatment coupling of this kind, the TATE is not point-identified: Appendix~\ref{app:impossibility} shows the identified set is all of $\mathbb{R}$ in the nonparametric model, and has width proportional to the temporal gap to the nearest observed trial of the target treatment under any local-smoothness restriction---shared effect structure, not smoothness in time, is what delivers identification. In practice, $R = 1$ is the dominant case: outcomes have a single temporal effect, and~\eqref{eq:tate_decomp} reduces to a scalar product of observed ATE and a temporal ratio. Appendix~\ref{app:R1_case} collects the $R = 1$ specializations of every result that follows; the main text is stated for general $R$ throughout. We now give two ways to identify the temporal effect.

\textbf{Strategy 1: Replicated trials.} The first strategy uses trials that compare the \emph{same non-target treatment pair} at both source and target times as temporal anchors. To estimate the effect of $A$ vs $B$ at July when no such trial exists in July, we use a $C$-vs-$D$ trial conducted at both January and July: the ratio of its ATEs identifies how effects scale temporally, and multiplying the target trial's January $A$-vs-$B$ ATE by this ratio gives the transported effect. Formally, suppose $R$ anchor pairs $(a_r^*, b_r^*)$ are observed at both source and target times, each drawn by random sampling with nonzero mean treatment contrast, and their effect contrasts $\boldsymbol{\Theta} := [\bar{\tilde{\boldsymbol{\theta}}}_{a_1^*, b_1^*}, \ldots, \bar{\tilde{\boldsymbol{\theta}}}_{a_R^*, b_R^*}]^\top$ are jointly full rank (Assumption~\ref{ass:replicated} in Appendix~\ref{app:assumption_discussion}).

\begin{corollary}[Strategy 1 identification]
\label{cor:replicated}
Under Assumption~\ref{ass:separable} and Assumption~\ref{ass:replicated}, the temporal-effect vector at the target time is identified by $\boldsymbol{\Lambda}^{\mathrm{target}} = \boldsymbol{\Theta}^{-1} \boldsymbol{\tau}^{\mathrm{target}}$, where $\boldsymbol{\tau}^{\mathrm{target}}$ stacks the $R$ anchor-pair ATEs at the target time.
\end{corollary}

The anchor pair is \emph{different} from the target pair in the non-trivial case: if a trial testing the target pair already exists at the target time, no transportation is needed. This cross-treatment temporal borrowing is the step that distinguishes our framework from methods that estimate effects only where treatments are administered. The cost is requiring exact replication of some non-target treatment pair, which is feasible in repeat A/B tests and platform trials with standard arms but demanding elsewhere.

\textbf{Strategy 2: Common arm.} An alternative requires only that $R$ treatment arms $c_1^*, \ldots, c_R^*$ appear at both source and target measurement times, with nonzero mean outcome effects jointly spanning $\mathbb{R}^R$ (Assumption~\ref{ass:common}). This is often easier to satisfy---control arms typically appear across many trials---but requires an additional restriction:

\begin{assumption}[Measurement-Time Structure]
\label{ass:measurement}
The temporal modifier depends only on measurement time: $\boldsymbol{\Lambda}(t_0, t_1) = \boldsymbol{\Lambda}(t_1)$ for all $t_0 \leq t_1$.
\end{assumption}

Assumption~\ref{ass:measurement} is appropriate when temporal variation reflects conditions at outcome measurement (seasonal demand, economic conditions, user engagement cycles) rather than elapsed time since treatment administration; it rules out effect decay.

\begin{corollary}[Strategy 2 identification]
\label{cor:common_arm}
Under Assumptions~\ref{ass:separable}, \ref{ass:measurement}, and \ref{ass:common}, the temporal-effect vector at the target time is identified by $\boldsymbol{\Lambda}(t_{1k^\star} + \delta_1) = \boldsymbol{\Theta}_c^{-1} \boldsymbol{\mu}^{\mathrm{target}}$, where $\boldsymbol{\Theta}_c := [\bar{\boldsymbol{\theta}}_{c_1^*}, \ldots, \bar{\boldsymbol{\theta}}_{c_R^*}]^\top$ and $\boldsymbol{\mu}^{\mathrm{target}}$ stacks the $R$ anchor-arm conditional means at the target measurement time.
\end{corollary}

Strategy~2 identifies the temporal ratio from conditional \emph{means} rather than treatment \emph{contrasts}. Because contrasts typically have much higher relative noise than means, this yields the substantial efficiency gain quantified in Section~\ref{sec:estimation}. The anchor arms can be any treatments, not just control.

\textbf{Why separability, and why not something weaker?} Assumption~\ref{ass:separable} plays a sharp structural role. Without any cross-treatment coupling, the TATE is not point-identified: two DGPs can match every observed-data distribution yet disagree on the TATE by any real number (Appendix~\ref{app:impossibility}). Replacing separability with local-smoothness restrictions on $\mu_a$ in time---Lipschitz-$L$, bounded variation, Sobolev---does not help: the identified-set width is at least $2L\Delta_{k^\star}$, where $\Delta_{k^\star}$ is the temporal gap to the nearest observed trial of the target treatment. What resolves the non-identification is a \emph{shared} temporal-effect vector linking arms, and Assumption~\ref{ass:separable} is the minimal such restriction---$R$ independent anchor pairs (Strategy~1) or arms (Strategy~2) then suffice, collapsing to a single anchor when $R = 1$.

\subsection{Partial Identification under Separability Violations}
\label{sec:sensitivity}

Assumption~\ref{ass:separable} is a structural restriction that cannot be falsified from a single pair of time points. A practitioner therefore needs to know not only what the TATE is if separability holds but also how much the conclusion would move if it fails by some amount. We derive a breakdown-frontier characterization of the identified set as a function of the magnitude of separability violations, giving a Rosenbaum-style sensitivity summary tailored to temporal transportation. Consider the contaminated model
\begin{equation}
\label{eq:contaminated}
Y_{t_1}(a, t_0) = \boldsymbol{\theta}_a(X)^\top \boldsymbol{\Lambda}(t_0, t_1) + \Gamma_a(X, t_0, t_1) + \epsilon_{t_1},
\end{equation}
with outcome-level sensitivity parameter $\gamma := \sup_{a,x,t_0,t_1} |\Gamma_a(x, t_0, t_1)|$ and contrast-level sensitivity $\tilde{\gamma} := \sup_{a,b,t_0,t_1} |\mathbb{E}[\Gamma_a(X, t_0, t_1) - \Gamma_b(X, t_0, t_1)]| \leq 2\gamma$; separability corresponds to $\gamma = \tilde{\gamma} = 0$.

\begin{proposition}[Breakdown frontier, first-order]
\label{prop:sensitivity}
Under the contaminated model~\eqref{eq:contaminated}, the conditions of Corollary~\ref{cor:replicated}, and $\sigma_{\min}(\boldsymbol{\Theta}) > 0$, to first order in $\gamma$,
\begin{equation}
\label{eq:bias_bounds}
|\mathrm{Bias}_1| \;\leq\; 2\gamma \cdot \bigl(1 + \|\boldsymbol{\Lambda}^{\mathrm{target}}\|\bigr) \cdot \kappa(\boldsymbol{\Theta}) \cdot \left(1 + \frac{\|\bar{\tilde{\boldsymbol{\theta}}}_{k^\star}\|}{\sigma_{\min}(\boldsymbol{\Theta})}\right) + o(\gamma),
\end{equation}
where $\kappa(\boldsymbol{\Theta}) = \sigma_{\max}(\boldsymbol{\Theta})/\sigma_{\min}(\boldsymbol{\Theta})$ is the condition number of the anchor-pair effect-contrast matrix. The Strategy~2 bound takes the same form with $\boldsymbol{\Theta}$ replaced by $\boldsymbol{\Theta}_c$.
\end{proposition}

The bound is informative in the well-conditioned regime ($\sigma_{\min}(\boldsymbol{\Theta})$ bounded away from zero) and is vacuous as $\sigma_{\min}(\boldsymbol{\Theta}) \to 0$; in that limit the population estimand itself diverges (Theorem~\ref{thm:decomposition} requires $\boldsymbol{\Theta}^{-1}$), so the blow-up is an honest signal rather than a slack in the bound. The reciprocal-of-$\sigma_{\min}$ factor is therefore a diagnostic for anchor near-collinearity: when it dominates the bound, the practitioner should choose anchor pairs with more linearly independent effect contrasts (Appendix~\ref{app:strategy_selection}) rather than rely on the reported $\gamma^\star$.

\begin{corollary}[$R = 1$ breakdown frontier]
\label{cor:sensitivity_r1}
When $R = 1$, the bound~\eqref{eq:bias_bounds} collapses to
\begin{equation}
\label{eq:bias_bound_r1}
|\mathrm{Bias}_1| \;\leq\; B(\gamma) + o(\gamma), \qquad B(\gamma) := 2\gamma \cdot (1 + |\rho|)\Bigl(1 + \tfrac{|\bar{\tilde{\theta}}|}{|\bar{\tilde{\theta}}_{\mathrm{anc}}|}\Bigr),
\end{equation}
and a tighter contrast-level bound $\tilde{B}(\tilde{\gamma}) = \tilde{\gamma} \cdot (1 + |\rho|)\bigl(1 + |\bar{\tilde{\theta}}|/|\bar{\tilde{\theta}}_{\mathrm{anc}}|\bigr)$ that vanishes under treatment-invariant contamination. The Strategy~2 bound is $|\mathrm{Bias}_2| \leq \gamma(1+|\rho_2|)\bigl(2+|\bar{\tilde{\theta}}|/|\bar{\theta}_{c^*}|\bigr) + o(\gamma)$.
\end{corollary}

Ill-conditioned anchors (large $\kappa(\boldsymbol{\Theta})$) compound non-separability, motivating anchor pairs with linearly independent effect contrasts. The $R = 1$ contrast-level bound $\tilde{B}(\tilde{\gamma})$ is at least twice as tight as $B(\gamma)$ and captures the DGP-C phenomenon of Section~\ref{sec:simulations}: treatment-symmetric violations cancel in the replicated-trial ratio. Both bounds define a \emph{breakdown frontier}: for each contamination level, the identified set is $[\hat{\psi} - B, \hat{\psi} + B]$, and the \emph{breakdown value} $\gamma^*$---the smallest $\gamma$ at which this interval first includes zero---is the temporal analog of Rosenbaum's $\Gamma$ \citep{rosenbaum2002observational}. All quantities other than $\gamma$ are estimable, so practitioners can report a single robustness summary alongside the point estimate (Appendix~\ref{app:sensitivity}).

\subsection{Comparing the Two Strategies and Relationship to Prior Work}
\label{sec:strategy_comparison}

Strategy~1 permits $\boldsymbol{\Lambda}(t_0, t_1)$ to depend on both administration and measurement times, accommodating settings where elapsed time matters (e.g., effect decay), but requires $R$ trials comparing the same non-target treatment pairs at source and target times. Strategy~2 requires only that $R$ treatment arms appear at multiple measurement times---an easier data requirement---but imposes the stronger restriction $\boldsymbol{\Lambda}(t_0, t_1) = \boldsymbol{\Lambda}(t_1)$. When both strategies are feasible, reporting both provides a robustness check: large discrepancies suggest violations of Assumption~\ref{ass:measurement}, and a formal specification test is available when $m \geq 2$ anchor arms exist (Appendix~\ref{app:spec_test}). We give a concrete four-step protocol for strategy selection in Appendix~\ref{app:strategy_selection}.

\textbf{On external baselines.} To the best of our knowledge, this is the first work to formalize temporal transportation as an identification problem. Panel-data methods with related technical ingredients (discussed in Appendix~\ref{app:related}) target a different estimand---the effect where the intervention \emph{was} administered---and so do not constitute honest baselines for the TATE. Section~\ref{sec:simulations} accordingly benchmarks against the Oracle (true temporal ratio) rather than external methods.

\section{Estimation and Inference}
\label{sec:estimation}

\textbf{Doubly robust score.} Observe i.i.d.\ data $\{(Y_i, A_i, S_i, X_i)\}_{i=1}^n$; define nuisance functions $\pi_k(X) := P(S = k \mid X)$, $e_k(a, X) := P(A = a \mid S = k, X)$, and $\mu_{a,k}(X) := \mathbb{E}[Y \mid A = a, S = k, X]$. Under compositional stability, the trial-$k$ marginal mean $\bar{\mu}_{a,k} = \mathbb{E}[\mu_{a,k}(X)]$ admits the doubly robust score
\begin{equation}
\label{eq:score_mean}
\varphi_{a,k}(O; \eta) = \frac{\mathbf{1}[S = k]}{\pi_k(X)} \cdot \frac{\mathbf{1}[A = a]}{e_k(a, X)} \big(Y - \mu_{a,k}(X)\big) + \mu_{a,k}(X),
\end{equation}
with $\mu_{a,k}(X)$ outside the trial-membership weight, delivering Neyman orthogonality in all three nuisance functions (Appendix~\ref{app:dr_properties}). The ATE score is $\varphi_{\tau_k}(O; \eta) = \varphi_{a_k,k}(O; \eta) - \varphi_{b_k,k}(O; \eta)$. Under compositional drift, a density-ratio extension retargets to the trial-$k^\star$ marginal (Appendix~\ref{app:assumption_discussion}).

\textbf{Plug-in TATE estimators.} Given cross-fitted building blocks $\hat{\bar{\mu}}_{a,k} = \mathbb{P}_n[\varphi_{a,k}(O; \hat{\eta})]$ and $\hat{\tau}_k = \mathbb{P}_n[\varphi_{\tau_k}(O; \hat{\eta})]$, the Strategy~1 TATE is a rotation-invariant functional of three observable ATE quantities: $\boldsymbol{\tau}_{k^\star}^{\mathrm{source}} \in \mathbb{R}^R$ stacks the target trial's ATEs at $R$ source times; $\mathbf{T}_{\mathrm{anc}}^{\mathrm{source}} \in \mathbb{R}^{R \times R}$ stacks anchor-pair ATEs across those $R$ source times; $\boldsymbol{\tau}^{\mathrm{target}} \in \mathbb{R}^R$ stacks anchor-pair ATEs at the target time. The plug-in estimator is
\begin{equation}
\label{eq:est_S1_R}
\hat{\psi}_1 = (\hat{\boldsymbol{\tau}}_{k^\star}^{\mathrm{source}})^\top \bigl(\hat{\mathbf{T}}_{\mathrm{anc}}^{\mathrm{source}}\bigr)^{-1} \hat{\boldsymbol{\tau}}^{\mathrm{target}}.
\end{equation}
The latent effect contrasts $\boldsymbol{\Theta}$ and $\mathbf{M}$ are identified only up to a common rotation; only their product $\boldsymbol{\Theta}\mathbf{M}^\top = \mathbf{T}_{\mathrm{anc}}^{\mathrm{source}}$ (observable) appears. Strategy~2 has the same form with $\mathbf{T}_{\mathrm{anc}}^{\mathrm{source}}$ replaced by $\mathbf{T}_{c}^{\mathrm{source}}$ (anchor-arm source-time means) and $\boldsymbol{\tau}^{\mathrm{target}}$ by $\boldsymbol{\mu}^{\mathrm{target}}$. The $R = 1$ scalar specialization used in the experiments is consolidated in Appendix~\ref{app:R1_case}.

Applying the functional delta method to~\eqref{eq:est_S1_R} via $\mathrm{d}(\mathbf{A}^{-1}) = -\mathbf{A}^{-1}(\mathrm{d}\mathbf{A})\mathbf{A}^{-1}$ gives the EIF
\begin{align}
\phi_{\psi_1}(O) &= \bigl(\mathbf{T}_{\mathrm{anc}}^{\mathrm{source}}\bigr)^{-\top}\!\boldsymbol{\tau}^{\mathrm{target}} \cdot \boldsymbol{\phi}_{\boldsymbol{\tau}_{k^\star}^{\mathrm{source}}}(O) + (\boldsymbol{\tau}_{k^\star}^{\mathrm{source}})^\top \bigl(\mathbf{T}_{\mathrm{anc}}^{\mathrm{source}}\bigr)^{-1} \boldsymbol{\phi}_{\boldsymbol{\tau}^{\mathrm{target}}}(O) \notag \\
&\quad - (\boldsymbol{\tau}_{k^\star}^{\mathrm{source}})^\top \bigl(\mathbf{T}_{\mathrm{anc}}^{\mathrm{source}}\bigr)^{-1} \mathrm{d}\mathbf{T}_{\mathrm{anc}}^{\mathrm{source}}[\boldsymbol{\phi}] \bigl(\mathbf{T}_{\mathrm{anc}}^{\mathrm{source}}\bigr)^{-1}\!\boldsymbol{\tau}^{\mathrm{target}},
\label{eq:eif_general_R}
\end{align}
with $\mathrm{d}\mathbf{T}_{\mathrm{anc}}^{\mathrm{source}}[\boldsymbol{\phi}]$ assembled from building-block ATE EIFs (Appendix~\ref{app:proof_eif}).

\begin{theorem}[Asymptotic Properties]
\label{thm:asymptotic}
Under the regularity conditions of Appendix~\ref{app:proof_asymptotic} (overlap, bounded fourth moments, nuisance product-rate $\|\hat{\mu}_{a,k} - \mu_{a,k}\|_2 \cdot (\|\hat{\pi}_k - \pi_k\|_2 + \|\hat{e}_k - e_k\|_2) = o_p(n^{-1/2})$, per-trial sample sizes $n_k/n \to p_k > 0$, and $\sigma_{\min}(\boldsymbol{\Theta})$ bounded below) and the relevant identification assumptions:
(a)~$\sqrt{n}(\hat{\psi} - \psi) \xrightarrow{d} \mathcal{N}(0, V)$ with $V = \mathbb{E}[\phi_{\psi}(O)^2]$ the semiparametric efficiency bound, for $\hat{\psi} \in \{\hat{\psi}_1, \hat{\psi}_2, \hat{\psi}_2^*\}$; (b)~estimators are consistent whenever, \emph{uniformly across all trials $(k, a)$ involved in the TATE functional}, either $\mu_{a,k}$ is correctly specified, or the propensity pair $(\pi_k, e_k)$ is correctly specified; (c)~when all models are correct and the target trial $k^\star$ is distinct from every anchor trial, the estimators achieve the efficiency bound and $\hat{\psi}_2^*$ is minimum-variance among GLS combinations of the $R$ anchor components; if the target trial overlaps with an anchor trial (e.g., its control arm serves as a Strategy~2 anchor), the asymptotic variance acquires cross-trial covariance terms and the minimum-variance GLS weights use the full sandwich covariance derived in Appendix~\ref{app:proof_asymptotic}.
\end{theorem}

The asymptotic variance inflates by a factor of $\kappa(\boldsymbol{\Theta})^2$ when the anchor effect contrasts are near-collinear, so ill-conditioned anchors degrade precision even when the $\sqrt{n}$ rate is achieved. In randomized trials, $e_k(a, X)$ is known by design, so the double robustness of part (b) is automatic. Any nuisance estimator with $n^{-1/4}$ $L_2$ error suffices to meet the product-rate condition.

\textbf{Why Strategy 2 is typically more efficient.} When target and anchor trials are distinct, the influence-function components have disjoint support and asymptotic variances decompose accordingly (scaled by $\kappa(\boldsymbol{\Theta})^2$ for general $R$; Appendix~\ref{app:proof_asymptotic}). Two compounding factors favor Strategy~2: under randomization, contrasts carry roughly twice the variance of means ($V_{\tau_k} \approx V_{\bar{\mu}_{a_k,k}} + V_{\bar{\mu}_{b_k,k}}$), and treatment effects are typically smaller than outcome means ($|\tau_k| \ll |\bar{\mu}_{a,k}|$), so the contrast coefficient of variation is much larger. The efficiency gain comes at the cost of Assumption~\ref{ass:measurement}; when that fails, Strategy~2 is biased while Strategy~1 remains valid.

\textbf{Multi-anchor combination and specification test.} When $m > R$ anchor arms satisfy Assumption~\ref{ass:common}, the system $\boldsymbol{\mu}^{\mathrm{target}} = \boldsymbol{\Theta}_c \boldsymbol{\Lambda}^{\mathrm{target}}$ is overdetermined. The efficient combination is the generalized least squares estimator $\hat{\boldsymbol{\Lambda}}^{\mathrm{target}} = (\boldsymbol{\Theta}_c^\top \mathbf{V}_{\boldsymbol{\mu}}^{-1} \boldsymbol{\Theta}_c)^{-1} \boldsymbol{\Theta}_c^\top \mathbf{V}_{\boldsymbol{\mu}}^{-1} \hat{\boldsymbol{\mu}}^{\mathrm{target}}$, which minimizes asymptotic variance among linear combinations of the overidentified system (Appendix~\ref{app:proof_asymptotic}). The same overdetermined system supports a specification test for Assumption~\ref{ass:measurement}: the GLS residual $\hat{\mathbf{r}} = \hat{\boldsymbol{\mu}}^{\mathrm{target}} - \boldsymbol{\Theta}_c \hat{\boldsymbol{\Lambda}}^{\mathrm{target}}$ is asymptotically mean-zero under $H_0$, and the Wald statistic $Q = n\,\hat{\mathbf{r}}^\top \hat{\mathbf{V}}_{\boldsymbol{\mu}}^{-1}\hat{\mathbf{r}}$ has a $\chi^2_{m - R}$ null distribution (Appendix~\ref{app:spec_test}).

\section{Simulation Study}
\label{sec:simulations}

We conduct Monte Carlo simulations to investigate finite-sample performance of the TEA-Time estimators. The study is organized around three data generating processes (DGPs) that progressively stress-test the framework's assumptions: DGP~A satisfies all assumptions and serves as a baseline; DGP~B violates the measurement-time restriction (Assumption~\ref{ass:measurement}); and DGP~C violates separability itself (Assumption~\ref{ass:separable}).

\textbf{Design.} All three DGPs share covariates $X = (X_1, X_2)^\top$ with $X_1 \sim \mathcal{N}(0, 1)$, $X_2 \sim \text{Bernoulli}(0.5)$, three treatments $a \in \{0, 1, 2\}$ with linear-in-$X$ effects $\theta_a(X)$, seasonal temporal effect $\Lambda(t) = 1 + 0.3 \sin(2\pi t/12)$, Gaussian noise, and $K = 6$ trials with the target trial $k^\star = 1$ comparing treatment~1 against control; the TATE target is the same comparison at measurement time $t_1 = 9$ (full $\theta_a$ expressions and trial tables in Appendix~\ref{app:sim_design}). Throughout, we set $R = 1$, so a single anchor pair (Strategy~1) or a single anchor arm (Strategy~2) suffices. The DGPs differ only in how the temporal effect enters potential outcomes: \textbf{DGP~A} is fully separable with measurement-time-only structure ($\Lambda(t_0, t_1) = \Lambda(t_1)$), so both strategies are correctly specified; \textbf{DGP~B} adds lag-dependent decay $\Lambda(t_0, t_1) = e^{-\beta(t_1 - t_0)}\Lambda(t_1)$ with $\beta = 0.1$ while preserving separability, violating Assumption~\ref{ass:measurement} and biasing Strategy~2 only; \textbf{DGP~C} adds a treatment-specific contamination $\delta_c \sin(2\pi t/12 + \varphi_a)$ with treatment-specific phases, violating separability itself (Assumption~\ref{ass:separable}). We compare five estimators---\textbf{Oracle} (true ratio), \textbf{S1}, \textbf{S2-C} (control anchor), \textbf{S2-T} (treatment-2 anchor), and \textbf{S2-M} (inverse-variance combination)---using gradient-boosted nuisance models with 5-fold cross-fitting, $B = 500$ Monte Carlo replications, and $n \in \{600, 1200, 2400, 4800\}$.

\textbf{Results.} Table~\ref{tab:sim_results} reveals a clear hierarchy. Under \emph{DGP~A}, all estimators are approximately unbiased with near-nominal coverage, and the Strategy~2 variants achieve roughly $2\times$ lower RMSE than Strategy~1 (at $n = 2400$, $0.07$ vs.\ $0.15$)---the efficiency gain predicted by Section~\ref{sec:estimation}. Under \emph{DGP~B}, Strategy~1 remains valid (bias shrinks with $n$, coverage near 95\%), but all three Strategy~2 variants develop a persistent bias of $\approx\!-0.05$ that does not decay; as their CIs tighten around this wrong target, coverage deteriorates from near $90\%$ at $n = 600$ to $80\%$--$83\%$ at $n = 4800$---the signature of biased estimators with correctly sized intervals. Under \emph{DGP~C}, Strategy~2's control-arm ($c^* = 0$) and multi-anchor combinations fail catastrophically (bias $\approx -0.66$, coverage $0\%$), while the treatment-2-anchor variant ($c^* = 2$) is less biased ($\approx -0.24$) but its coverage still collapses from $78\%$ to $6\%$ as $n$ grows. Strategy~1 survives at every $n$ (bias $\leq +0.06$, coverage $\geq 92\%$): the treatment-specific contamination cancels in the replicated-trial ratio because both target and anchor compare the same pair. When the separability violation is treatment-specific, Strategy~1 absorbs it through symmetric contamination; when it is covariate- or trial-specific, neither strategy is consistent, and the breakdown frontier of Proposition~\ref{prop:sensitivity} quantifies the worst-case bias.

\textbf{Degradation under violations.} Figure~\ref{fig:diagnostics} sweeps each of the framework's four principal stressors one at a time at $n = 1800$. \emph{(a)}~Lag-decay rate $\beta$ (Assumption~\ref{ass:measurement} violated at $\beta > 0$): Strategy~1's bias stays below Strategy~2's at every $\beta$. \emph{(b)}~Contamination radius $\gamma$ (Assumption~\ref{ass:separable} violated): Strategy~1's empirical bias stays well below the Proposition~\ref{prop:sensitivity} bound $B(\gamma)$ for both contamination types; the covariate-time curve is near zero because $\mathbb{E}[\Gamma_a - \Gamma_b] = 0$ cancels the contrast. \emph{(c)}~Rank-$2$ misspecification $\delta$: both strategies stay well below the no-temporal-adjustment baseline, so a rank-$1$ correction buys most of the value even when the truth is higher-rank. \emph{(d)}~Target-time covariate drift $\delta_{\mathrm{drift}}$: the DR estimator keeps coverage near $95\%$ while the unadjusted estimator falls into the failure region, confirming Theorem~\ref{thm:asymptotic}(b). The specification test of Appendix~\ref{app:spec_test} complements panel~(a) and is documented in Appendix~\ref{app:simulations_extra} along with double-robustness and $\Lambda$-shape invariance studies.

\section{Empirical Application: Upworthy Headline Tests}
\label{sec:casestudy}

\begin{table*}[!htbp]
\centering
\caption{Simulation results across three DGPs ($B = 500$ Monte Carlo replications per cell). Each entry reports Bias / RMSE / Coverage of the 95\% CI. Bold marks the lower-RMSE strategy among non-Oracle; \colorbox{red!10}{red shading} marks coverage below 90\% (the failure regime). Oracle is undefined under DGP~C (separability itself is violated). The red shading is the paper's central empirical finding: Strategy~2 fails silently under assumption violations, with \emph{worsening} coverage at larger $n$ because its tighter intervals center on a biased estimand.}
\label{tab:sim_results}
\small
\begin{tabular}{@{}cl ccc ccc ccc @{}}
\toprule
& & \multicolumn{3}{c}{\textbf{DGP~A}: assumptions hold} & \multicolumn{3}{c}{\textbf{DGP~B}: lag-dependent $\Lambda$} & \multicolumn{3}{c}{\textbf{DGP~C}: trt--time interaction} \\
\cmidrule(lr){3-5} \cmidrule(lr){6-8} \cmidrule(lr){9-11}
$n$ & Est. & Bias & RMSE & Cov. & Bias & RMSE & Cov. & Bias & RMSE & Cov. \\
\midrule
\multirow{5}{*}{600}
    & Oracle & $-.00$ & .13 & .98 & $-.00$ & .11 & .95 & --- & --- & --- \\
    & S1     & $+.03$ & .34 & .93 & $+.03$ & .32 & .92 & $+.06$ & \textbf{.46} & .92 \\
    & S2-C   & $-.00$ & .16 & .96 & $-.05$ & .14 & .91 & $-.66$ & .67 & \colorbox{red!10}{.00} \\
    & S2-T   & $+.00$ & .16 & .98 & $-.04$ & .13 & .90 & $-.24$ & \textbf{.30} & \colorbox{red!10}{.78} \\
    & S2-M   & $-.02$ & \textbf{.15} & .97 & $-.06$ & \textbf{.13} & .90 & $-.64$ & .65 & \colorbox{red!10}{.01} \\
\midrule
\multirow{5}{*}{2400}
    & Oracle & $+.00$ & .06 & .97 & $+.00$ & .05 & .97 & --- & --- & --- \\
    & S1     & $+.01$ & .15 & .97 & $+.01$ & .13 & .97 & $+.02$ & \textbf{.19} & .97 \\
    & S2-C   & $+.00$ & .08 & .99 & $-.05$ & .07 & .90 & $-.66$ & .66 & \colorbox{red!10}{.00} \\
    & S2-T   & $-.00$ & \textbf{.07} & .98 & $-.05$ & .08 & \colorbox{red!10}{.89} & $-.25$ & .26 & \colorbox{red!10}{.34} \\
    & S2-M   & $-.00$ & \textbf{.07} & .99 & $-.05$ & \textbf{.07} & \colorbox{red!10}{.88} & $-.66$ & .66 & \colorbox{red!10}{.00} \\
\midrule
\multirow{5}{*}{4800}
    & Oracle & $+.00$ & .04 & .97 & $+.00$ & .03 & .96 & --- & --- & --- \\
    & S1     & $+.01$ & .09 & .97 & $+.01$ & .09 & .96 & $+.01$ & \textbf{.12} & .97 \\
    & S2-C   & $+.00$ & \textbf{.05} & .99 & $-.05$ & .06 & \colorbox{red!10}{.83} & $-.66$ & .66 & \colorbox{red!10}{.00} \\
    & S2-T   & $+.00$ & \textbf{.05} & .99 & $-.05$ & \textbf{.06} & \colorbox{red!10}{.82} & $-.24$ & .25 & \colorbox{red!10}{.06} \\
    & S2-M   & $-.00$ & \textbf{.05} & 1.0 & $-.05$ & \textbf{.06} & \colorbox{red!10}{.80} & $-.66$ & .67 & \colorbox{red!10}{.00} \\
\bottomrule
\end{tabular}
\end{table*}
\begin{figure}[!htbp]
\centering
\includegraphics[width=\linewidth]{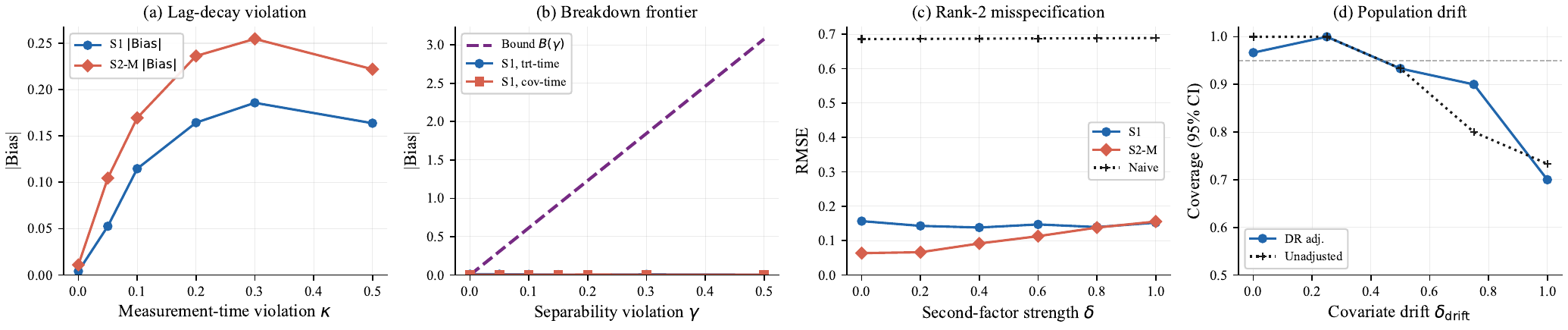}
\caption{Degradation under four assumption violations ($n = 1800$, $B = 60$). \textbf{(a)}~$|\mathrm{Bias}|$ vs.\ lag-decay rate $\beta$ (violates Assumption~\ref{ass:measurement}). \textbf{(b)}~Strategy~1 $|\mathrm{Bias}|$ vs.\ contamination radius $\gamma$ (violates Assumption~\ref{ass:separable}) with the Proposition~\ref{prop:sensitivity} bound $B(\gamma)$ as upper envelope. \textbf{(c)}~RMSE vs.\ second-factor strength $\delta$ (rank misspecification) vs.\ no-adjustment baseline. \textbf{(d)}~$95\%$ CI coverage vs.\ target-time covariate drift $\delta_{\mathrm{drift}}$.}
\label{fig:diagnostics}
\end{figure}

We apply our framework to headline A/B tests from the Upworthy Research Archive \citep{matias2021upworthy}, one of the largest publicly available collections of randomized digital experiments. Full details appear in Appendix~\ref{app:upworthy}.

\textbf{Data and setup.} The Upworthy Research Archive \citep{matias2021upworthy} contains 32{,}487 headline A/B tests from January 2013 to April 2015; we analyze the 105{,}551-arm confirmatory subset. Following \citep{matias2024upworthy_update}, we exclude tests from June 25, 2013 to January 10, 2014, a window with cache-driven block-like assignment. Because each test uses unique headline text, we cluster semantically similar headlines using Sentence-BERT embeddings \citep{reimers2019sentence} via constrained Hungarian matching into 50 clusters (Appendix~\ref{app:upworthy_clustering}).

\textbf{Research question.} We estimate TATEs for two A/B tests (Trial~A: February 2014, Trial~B: January 2014), transporting observed effects to the remaining months of 2014. For each target month, we compare both strategies against a ``ground-truth'' TATE pooled from all available comparisons of the target cluster pair at that month.

\textbf{Results.} Figure~\ref{fig:tate_comparison} and Table~\ref{tab:summary_stats} show a variance-bias tradeoff that manifests differently across the two trials. Strategy~2 yields $3\times$ tighter intervals in both trials, as predicted. On \emph{Trial~A}, Strategy~1 tracks the per-month ground truth closely ($+0.97$ correlation, RMSE $0.0027$) while Strategy~2 sits near a constant and anti-tracks ($-0.12$). On \emph{Trial~B}, neither strategy recovers the ground truth cleanly: Strategy~2 is again near-flat ($+0.26$ correlation, lowest RMSE) while Strategy~1 is more variable but \emph{anti}-correlates with the ground truth ($-0.87$), indicating that for this cluster pair the temporal ratio estimated from the target pair's own history is not a useful signal for the true dynamics. Both diagnostics are informative: they tell a practitioner that the transportation assumptions do not cleanly hold on Trial~B. The breakdown frontier quantifies this honestly: Trial~A Strategy~1 has $\gamma^\star \!\approx\! 0.00024$ ($\sim 2\%$ of the cluster-10 mean control CTR); Strategy~2's numerically larger $\gamma^\star \!\approx\! 0.00142$ reflects its tighter intervals rather than genuine robustness; Trial~B Strategy~1's observed effect is already within $1.96\!\cdot\!\mathrm{SE}$ of zero so $\gamma^\star = 0$, and Strategy~2 has $\gamma^\star \!\approx\! 0.00083$. For calibration, the median month-over-month $|\Delta \mathrm{CTR}|$ across Upworthy control clusters in our window is $0.00223$ ($\approx 15\%$ of cluster mean)---several times larger than any reported $\gamma^\star$. The breakdown frontier thus flags the TATE estimates as fragile to violations comparable to the empirical monthly drift the archive exhibits. Results are robust to clustering configuration ($K \in \{30, 50, 75, 100\}$, alternative embedding; Appendix~\ref{app:clustering_sensitivity}). The framework's intended use is exactly this workflow: compute both strategies, inspect agreement, read off $\gamma^\star$, and decide whether to deploy the finding.

\begin{table}[!htbp]
\centering
\caption{Estimation performance against the ground-truth TATE (per-month pooled across the target cluster pair). ``Corr'' is the Pearson correlation across target months; source months excluded.}
\label{tab:summary_stats}
\small
\begin{tabular}{@{}lcccc@{}}
\toprule
& \multicolumn{2}{c}{Trial A ($n=11$ months)} & \multicolumn{2}{c}{Trial B ($n=12$ months)} \\
\cmidrule(lr){2-3} \cmidrule(lr){4-5}
& S1 & S2 & S1 & S2 \\
\midrule
RMSE         & \textbf{.0027}  & .0105           & .0134           & \textbf{.0081} \\
Bias         & $+.0002$        & $+.0074$        & $+.0040$        & $-.0014$ \\
Avg Std Err  & .0054           & \textbf{.0019}  & .0052           & \textbf{.0018} \\
Corr         & $\mathbf{+0.97}$ & $-0.12$        & $-0.87$         & $+0.26$ \\
\bottomrule
\end{tabular}
\end{table}

\begin{figure*}[!htbp]
    \centering
    \includegraphics[width=0.48\linewidth]{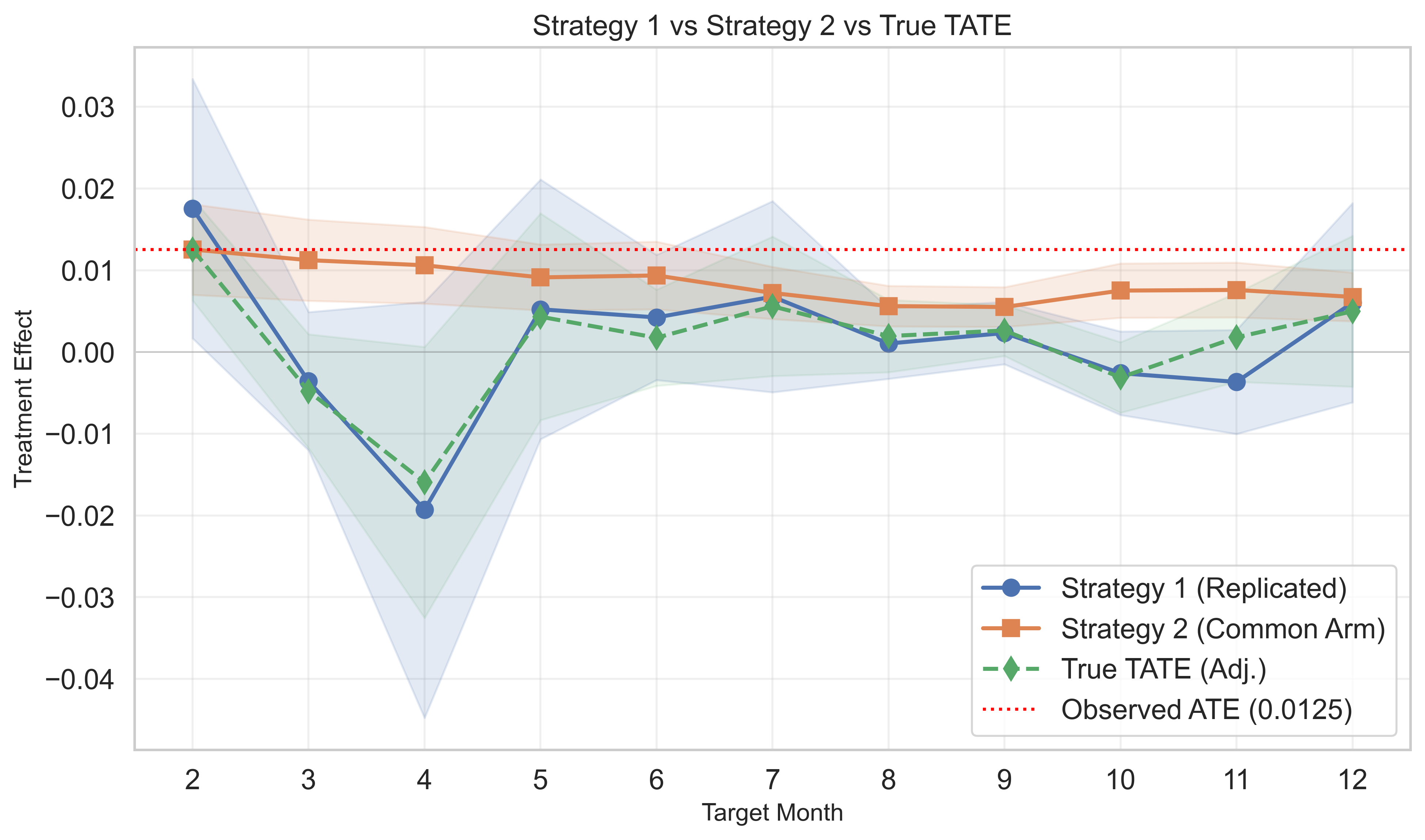}
    \hfill
    \includegraphics[width=0.48\linewidth]{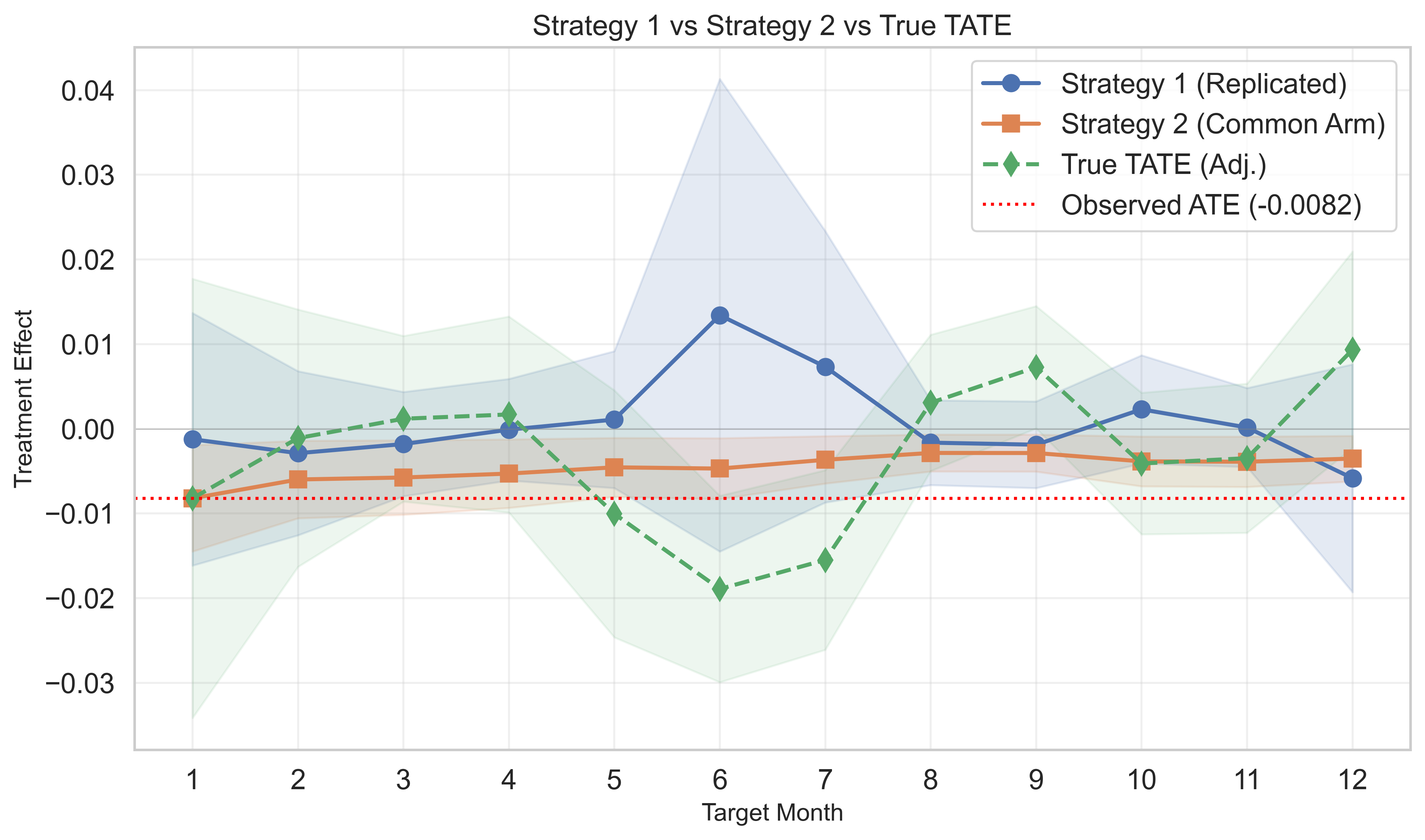}
    \caption{TATE estimates by target month for both trials. Each panel shows Strategy~1 (blue circles), Strategy~2 (orange squares), and true TATE (green diamonds) with 95\% confidence bands; dashed red line marks the observed ATE at the source time. \textbf{(a)}~Trial~A (Cluster~17 vs.\ 10), source time February 2014. \textbf{(b)}~Trial~B (Cluster~46 vs.\ 12), source time January 2014.}
    \label{fig:tate_comparison}
\end{figure*}

\section{Discussion and Conclusion}\label{sec:conclusion}

We introduced the transported average treatment effect (TATE) and gave two identification strategies---replicated trials and common arm---with doubly robust, semiparametrically efficient estimators and a first-order breakdown frontier for separability violations. The simulation study showed the common-arm strategy's efficiency gains are real but fail when the measurement-time structure fails, while replicated trials remains valid under measurement-time violations and robust to treatment-symmetric contamination.

\textbf{Scope and limitations.} Separability is structural and untestable at a single pair of time points; the specification test, Proposition~\ref{prop:sensitivity}, and the multi-anchor Strategy-1 disagreement diagnostic provide only partial safeguards, and the framework interpolates rather than extrapolates---anchor trials at the target time are required. The rank $R$ is treated as fixed and known, and time as discrete; data-driven rank selection and continuous-time extensions are natural next steps. The framework is most credible when a common temporal driver (seasonality, cohort composition, macroeconomic conditions) governs outcome levels across heterogeneous treatments, and is not appropriate when treatment--time interactions dominate or when the true temporal rank exceeds the available anchors. On the deployment side, the variance-bias tradeoff carries a risk that practitioners drawn to Strategy~2's tighter CIs may silently prefer a biased estimator when its extra measurement-time assumption fails; reporting both estimates together with $\gamma^\star$ is the intended mitigation.

\bibliographystyle{plainnat}
\bibliography{reference}

\newpage
\appendix

\section{Assumption Discussion}
\label{app:assumption_discussion}

This appendix collects supporting material on Assumption~\ref{ass:separable} cut from the main text, and makes explicit the two identification assumptions (Replicated Trials and Common Arm) used in Corollaries~\ref{cor:replicated}--\ref{cor:common_arm}.

\paragraph{Rank and treatment contrasts.}
The rank $R$ in Assumption~\ref{ass:separable} governs the tradeoff between model flexibility and data requirements: larger $R$ accommodates richer treatment-time interactions but requires more anchor trials; $R$ anchor pairs (or arms) with linearly independent effect contrasts suffice (Corollaries~\ref{cor:replicated}--\ref{cor:common_arm}). Under Assumption~\ref{ass:separable}, the individual treatment effect takes the contrast form
\begin{equation*}
Y_{t_1}(a,t_0) - Y_{t_1}(b,t_0) = \tilde{\boldsymbol{\theta}}_{a,b}(X)^\top \boldsymbol{\Lambda}(t_0,t_1), \qquad \tilde{\boldsymbol{\theta}}_{a,b}(X) := \boldsymbol{\theta}_a(X) - \boldsymbol{\theta}_b(X) \in \mathbb{R}^R,
\end{equation*}
so the TATE decomposition (Theorem~\ref{thm:decomposition}) follows because target and source-time ATEs share the same mean contrast and differ only through the temporal effect.

\paragraph{Rank selection in practice.}
When $R$ is not known a priori, a natural data-driven approach is to inspect the spectral structure of the observable anchor-ATE matrix $\mathbf{T}_{\mathrm{anc}}^{\mathrm{source}} \in \mathbb{R}^{m \times T}$ (anchor pairs $\times$ observed source times): the number of singular values clearly separated from the noise floor is a scale-free estimate of $R$. Cross-validation on held-out anchor trials provides an alternative when sufficient anchors are available. In our experiments and Upworthy application we take $R = 1$ throughout; rank-misspecification robustness is checked in Appendix~\ref{app:simulations_extra} (rank-2 DGP, estimator with $R=1$), where both strategies degrade gracefully.

\paragraph{Regularity and rates for Theorem~\ref{thm:asymptotic}.}
The asymptotic results in the main text rely on standard semiparametric regularity conditions:
\begin{assumption}[Regularity]
\label{ass:regularity}
(i) \emph{Overlap}: $\pi_k(X) > \epsilon$ and $e_k(a, X) > \epsilon$ a.s.\ for some $\epsilon > 0$.
(ii) \emph{Bounded moments}: $\mathbb{E}[Y^4] < \infty$.
(iii) \emph{Non-degeneracy and conditioning}: for $R = 1$, $\tau_{j'} \neq 0$ (Strategy 1) or $\bar{\mu}_{c^*,\ell'} \neq 0$ (Strategy 2). For general $R$, the effect-contrast matrix $\boldsymbol{\Theta}$ (resp.\ $\boldsymbol{\Theta}_c$) has smallest singular value bounded below: $\sigma_{\min}(\boldsymbol{\Theta}) \geq c > 0$ for some constant $c$. Equivalently, the condition number $\kappa(\boldsymbol{\Theta}) = \sigma_{\max}(\boldsymbol{\Theta})/\sigma_{\min}(\boldsymbol{\Theta})$ is bounded. Near-collinear anchors inflate the asymptotic variance by a factor of $\kappa(\boldsymbol{\Theta})^2$ even when the $n^{-1/2}$ rate is achieved.
\end{assumption}
\begin{assumption}[Nuisance Estimation Rates]
\label{ass:rates}
$\|\hat{\mu}_{a,k} - \mu_{a,k}\|_2 \cdot (\|\hat{\pi}_k - \pi_k\|_2 + \|\hat{e}_k - e_k\|_2) = o_p(n^{-1/2})$. This product-rate condition is weaker than individual $n^{-1/2}$ rates and is satisfied when each nuisance converges at $n^{-1/4}$, a rate achievable by many machine learning methods with appropriate regularization.
\end{assumption}

\begin{assumption}[Replicated Trials]
\label{ass:replicated}
(i) Each trial draws an independent random sample from $\mathcal{P}$, with $Y(a) \perp S \mid X$. (ii) There exist $R$ anchor pairs $(a_r^*, b_r^*)$, $r = 1, \ldots, R$, each with $\bar{\tilde{\boldsymbol{\theta}}}_{a_r^*, b_r^*} := \mathbb{E}[\tilde{\boldsymbol{\theta}}_{a_r^*, b_r^*}(X)] \neq \mathbf{0}$ and $\boldsymbol{\Lambda}(t_0, t_1) \neq \mathbf{0}$ for all relevant timing pairs, such that the $R \times R$ matrix $\boldsymbol{\Theta} := [\bar{\tilde{\boldsymbol{\theta}}}_{a_1^*, b_1^*}, \ldots, \bar{\tilde{\boldsymbol{\theta}}}_{a_R^*, b_R^*}]^\top$ has full rank. (iii) Both source and target timing belong to the observed timing sets for each anchor pair.
\end{assumption}

\begin{assumption}[Common Arm]
\label{ass:common}
(i) Each trial draws an independent random sample from $\mathcal{P}$. (ii) There exist $R$ anchor arms $c_1^*, \ldots, c_R^* \in \mathcal{A} \cup \{0\}$ with $\bar{\boldsymbol{\theta}}_{c_r^*} := \mathbb{E}[\boldsymbol{\theta}_{c_r^*}(X)] \neq \mathbf{0}$, such that the $R \times R$ matrix $\boldsymbol{\Theta}_c := [\bar{\boldsymbol{\theta}}_{c_1^*}, \ldots, \bar{\boldsymbol{\theta}}_{c_R^*}]^\top$ has full rank. (iii) The source and target measurement times lie in the observed measurement-time set for each anchor arm.
\end{assumption}

The random sampling condition ensures that trial-specific expectations coincide with population expectations. When $R = 1$, both assumptions reduce to requiring a single anchor pair (resp.\ arm) with nonzero mean contrast (resp.\ mean outcome effect). Assumption~\ref{ass:common}(ii) is typically weaker than Assumption~\ref{ass:replicated}(ii) when outcomes are naturally positive, since nonzero mean outcome effects are easier to guarantee than nonzero mean treatment contrasts. Identification proofs under these assumptions are in Appendix~\ref{app:proofs}.

\paragraph{Population drift and non-stationarity.}
Two distinct violations of population stability merit separate treatment. \emph{Compositional drift}---trial-specific covariate distributions $P(X \mid S = k)$ differ across trials, so $X \not\perp S$ at the super-population level---breaks the equivalence between the pooled mean $\mathbb{E}[\mu_{a,k}(X)]$ (the target of score~\eqref{eq:score_mean} under stability) and the trial-$k^\star$ marginal $\mathbb{E}[\mu_{a,k}(X) \mid S = k^\star]$ (the TATE-relevant quantity). To target the trial-$k^\star$ marginal under drift, the score must be re-weighted by the density ratio $\omega(X) := P(S = k^\star \mid X)/P(S = k^\star)$, giving
\begin{align*}
\varphi_{a,k}^{\mathrm{drift}}(O; \eta, \omega) = \omega(X)\left[\frac{\mathbf{1}[S = k]\mathbf{1}[A = a]}{\pi_k(X) e_k(a, X)} (Y - \mu_{a,k}(X)) + \mu_{a,k}(X)\right],
\end{align*}
which targets $\mathbb{E}[\mu_{a,k}(X) \mid S = k^\star]$ consistently (the additional nuisance $\omega$ inherits the usual positivity and product-rate conditions). This is the analog of covariate-shift transportability \citep{dahabreh2019generalizing, stuart2011use}. The population-drift simulation in Appendix~\ref{app:simulations_extra} confirms that without the density-ratio reweighting, covariate drift causes systematic bias whereas the drift-corrected score recovers near-nominal coverage. \emph{Individual non-stationarity}---$\boldsymbol{\theta}_a(X)$ itself changes over time for the same unit---violates the separable-effects structure at the unit level and cannot be resolved by any reweighting. When multiple anchor arms are available, the specification test of Appendix~\ref{app:spec_test} can detect certain non-stationarity patterns: if compositional changes differentially affect anchor arms, the single-anchor ratios $R_{2,j}$ diverge and the test rejects.

\paragraph{Trial-timing exogeneity.}
\label{app:trial_timing}
Identification also requires that the trial scheduling process be independent of the potential outcomes at the population level---the temporal analog of the ignorability assumption in observational causal inference. Formally, let $T_k := (t_{0k}, t_{1k})$ be the random pair encoding when trial $k$ is run, viewed as a function of the scheduling process (editorial calendar, protocol timeline, platform cadence, etc.). Trial-timing exogeneity requires $T_k \perp \{Y_{t_1}(a, t_0)\}_{a, t_0, t_1}$ jointly across the super-population; conditioning on trial membership is not the content of the assumption (conditioning on $S = k$ merely fixes the realization of $T_k$). This is automatically satisfied when trials are scheduled on a predetermined cadence, which is the norm in industrial A/B-testing platforms (daily/weekly/monthly launches) and in phased clinical development (fixed protocol timing). It can fail when experimenters or editors select trial timing based on prior beliefs about effect magnitude---for example, scheduling high-variance headline tests during high-traffic periods to maximize learning, or running drug trials during seasons when disease prevalence is expected to be high. In such settings, the observed ATEs at different times mix temporal scaling with selection-on-expected-effect, and our identification formulas do not recover the TATE. Where available, including temporal-context covariates $Z_t$ (day-of-week indicators, macroeconomic state, traffic volume) in the adjustment set can partially mitigate selection on expected effect. For the Upworthy analysis (Appendix~\ref{app:upworthy}), we follow the editorial-calendar documentation of the archive: tests appear to have been scheduled based on editorial publishing pipelines rather than expected-effect selection, consistent with exogeneity.

\section{Strategy Selection Protocol}
\label{app:strategy_selection}

When both strategies are feasible, we recommend the following four-step protocol:
\begin{enumerate}[leftmargin=*, itemsep=2pt]
    \item Compute both $\hat{\psi}_1$ (Strategy~1) and $\hat{\psi}_2$ or $\hat{\psi}_2^*$ (Strategy~2, optionally with multiple anchors).
    \item If $m \geq 2$ anchor arms are available, run the specification test of Appendix~\ref{app:spec_test} for Assumption~\ref{ass:measurement}.
    \item If the test rejects (or a substantive prior on lag-dependent dynamics suggests it should), prefer $\hat{\psi}_1$: it accommodates $\boldsymbol{\Lambda}(t_0, t_1)$ with arbitrary dependence on both administration and measurement times.
    \item If the test does not reject and the measurement lag $t_1 - t_0$ is fixed by design (the norm in digital experimentation, where measurement windows are standardized), prefer $\hat{\psi}_2$ for its $2$--$3\times$ efficiency advantage.
\end{enumerate}
In all cases, report both estimates: substantial disagreement is itself a diagnostic. When the specification test is inconclusive and both estimates are similar, an inverse-variance-weighted shrinkage combination $\hat{\psi}_{\mathrm{shr}} = w \hat{\psi}_1 + (1 - w)\hat{\psi}_2$ with $w = V_{\hat{\psi}_2}/(V_{\hat{\psi}_1} + V_{\hat{\psi}_2})$ provides a finite-sample compromise, though it loses the robustness guarantees of $\hat{\psi}_1$ when Assumption~\ref{ass:measurement} fails. We do not pursue shrinkage further in this paper.

\paragraph{Anchor selection and conditioning.}
The asymptotic variance of both strategies and the partial-identification bound of Proposition~\ref{prop:sensitivity} scale with the condition number $\kappa(\boldsymbol{\Theta}) = \sigma_{\max}(\boldsymbol{\Theta})/\sigma_{\min}(\boldsymbol{\Theta})$ (or $\kappa(\boldsymbol{\Theta}_c)$ for Strategy~2). Ill-conditioned anchors simultaneously inflate $V$ by a factor of $\kappa^2$ (Theorem~\ref{thm:asymptotic}) and the first-order bias bound $B_R(\gamma)$ by a factor of $\kappa \cdot (1 + \|\bar{\tilde{\boldsymbol{\theta}}}_{k^\star}\|/\sigma_{\min}(\boldsymbol{\Theta}))$ (Proposition~\ref{prop:sensitivity}). In the limit $\sigma_{\min}(\boldsymbol{\Theta}) \to 0$ both diverge---consistent with the underlying population estimand itself becoming unbounded because Corollary~\ref{cor:replicated} requires $\boldsymbol{\Theta}^{-1}$. When multiple candidate anchor sets are available we therefore recommend: (i) estimate $\hat{\boldsymbol{\Theta}}$ on source-time data using the DR score of Equation~\eqref{eq:score_mean} applied to every candidate anchor pair, (ii) prefer the anchor set that maximizes $\sigma_{\min}(\hat{\boldsymbol{\Theta}})$ (equivalently, minimizes $\kappa(\hat{\boldsymbol{\Theta}})$) among sets with sufficient sample size per arm, and (iii) report $\kappa(\hat{\boldsymbol{\Theta}})$ alongside $\gamma^\star$ as a transparent precision-robustness summary. When no candidate set achieves a well-conditioned $\hat{\boldsymbol{\Theta}}$ (e.g., all candidate anchor pairs yield $\kappa(\hat{\boldsymbol{\Theta}}) > 10$ for $R = 2$ or beyond), the data do not support rank-$R$ identification with the available anchors, and the practitioner should lower $R$ or report an honestly wide identified set rather than a point estimate whose precision is artefactual.

\section{$R = 1$ Specialization: A Consolidated Reference}
\label{app:R1_case}

The main text states every result for general $R$. In practice, $R = 1$ is the dominant regime: outcomes have a single temporal effect, and all matrices collapse to scalars. This appendix collects the scalar specializations of the paper's main results in one place so practitioners implementing the $R = 1$ case can read off the estimators directly.

\paragraph{Notation for this appendix.}
When $R = 1$, $\boldsymbol{\Lambda}(t_0, t_1) \in \mathbb{R}$, $\boldsymbol{\Theta} = \bar{\tilde{\theta}}_{\mathrm{anc}} \in \mathbb{R}$ (the mean anchor-pair contrast), $\boldsymbol{\Theta}_c = \bar{\theta}_{c^*} \in \mathbb{R}$ (the mean anchor-arm effect), and $\mathbf{M} = \Lambda^{\mathrm{source}} \in \mathbb{R}$. Let $j, j'$ index the anchor pair at source and target times for Strategy~1, and let $\ell, \ell'$ index source and target measurement times for Strategy~2.

\paragraph{TATE decomposition (Theorem~\ref{thm:decomposition}).}
The decomposition~\eqref{eq:tate_decomp} becomes a scalar product of the observed source-time ATE and a temporal ratio:
\begin{equation}
\label{eq:tate_decomp_r1}
\tau_{k^\star}(\delta_0, \delta_1) = \tau_{k^\star}(0, 0) \cdot \frac{\Lambda(t_{0k^\star} + \delta_0, t_{1k^\star} + \delta_1)}{\Lambda(t_{0k^\star}, t_{1k^\star})}.
\end{equation}

\paragraph{Strategy 1 identification (Corollary~\ref{cor:replicated}).}
The matrix identity $\boldsymbol{\Lambda}^{\mathrm{target}} = \boldsymbol{\Theta}^{-1}\boldsymbol{\tau}^{\mathrm{target}}$ becomes a scalar ratio of ATEs from the replicated anchor pair:
\begin{equation}
\label{eq:S1_ratio_r1}
\frac{\Lambda^{\mathrm{target}}}{\Lambda^{\mathrm{source}}} = \frac{\tau_j}{\tau_{j'}},
\end{equation}
where $\tau_{j'}$ is the anchor-pair ATE at source time and $\tau_j$ at target time. The TATE estimand is $\psi_1 = \tau_{k^\star} \cdot \tau_j / \tau_{j'}$.

\paragraph{Strategy 2 identification (Corollary~\ref{cor:common_arm}).}
The temporal ratio is identified from the anchor-arm marginal means at the two measurement times:
\begin{equation}
\label{eq:S2_ratio_r1}
\frac{\Lambda(t_{1k^\star} + \delta_1)}{\Lambda(t_{1k^\star})} = \frac{\bar{\mu}_{c^*, \ell}}{\bar{\mu}_{c^*, \ell'}},
\end{equation}
and $\psi_2 = \tau_{k^\star} \cdot \bar{\mu}_{c^*,\ell}/\bar{\mu}_{c^*,\ell'}$.

\paragraph{Efficient influence functions (Proposition~\ref{prop:eif}).}
The general-$R$ EIF~\eqref{eq:eif_general} collapses to a three-term scalar expression for each strategy. Writing $R_1 := \tau_j/\tau_{j'}$ and $R_2 := \bar{\mu}_{c^*,\ell}/\bar{\mu}_{c^*,\ell'}$,
\begin{align}
\phi_{\psi_1}(O) &= R_1 \phi_{\tau_{k^\star}}(O) + \frac{\psi_1}{\tau_j} \phi_{\tau_j}(O) - \frac{\psi_1}{\tau_{j'}} \phi_{\tau_{j'}}(O), \label{eq:eif_psi1} \\
\phi_{\psi_2}(O) &= R_2 \phi_{\tau_{k^\star}}(O) + \frac{\psi_2}{\bar{\mu}_{c^*,\ell}} \phi_{\bar{\mu}_{c^*,\ell}}(O) - \frac{\psi_2}{\bar{\mu}_{c^*,\ell'}} \phi_{\bar{\mu}_{c^*,\ell'}}(O). \label{eq:eif_psi2}
\end{align}
Both follow from the chain rule applied to $g(x, y, z) = xy/z$; see the derivation below.

\begin{proof}[Derivation of~\eqref{eq:eif_psi1}]
With $R = 1$, $\psi_1 = g(\tau_{k^\star}, \tau_j, \tau_{j'}) = \tau_{k^\star}\tau_j/\tau_{j'}$. The gradient is $\nabla g = (\tau_j/\tau_{j'},\; \tau_{k^\star}/\tau_{j'},\; -\tau_{k^\star}\tau_j/\tau_{j'}^2)^\top = (R_1,\; \psi_1/\tau_j,\; -\psi_1/\tau_{j'})^\top$. Applying the functional delta method with building-block EIFs $\phi_{\tau_{k^\star}}$, $\phi_{\tau_j}$, $\phi_{\tau_{j'}}$ yields~\eqref{eq:eif_psi1}. Strategy~2 is identical with $(\tau_j, \tau_{j'})$ replaced by $(\bar{\mu}_{c^*,\ell}, \bar{\mu}_{c^*,\ell'})$. This is the $R=1$ specialization of the rotation-invariant form~\eqref{eq:eif_general}: with $\mathbf{T}_{\mathrm{anc}}^{\mathrm{source}} \to \tau_{j'}$ (scalar) and $\boldsymbol{\tau}_{k^\star}^{\mathrm{source}} \to \tau_{k^\star}$ (scalar), the first term becomes $(\tau_{j'})^{-1}\tau_j \phi_{\tau_{k^\star}} = R_1\phi_{\tau_{k^\star}}$, the second becomes $(\tau_{k^\star}/\tau_{j'})\phi_{\tau_j} = (\psi_1/\tau_j)\phi_{\tau_j}$, and the third becomes $-(\tau_{k^\star}/\tau_{j'})(1/\tau_{j'})\tau_j\phi_{\tau_{j'}} = -(\psi_1/\tau_{j'})\phi_{\tau_{j'}}$.
\end{proof}

\paragraph{Multi-anchor combination (Section~\ref{sec:estimation}, \protect\eqref{eq:est_S1_R}).}
With $R = 1$ and $m$ anchor arms, the GLS solution $\hat{\boldsymbol{\Lambda}}^{\mathrm{target}} = (\boldsymbol{\Theta}_c^\top \mathbf{V}_{\boldsymbol{\mu}}^{-1} \boldsymbol{\Theta}_c)^{-1} \boldsymbol{\Theta}_c^\top \mathbf{V}_{\boldsymbol{\mu}}^{-1} \hat{\boldsymbol{\mu}}^{\mathrm{target}}$ reduces to inverse-variance weighting over the $m$ single-anchor ratios $R_{2,j} = \bar{\mu}_{c_j, \ell}/\bar{\mu}_{c_j, \ell'}$:
\begin{equation}
w^* = \frac{\mathbf{V}_R^{-1}\mathbf{1}}{\mathbf{1}^\top \mathbf{V}_R^{-1}\mathbf{1}}, \qquad \hat{\psi}_2^* = \hat{\tau}_{k^\star} \cdot {w^*}^\top \hat{\mathbf{R}},
\end{equation}
where $\hat{\mathbf{R}} = (\hat{R}_{2,1}, \ldots, \hat{R}_{2,m})^\top$ and $\mathbf{V}_R$ is the asymptotic covariance of $\hat{\mathbf{R}}$. The specification statistic $Q$ of Section~\ref{sec:estimation} retains its $\chi^2_{m - 1}$ null distribution.

\paragraph{Breakdown frontier (Proposition~\ref{prop:sensitivity}, Corollary~\ref{cor:sensitivity_r1}).}
The general-$R$ bound $B_R(\gamma)$ in~\eqref{eq:bias_bounds} collapses to the scalar form $B(\gamma) = 2\gamma(1+|\rho|)(1 + |\bar{\tilde{\theta}}|/|\bar{\tilde{\theta}}_{\mathrm{anc}}|)$ in Corollary~\ref{cor:sensitivity_r1}: $\kappa(\boldsymbol{\Theta}) = 1$, $\sigma_{\min}(\boldsymbol{\Theta}) = |\bar{\tilde{\theta}}_{\mathrm{anc}}|$, and the condition-number amplification disappears. The tighter contrast-level bound $\tilde{B}(\tilde{\gamma}) = \tilde{\gamma}\, (1+|\rho|)(1 + |\bar{\tilde{\theta}}|/|\bar{\tilde{\theta}}_{\mathrm{anc}}|)$ uses the contrast-level sensitivity $\tilde{\gamma} \leq 2\gamma$ and vanishes under treatment-invariant contamination. The Strategy~2 analog is $|\mathrm{Bias}_2| \leq \gamma(1+|\rho_2|)(2 + |\bar{\tilde{\theta}}|/|\bar{\theta}_{c^*}|)$.

\section{Related Work}
\label{app:related}

Having established our framework, we now position it relative to existing literature. Our contribution sits at the intersection of several research streams, but differs from each in important ways.

\paragraph{Transportability and external validity.} The literature on generalizing experimental findings to new populations \citep{huang2024towards, degtiar2023review, bareinboim2016causal, cole2010generalizing, westreich2017transportability} typically assumes access to covariate information in the target population and leverages selection-on-observables and positivity assumptions. Our setting differs in that the ``target'' is defined by \textit{temporal shift} rather than \textit{covariate shift}. 

\paragraph{Meta-analysis.}
Meta-analytic methods \citep{borenstein2021introduction, higgins2009re, riley2011interpretation, parikh2025double} synthesize evidence across studies and can incorporate time as a study-level moderator. However, the goal is typically \emph{pooling}---estimating an average effect or characterizing heterogeneity across observed studies---rather than \emph{extrapolation} to counterfactual timing. Hierarchical Bayesian meta-analysis \citep{meager2019} is the closest relative: a model of the form $\hat{\tau}_k \sim \mathcal{N}(\theta_k, \sigma_k^2),\; \theta_k \sim \mathcal{N}(\alpha + \beta t_k, \tau^2)$ does produce a posterior for the effect at an unobserved target time, so it shares our extrapolation target. The tradeoff is where each approach places its structural commitment: our frequentist factor model commits to multiplicative separability but derives the temporal transport from \emph{observed} anchor trials at both source and target times; the Bayesian alternative commits to a parametric prior over how effects drift (e.g., linear in time with normal residuals) and lets the prior do the temporal extrapolation even when no anchor data spans source and target. Neither dominates: practitioners with credible anchor trials at target time should prefer the identification-based approach we develop, while practitioners with credible domain priors on temporal structure but no such anchors should prefer the Bayesian alternative. A detailed empirical comparison is beyond the scope of this paper's identification-focused contribution.

\paragraph{Factor models, synthetic control, and difference-in-differences.}
Our separable temporal effects assumption (Assumption~\ref{ass:separable}) shares structural similarities with several literatures that leverage factor structure for causal inference. Interactive fixed effects models in panel econometrics \citep{bai2009panel} decompose outcomes into unit-specific loadings and time-specific factors; matrix completion approaches \citep{athey2021matrix} view the potential outcome matrix as approximately low-rank; and synthetic control \citep{abadie2010synthetic, abadie2015comparative} and difference-in-differences methods \citep{callaway2021difference, sun2021estimating, goodman2021difference, dechaisemartin2020two, borusyak2024revisiting} implicitly rely on factor structure to construct counterfactuals, a connection made explicit by \citet{xu2017generalized} and \citet{arkhangelsky2021synthetic}. 

Our separability assumption (A.\ref{ass:separable}) is similar to this broader literature. However, our setting and goals differ in three key respects. \emph{First}, these methods estimate effects or impute missing outcomes \emph{within} the observed temporal support of the study; we \emph{extrapolate} to timing configurations where the relevant trial was not conducted. \emph{Second}, our setup involves integrating information from multiple trials comparing different intervention or treatment pairs. \emph{Third}, existing methods use outcomes from control units to impute missing potential outcome for the treated unit; we instead use outcomes from trials with \emph{entirely different treatments} to identify how effects scale across time.

\paragraph{Semiparametric estimation.}
Our estimation approach builds on the semiparametric efficiency literature \citep{robins1994estimation, bang2005doubly, kennedy2022semiparametric, parikh2025double}, particularly doubly robust methods and debiased machine learning \citep{chernozhukov2018double}.

\paragraph{Sensitivity analysis and partial identification.}
Our breakdown frontier framework draws on Rosenbaum's sensitivity analysis for unmeasured confounding \citep{rosenbaum2002observational}, Manski's partial identification theory \citep{manski2003partial}, and recent work on breakdown frontiers \citep{masten2021salvaging, cinelli2020making, dorn2022sharp}. We adapt this approach to the temporal transportation setting, where the sensitivity parameter measures non-separable contamination of treatment effects rather than unmeasured confounding.

\section{Proofs for Identification Results}
\label{app:proofs}

\begin{proof}[Proof of Theorem~\ref{thm:decomposition}]
Under random assignment and SUTVA, the observed ATE is identified by
\begin{align*}
\tau_{k^\star}(0,0) &= \mathbb{E}[Y_{t_{1k^\star}}(a_{k^\star}, t_{0k^\star}) - Y_{t_{1k^\star}}(b_{k^\star}, t_{0k^\star}) \mid S = k^\star] \\
&= \mathbb{E}[\mathbb{E}[Y_{t_{1k^\star}} \mid A = a_{k^\star}, S = k^\star, X] - \mathbb{E}[Y_{t_{1k^\star}} \mid A = b_{k^\star}, S = k^\star, X] \mid S = k^\star].
\end{align*}
(Under simple randomization, $A \perp X \mid S=k^\star$ and this iterated expectation collapses to the unadjusted difference; under covariate-adaptive designs the iterated-expectation form is the correct one.) Under Assumption~\ref{ass:separable} with rank $R$,
\begin{align*}
\mathbb{E}[Y_{t_1}(a, t_0) - Y_{t_1}(b, t_0) \mid X] &= \boldsymbol{\theta}_a(X)^\top \boldsymbol{\Lambda}(t_0, t_1) - \boldsymbol{\theta}_b(X)^\top \boldsymbol{\Lambda}(t_0, t_1) \\
&= \tilde{\boldsymbol{\theta}}_{a,b}(X)^\top \boldsymbol{\Lambda}(t_0, t_1),
\end{align*}
where $\tilde{\boldsymbol{\theta}}_{a,b}(X) := \boldsymbol{\theta}_a(X) - \boldsymbol{\theta}_b(X) \in \mathbb{R}^R$. Therefore
\begin{align*}
\tau_{k^\star}(0,0) &= \mathbb{E}[\tilde{\boldsymbol{\theta}}_{a_{k^\star}, b_{k^\star}}(X)^\top \boldsymbol{\Lambda}(t_{0k^\star}, t_{1k^\star}) \mid S = k^\star] \\
&= \bar{\tilde{\boldsymbol{\theta}}}_{a_{k^\star}, b_{k^\star}, k^\star}^\top \boldsymbol{\Lambda}(t_{0k^\star}, t_{1k^\star}),
\end{align*}
where $\bar{\tilde{\boldsymbol{\theta}}}_{a_{k^\star}, b_{k^\star}, k^\star} := \mathbb{E}[\tilde{\boldsymbol{\theta}}_{a_{k^\star}, b_{k^\star}}(X) \mid S = k^\star] \in \mathbb{R}^R$.

Similarly, the transported ATE satisfies:
\begin{align*}
\tau_{k^\star}(\delta_0, \delta_1) &= \mathbb{E}[Y_{t_{1k^\star} + \delta_1}(a_{k^\star}, t_{0k^\star} + \delta_0) - Y_{t_{1k^\star} + \delta_1}(b_{k^\star}, t_{0k^\star} + \delta_0) \mid S = k^\star] \\
&= \bar{\tilde{\boldsymbol{\theta}}}_{a_{k^\star}, b_{k^\star}, k^\star}^\top \boldsymbol{\Lambda}(t_{0k^\star} + \delta_0, t_{1k^\star} + \delta_1),
\end{align*}
which yields~\eqref{eq:tate_decomp}. When $R = 1$, dividing gives the scalar ratio form: $\tau_{k^\star}(\delta_0, \delta_1)/\tau_{k^\star}(0,0) = \Lambda(t_{0k^\star} + \delta_0, t_{1k^\star} + \delta_1)/\Lambda(t_{0k^\star}, t_{1k^\star})$.
\end{proof}

\begin{proof}[Proof of Corollary~\ref{cor:replicated}]
Under the random sampling condition in Assumption~\ref{ass:replicated}(i), each trial draws an independent sample from the common super-population $\mathcal{P}$, so $X \perp S$ at the population level. Consequently, for any trial $\ell$ comparing anchor pair $(a_r^*, b_r^*)$:
\begin{align*}
\bar{\tilde{\boldsymbol{\theta}}}_{a_r^*, b_r^*, \ell} := \mathbb{E}[\tilde{\boldsymbol{\theta}}_{a_r^*, b_r^*}(X) \mid S = \ell] = \mathbb{E}[\tilde{\boldsymbol{\theta}}_{a_r^*, b_r^*}(X)] =: \bar{\tilde{\boldsymbol{\theta}}}_{a_r^*, b_r^*},
\end{align*}
where the second equality follows purely from $X \perp S$ (not from transportability of $Y(a)$).

By the same argument as in Theorem~\ref{thm:decomposition}, for any trial $\ell$ comparing $(a_r^*, b_r^*)$ at time $(t_{0\ell}, t_{1\ell})$:
\begin{align*}
\tau_\ell(0,0) = \bar{\tilde{\boldsymbol{\theta}}}_{a_r^*, b_r^*}^\top \boldsymbol{\Lambda}(t_{0\ell}, t_{1\ell}).
\end{align*}

Stacking the $R$ anchor pairs at the target time yields the linear system:
\begin{align*}
\boldsymbol{\tau}^{\mathrm{target}} = \boldsymbol{\Theta} \boldsymbol{\Lambda}^{\mathrm{target}},
\end{align*}
where $\boldsymbol{\Theta} = [\bar{\tilde{\boldsymbol{\theta}}}_{a_1^*, b_1^*}, \ldots, \bar{\tilde{\boldsymbol{\theta}}}_{a_R^*, b_R^*}]^\top \in \mathbb{R}^{R \times R}$. By Assumption~\ref{ass:replicated}(ii), $\boldsymbol{\Theta}$ has full rank, so $\boldsymbol{\Lambda}^{\mathrm{target}} = \boldsymbol{\Theta}^{-1}\boldsymbol{\tau}^{\mathrm{target}}$.

When $R = 1$, $\boldsymbol{\Theta}$ reduces to the scalar $\bar{\tilde{\theta}}_{a_1^*, b_1^*} = \bar{\tilde{\theta}}_{\mathrm{anc}}$, and the system becomes $\Lambda^{\mathrm{target}} = \tau_\ell / \bar{\tilde{\theta}}_{\mathrm{anc}}$. Taking the ratio with the source-time equation gives $\Lambda^{\mathrm{target}}/\Lambda^{\mathrm{source}} = \tau_\ell/\tau_{\ell'}$.
\end{proof}

\begin{proof}[Proof of Corollary~\ref{cor:common_arm}]
Under Assumption~\ref{ass:separable} with treatment $c_r^*$:
\begin{align*}
\mathbb{E}[Y_{t_1}(c_r^*, t_0) \mid X] = \boldsymbol{\theta}_{c_r^*}(X)^\top \boldsymbol{\Lambda}(t_0, t_1).
\end{align*}

Under Assumption~\ref{ass:measurement}, $\boldsymbol{\Lambda}(t_0, t_1) = \boldsymbol{\Lambda}(t_1)$, so:
\begin{align*}
\mathbb{E}[Y_{t_1}(c_r^*, t_0) \mid X] = \boldsymbol{\theta}_{c_r^*}(X)^\top \boldsymbol{\Lambda}(t_1).
\end{align*}

For units receiving treatment $c_r^*$ in trial $\ell \in \mathcal{K}[c_r^*]$, by random assignment:
\begin{align*}
\mathbb{E}[Y_{t_{1\ell}} \mid A = c_r^*, S = \ell] &= \mathbb{E}[Y_{t_{1\ell}}(c_r^*, t_{0\ell}) \mid S = \ell] \\
&= \mathbb{E}\bigl[\boldsymbol{\theta}_{c_r^*}(X)^\top \boldsymbol{\Lambda}(t_{1\ell}) \bigm| S = \ell\bigr] \\
&= \bar{\boldsymbol{\theta}}_{c_r^*}^\top \boldsymbol{\Lambda}(t_{1\ell}),
\end{align*}
where the last equality uses the compositional-stability condition $\bar{\boldsymbol{\theta}}_{c_r^*} := \mathbb{E}[\boldsymbol{\theta}_{c_r^*}(X) \mid S = \ell] = \mathbb{E}[\boldsymbol{\theta}_{c_r^*}(X)]$ for all $\ell$.

Stacking the $R$ anchor arms at the target measurement time yields:
\begin{align*}
\boldsymbol{\mu}^{\mathrm{target}} = \boldsymbol{\Theta}_c \boldsymbol{\Lambda}(t_{1k^\star} + \delta_1),
\end{align*}
where $\boldsymbol{\Theta}_c = [\bar{\boldsymbol{\theta}}_{c_1^*}, \ldots, \bar{\boldsymbol{\theta}}_{c_R^*}]^\top \in \mathbb{R}^{R \times R}$. By Assumption~\ref{ass:common}(ii), $\boldsymbol{\Theta}_c$ has full rank, so $\boldsymbol{\Lambda}(t_{1k^\star} + \delta_1) = \boldsymbol{\Theta}_c^{-1}\boldsymbol{\mu}^{\mathrm{target}}$.

When $R = 1$, $\boldsymbol{\Theta}_c$ is the scalar $\bar{\theta}_{c^*}$, and the ratio of target-time to source-time equations gives $\Lambda(t_{1k^\star} + \delta_1)/\Lambda(t_{1k^\star}) = \bar{\mu}_{c^*,\ell}/\bar{\mu}_{c^*,\ell'}$.
\end{proof}

\section{Proofs for Estimation Results}
\label{app:estimation_proofs}

\subsection{Properties of Doubly Robust Scores}
\label{app:dr_properties}

\begin{lemma}[Properties of Doubly Robust Scores]
\label{lem:dr_scores}
Under random assignment, SUTVA, and Assumption~\ref{ass:regularity}(i)--(ii), the score $\varphi_{a,k}(O; \eta)$ defined in~\eqref{eq:score_mean} satisfies:
\begin{enumerate}
    \item[(i)] \emph{Unbiasedness}: $\mathbb{E}[\varphi_{a,k}(O; \eta)] = \bar{\mu}_{a,k}$.
    \item[(ii)] \emph{Neyman orthogonality}: The pathwise derivative with respect to each nuisance function vanishes at the truth.
    \item[(iii)] \emph{Bounded variance}: $\mathbb{E}[\varphi_{a,k}(O; \eta)^2] < \infty$.
\end{enumerate}
Consequently, $\phi_{\bar{\mu}_{a,k}}(O) = \varphi_{a,k}(O; \eta) - \bar{\mu}_{a,k}$ is the efficient influence function for $\bar{\mu}_{a,k}$.
\end{lemma}

\begin{proof}
\textit{Part (i): Unbiasedness.} We compute the expectation of each term in~\eqref{eq:score_mean}. For the first term:
\begin{align*}
\mathbb{E}\left[\frac{\mathbf{1}[S = k]}{\pi_k(X)} \cdot \frac{\mathbf{1}[A = a]}{e_k(a, X)} (Y - \mu_{a,k}(X))\right]
&= \mathbb{E}\!\left[\frac{\mathbf{1}[S = k]}{\pi_k(X)} \cdot \frac{\mathbf{1}[A = a]}{e_k(a, X)} \mathbb{E}[Y - \mu_{a,k}(X) \mid A, S, X]\right] \\
&= 0,
\end{align*}
since $\mathbb{E}[Y \mid A = a, S = k, X] = \mu_{a,k}(X)$ by definition. The second term directly yields the target:
\[
\mathbb{E}[\mu_{a,k}(X)] = \bar{\mu}_{a,k}
\]
under compositional stability ($X \perp S$), which equates the pooled mean $\mathbb{E}[\mu_{a,k}(X)]$ with the trial-$k$ marginal $\mathbb{E}[Y \mid A = a, S = k]$.

\textit{Part (ii): Neyman orthogonality.} We verify that the pathwise derivative with respect to each nuisance component vanishes at the true parameter values.

\textit{Derivative with respect to $\mu_{a,k}$}: Let $\mu_{a,k}^r = \mu_{a,k} + r(\tilde{\mu}_{a,k} - \mu_{a,k})$ for $r \in [0,1]$. Differentiating both terms of~\eqref{eq:score_mean}:
\begin{align*}
\frac{\partial}{\partial r} \mathbb{E}[\varphi_{a,k}(O; \pi_k, e_k, \mu_{a,k}^r)]\Big|_{r=0} = \mathbb{E}\left[\left(1 - \frac{\mathbf{1}[S = k]}{\pi_k(X)} \cdot \frac{\mathbf{1}[A = a]}{e_k(a, X)}\right) (\tilde{\mu}_{a,k}(X) - \mu_{a,k}(X))\right].
\end{align*}
Conditioning on $X$:
\begin{align*}
\mathbb{E}\left[\frac{\mathbf{1}[S = k]\mathbf{1}[A = a]}{\pi_k(X) e_k(a, X)} \,\Big|\, X\right] &= \frac{P(S = k \mid X)\cdot P(A = a \mid S = k, X)}{\pi_k(X) e_k(a, X)} = 1,
\end{align*}
so the derivative equals zero.

\textit{Derivative with respect to $e_k$}: Let $e_k^r(a, X) = e_k(a, X) + r(\tilde{e}_k(a, X) - e_k(a, X))$. Only the first term of~\eqref{eq:score_mean} depends on $e_k$:
\begin{align*}
\frac{\partial}{\partial r} \mathbb{E}[\varphi_{a,k}(O; \pi_k, e_k^r, \mu_{a,k})]\Big|_{r=0}
&= \mathbb{E}\!\left[\frac{\mathbf{1}[S = k]}{\pi_k(X)} \cdot \frac{-\mathbf{1}[A = a]}{e_k(a, X)^2} (Y - \mu_{a,k}(X)) \right. \\
&\qquad \left. \cdot (\tilde{e}_k(a, X) - e_k(a, X))\right].
\end{align*}
Conditioning on $(S, X)$ and using $\mathbb{E}[Y - \mu_{a,k}(X) \mid A = a, S = k, X] = 0$ makes the conditional expectation zero.

\textit{Derivative with respect to $\pi_k$}: Let $\pi_k^r(X) = \pi_k(X) + r(\tilde{\pi}_k(X) - \pi_k(X))$. Only the first term of~\eqref{eq:score_mean} depends on $\pi_k$:
\begin{align*}
\frac{\partial}{\partial r} \mathbb{E}[\varphi_{a,k}(O; \pi_k^r, e_k, \mu_{a,k})]\Big|_{r=0}
&= \mathbb{E}\!\left[\frac{-\mathbf{1}[S = k]\mathbf{1}[A = a]}{\pi_k(X)^2 e_k(a, X)} (Y - \mu_{a,k}(X)) \cdot (\tilde{\pi}_k(X) - \pi_k(X))\right].
\end{align*}
Conditioning on $(S, X)$ again makes $\mathbb{E}[Y - \mu_{a,k}(X) \mid A = a, S = k, X] = 0$, so the derivative vanishes. This is the key advantage of placing $\mu_{a,k}(X)$ outside the trial-membership weight: the resulting score is Neyman-orthogonal in all three nuisance functions, including $\pi_k$, so plug-in machine-learning nuisance estimation at the $n^{-1/4}$ product rate retains $\sqrt{n}$-consistency.

\textit{Part (iii): Bounded variance.} Assumption~\ref{ass:regularity}(i) ensures $\pi_k(X)^{-1} \leq \epsilon^{-1}$ and $e_k(a, X)^{-1} \leq \epsilon^{-1}$ almost surely. Combined with Assumption~\ref{ass:regularity}(ii) giving $\mathbb{E}[Y^4] < \infty$, the Cauchy-Schwarz inequality yields:
\begin{align*}
\mathbb{E}[\varphi_{a,k}(O; \eta)^2] \leq \epsilon^{-4} \mathbb{E}\left[\left(\frac{\mathbf{1}[A = a]}{e_k(a,X)}|Y - \mu_{a,k}(X)| + |\mu_{a,k}(X)|\right)^2\right] < \infty.
\end{align*}
\end{proof}

\subsection{Derivation of Influence Functions}
\label{app:proof_eif}

This subsection states and proves the efficient influence functions referenced in the main text. 

\begin{remark}[Observable reformulation of the TATE functional]
\label{rem:theta_ident}
The factor model $\bar{\mu}_{a,k}(x) = \boldsymbol{\theta}_a(x)^\top \boldsymbol{\Lambda}(t_{0k}, t_{1k})$ is invariant under non-singular rotations: for any invertible $\mathbf{P} \in \mathbb{R}^{R\times R}$, the substitution $\boldsymbol{\theta}_a \mapsto \mathbf{P}^{-\top}\boldsymbol{\theta}_a$, $\boldsymbol{\Lambda} \mapsto \mathbf{P}\boldsymbol{\Lambda}$ preserves all observable ATEs. Therefore the latent effect-contrast matrices $\boldsymbol{\Theta}$ (anchor pairs) and $\mathbf{M}$ (target trial at $R$ source times) are identified only up to a common rotation, and no estimator can depend on them individually. Observables are products of effect contrasts and temporal effects: the anchor source-time matrix
\[
\mathbf{T}_{\mathrm{anc}}^{\mathrm{source}} := \boldsymbol{\Theta}\, \mathbf{M}^\top \in \mathbb{R}^{R \times R}
\]
stacks, in row $r$ and column $s$, the expected anchor-pair-$r$ ATE at source time $s$; the target-trial source-time vector $\boldsymbol{\tau}_{k^\star}^{\mathrm{source}} = \mathbf{M}\bar{\tilde{\boldsymbol{\theta}}}_{k^\star}$ stacks the target trial's ATEs at the $R$ source times; and the anchor target-time vector $\boldsymbol{\tau}^{\mathrm{target}} = \boldsymbol{\Theta}\boldsymbol{\Lambda}^{\mathrm{target}}$ stacks anchor-pair ATEs at the target time. Algebraically,
\[
\psi_1 \;=\; \bar{\tilde{\boldsymbol{\theta}}}_{k^\star}^\top\boldsymbol{\Theta}^{-1}\boldsymbol{\tau}^{\mathrm{target}} \;=\; (\boldsymbol{\tau}_{k^\star}^{\mathrm{source}})^\top \bigl(\mathbf{T}_{\mathrm{anc}}^{\mathrm{source}}\bigr)^{-1}\boldsymbol{\tau}^{\mathrm{target}},
\]
depending only on observable ATEs. We state Proposition~\ref{prop:eif} in this rotation-invariant form; the derivation is the same as with $\boldsymbol{\Theta}$ and $\mathbf{M}$ separately, but the resulting EIF involves only observable differentials.
\end{remark}

\begin{proposition}[Efficient Influence Function for the TATE, general $R$]
\label{prop:eif}
Under the identification conditions of Theorem~\ref{thm:decomposition} and Corollary~\ref{cor:replicated}, the EIF for the Strategy~1 TATE $\psi_1 = (\boldsymbol{\tau}_{k^\star}^{\mathrm{source}})^\top (\mathbf{T}_{\mathrm{anc}}^{\mathrm{source}})^{-1}\boldsymbol{\tau}^{\mathrm{target}}$ is
\begin{align}
\phi_{\psi_1}(O) &= \bigl(\mathbf{T}_{\mathrm{anc}}^{\mathrm{source}}\bigr)^{-\top} \boldsymbol{\tau}^{\mathrm{target}} \cdot \boldsymbol{\phi}_{\boldsymbol{\tau}_{k^\star}^{\mathrm{source}}}(O) \;+\; (\boldsymbol{\tau}_{k^\star}^{\mathrm{source}})^\top \bigl(\mathbf{T}_{\mathrm{anc}}^{\mathrm{source}}\bigr)^{-1} \boldsymbol{\phi}_{\boldsymbol{\tau}^{\mathrm{target}}}(O) \notag \\[4pt]
&\quad - (\boldsymbol{\tau}_{k^\star}^{\mathrm{source}})^\top \bigl(\mathbf{T}_{\mathrm{anc}}^{\mathrm{source}}\bigr)^{-1} \mathrm{d}\mathbf{T}_{\mathrm{anc}}^{\mathrm{source}}[\boldsymbol{\phi}] \bigl(\mathbf{T}_{\mathrm{anc}}^{\mathrm{source}}\bigr)^{-1} \boldsymbol{\tau}^{\mathrm{target}}, \label{eq:eif_general}
\end{align}
where $\boldsymbol{\phi}_{\boldsymbol{\tau}_{k^\star}^{\mathrm{source}}}(O)$ and $\boldsymbol{\phi}_{\boldsymbol{\tau}^{\mathrm{target}}}(O)$ stack the EIFs of the building-block ATEs and $\mathrm{d}\mathbf{T}_{\mathrm{anc}}^{\mathrm{source}}[\boldsymbol{\phi}]$ denotes the matrix of ATE EIFs arranged in the same layout as $\mathbf{T}_{\mathrm{anc}}^{\mathrm{source}}$. An analogous expression holds for Strategy~2, replacing $\mathbf{T}_{\mathrm{anc}}^{\mathrm{source}}$ with the anchor-arm source-time mean matrix $\mathbf{T}_{c}^{\mathrm{source}}$ and $\boldsymbol{\tau}^{\mathrm{target}}$ with the anchor-arm target-time mean vector $\boldsymbol{\mu}^{\mathrm{target}}$.
\end{proposition}

\subsubsection{Proof of Proposition~\ref{prop:eif}}

\begin{proof}[Proof of Proposition~\ref{prop:eif}]
The functional is $\psi_1 = g(\mathbf{u}, \mathbf{A}, \mathbf{v}) := \mathbf{u}^\top \mathbf{A}^{-1}\mathbf{v}$, evaluated at $\mathbf{u} = \boldsymbol{\tau}_{k^\star}^{\mathrm{source}}$, $\mathbf{A} = \mathbf{T}_{\mathrm{anc}}^{\mathrm{source}}$, $\mathbf{v} = \boldsymbol{\tau}^{\mathrm{target}}$. Each of these three building blocks is a collection of observable ATEs, with their own EIFs $\boldsymbol{\phi}_{\mathbf{u}}$, $\mathrm{d}\mathbf{A}[\boldsymbol{\phi}]$, $\boldsymbol{\phi}_{\mathbf{v}}$ (each component is the DR score EIF of Lemma~\ref{lem:dr_scores}).

\textit{Step 1: Derivative with respect to $\mathbf{u}$.} Since $g$ is linear in $\mathbf{u}$: $\partial g/\partial\mathbf{u} = \mathbf{A}^{-\top}\mathbf{v}$. This contributes the first term in~\eqref{eq:eif_general}.

\textit{Step 2: Derivative with respect to $\mathbf{v}$.} Since $g$ is linear in $\mathbf{v}$: $\partial g/\partial\mathbf{v} = \mathbf{A}^{-\top}\mathbf{u}$. This contributes the second term.

\textit{Step 3: Derivative with respect to $\mathbf{A}$.} Using $\mathrm{d}(\mathbf{A}^{-1}) = -\mathbf{A}^{-1}(\mathrm{d}\mathbf{A})\mathbf{A}^{-1}$: $\mathrm{d}_{\mathbf{A}} g = -\mathbf{u}^\top \mathbf{A}^{-1}(\mathrm{d}\mathbf{A})\mathbf{A}^{-1}\mathbf{v}$, giving the third term.

\textit{Step 4: Assembling.} Combining Steps 1--3 via the multivariate delta method and substituting the building-block EIFs yields~\eqref{eq:eif_general}. Each building block inherits Neyman orthogonality from Lemma~\ref{lem:dr_scores}.
\end{proof}

\subsubsection{Auxiliary Influence Functions}

\begin{proposition}[Influence Function for Ratio Components]
\label{prop:if_ratio}
The influence function for $R_{2,j} = \bar{\mu}_{c_j,\ell_j} / \bar{\mu}_{c_j,\ell'_j}$ is:
\begin{align*}
\phi_{R_{2,j}}(O) = \frac{1}{\bar{\mu}_{c_j,\ell'_j}} \phi_{\bar{\mu}_{c_j,\ell_j}}(O) - \frac{R_{2,j}}{\bar{\mu}_{c_j,\ell'_j}} \phi_{\bar{\mu}_{c_j,\ell'_j}}(O).
\end{align*}
\end{proposition}

\begin{proof}
For $h(y, z) = y/z$, the delta method gives:
\begin{align*}
\phi_{R_{2,j}}(O) = \frac{\partial h}{\partial y} \phi_{\bar{\mu}_{c_j,\ell_j}}(O) + \frac{\partial h}{\partial z} \phi_{\bar{\mu}_{c_j,\ell'_j}}(O) = \frac{1}{\bar{\mu}_{c_j,\ell'_j}} \phi_{\bar{\mu}_{c_j,\ell_j}}(O) - \frac{\bar{\mu}_{c_j,\ell_j}}{\bar{\mu}_{c_j,\ell'_j}^2} \phi_{\bar{\mu}_{c_j,\ell'_j}}(O).
\end{align*}
Using $R_{2,j} = \bar{\mu}_{c_j,\ell_j} / \bar{\mu}_{c_j,\ell'_j}$ gives the result.
\end{proof}

\begin{proposition}[Influence Function for Combined Estimator]
\label{prop:if_combined}
The influence function for $\psi_2^* = \tau_{k^\star} \cdot R_2^*$ with $R_2^* = \sum_{j=1}^m w_j^* R_{2,j}$ is:
\begin{align*}
\phi_{\psi_2^*}(O) = R_2^* \cdot \phi_{\tau_{k^\star}}(O) + \tau_{k^\star} \sum_{j=1}^m w_j^* \phi_{R_{2,j}}(O).
\end{align*}
\end{proposition}

\begin{proof}
The estimand $\psi_2^* = \tau_{k^\star} \cdot \sum_{j=1}^m w_j^* R_{2,j}$ is a smooth function of $(\tau_{k^\star}, R_{2,1}, \ldots, R_{2,m})$. By the delta method:
\begin{align*}
\phi_{\psi_2^*}(O) &= \frac{\partial \psi_2^*}{\partial \tau_{k^\star}} \phi_{\tau_{k^\star}}(O) + \sum_{j=1}^m \frac{\partial \psi_2^*}{\partial R_{2,j}} \phi_{R_{2,j}}(O) \\
&= R_2^* \cdot \phi_{\tau_{k^\star}}(O) + \sum_{j=1}^m \tau_{k^\star} w_j^* \phi_{R_{2,j}}(O) \\
&= R_2^* \cdot \phi_{\tau_{k^\star}}(O) + \tau_{k^\star} \sum_{j=1}^m w_j^* \phi_{R_{2,j}}(O).
\end{align*}
\end{proof}

\subsection{Deriving Efficient Estimators from the EIF}
\label{sec:deriving_estimators}

The TATE is a smooth function of identified quantities. For $R = 1$, $\psi_1 = g(\tau_{k^\star}, \tau_j, \tau_{j'})$ where $g(x,y,z) = xy/z$; for general $R$, $\psi_1 = \bar{\tilde{\boldsymbol{\theta}}}_{k^\star}^\top \boldsymbol{\Theta}^{-1}\boldsymbol{\tau}^{\mathrm{target}}$. We show that plug-in estimation with doubly robust building blocks yields semiparametrically efficient estimators. The derivation below uses the $R = 1$ functional for concreteness; the general-$R$ case follows by the same Taylor expansion argument applied to the matrix-valued functional of Proposition~\ref{prop:eif}.

The doubly robust score for $\bar{\mu}_{a,k}$ is:
\begin{align}
\varphi_{a,k}(O; \eta) = \frac{\mathbf{1}[S = k]}{\pi_k(X)} \cdot \frac{\mathbf{1}[A = a]}{e_k(a, X)} \big(Y - \mu_{a,k}(X)\big) + \mu_{a,k}(X).
\end{align}
This score satisfies $\mathbb{E}[\varphi_{a,k}(O; \eta)] = \bar{\mu}_{a,k}$ at the true nuisances under compositional stability ($X \perp S$), and the centered score $\phi_{\bar{\mu}_{a,k}}(O) = \varphi_{a,k}(O; \eta) - \bar{\mu}_{a,k}$ is the efficient influence function.

Under Assumption~\ref{ass:rates}, the estimator $\hat{\bar{\mu}}_{a,k} = \mathbb{P}_n[\varphi_{a,k}(O; \hat{\eta})]$ satisfies:
\begin{align}
\hat{\bar{\mu}}_{a,k} - \bar{\mu}_{a,k} = \mathbb{P}_n[\phi_{\bar{\mu}_{a,k}}(O)] + o_p(n^{-1/2}).
\end{align}
The $o_p(n^{-1/2})$ remainder arises from Neyman orthogonality: the cross-term involving nuisance estimation error is second-order. Similarly, $\hat{\tau}_k = \mathbb{P}_n[\varphi_{\tau_k}(O; \hat{\eta})]$ satisfies:
\begin{align}
\hat{\tau}_k - \tau_k = \mathbb{P}_n[\phi_{\tau_k}(O)] + o_p(n^{-1/2}).
\end{align}

Consider the plug-in estimator $\hat{\psi}_1 = g(\hat{\tau}_{k^\star}, \hat{\tau}_j, \hat{\tau}_{j'}) = \hat{\tau}_{k^\star} \cdot \hat{\tau}_j / \hat{\tau}_{j'}$. A first-order Taylor expansion around the true values gives:
\begin{align}
\hat{\psi}_1 - \psi_1 &= \nabla g(\tau_{k^\star}, \tau_j, \tau_{j'})^\top 
\begin{pmatrix} \hat{\tau}_{k^\star} - \tau_{k^\star} \\ \hat{\tau}_j - \tau_j \\ \hat{\tau}_{j'} - \tau_{j'} \end{pmatrix} + o_p(n^{-1/2}) \\[6pt]
&= \frac{\tau_j}{\tau_{j'}}(\hat{\tau}_{k^\star} - \tau_{k^\star}) + \frac{\tau_{k^\star}}{\tau_{j'}}(\hat{\tau}_j - \tau_j) - \frac{\tau_{k^\star}\tau_j}{\tau_{j'}^2}(\hat{\tau}_{j'} - \tau_{j'}) + o_p(n^{-1/2}) \\[6pt]
&= R_1(\hat{\tau}_{k^\star} - \tau_{k^\star}) + \frac{\psi_1}{\tau_j}(\hat{\tau}_j - \tau_j) - \frac{\psi_1}{\tau_{j'}}(\hat{\tau}_{j'} - \tau_{j'}) + o_p(n^{-1/2}).
\end{align}

Substituting the results:
\begin{align}
\hat{\psi}_1 - \psi_1 &= R_1 \cdot \mathbb{P}_n[\phi_{\tau_{k^\star}}(O)] + \frac{\psi_1}{\tau_j} \cdot \mathbb{P}_n[\phi_{\tau_j}(O)] - \frac{\psi_1}{\tau_{j'}} \cdot \mathbb{P}_n[\phi_{\tau_{j'}}(O)] + o_p(n^{-1/2}) \\[6pt]
&= \mathbb{P}_n\left[ R_1 \cdot \phi_{\tau_{k^\star}}(O) + \frac{\psi_1}{\tau_j} \phi_{\tau_j}(O) - \frac{\psi_1}{\tau_{j'}} \phi_{\tau_{j'}}(O) \right] + o_p(n^{-1/2}) \\[6pt]
&= \mathbb{P}_n[\phi_{\psi_1}(O)] + o_p(n^{-1/2}).
\end{align}

This establishes that $\hat{\psi}_1$ is asymptotically linear with influence function $\phi_{\psi_1}(O)$—precisely the efficient influence function derived in Proposition~\ref{prop:eif}. By the convolution theorem, any regular estimator has asymptotic variance at least $\mathbb{E}[\phi_{\psi_1}(O)^2]$, and $\hat{\psi}_1$ achieves this bound.

\paragraph{Strategy 2.}
The same argument applies to $\hat{\psi}_2 = \hat{\tau}_{k^\star} \cdot \hat{\bar{\mu}}_{c^*,\ell} / \hat{\bar{\mu}}_{c^*,\ell'}$. The Taylor expansion yields:
\begin{align}
\hat{\psi}_2 - \psi_2 = R_2(\hat{\tau}_{k^\star} - \tau_{k^\star}) + \frac{\psi_2}{\bar{\mu}_{c^*,\ell}}(\hat{\bar{\mu}}_{c^*,\ell} - \bar{\mu}_{c^*,\ell}) - \frac{\psi_2}{\bar{\mu}_{c^*,\ell'}}(\hat{\bar{\mu}}_{c^*,\ell'} - \bar{\mu}_{c^*,\ell'}) + o_p(n^{-1/2}),
\end{align}
and substituting asymptotic linear representations gives $\hat{\psi}_2 - \psi_2 = \mathbb{P}_n[\phi_{\psi_2}(O)] + o_p(n^{-1/2})$.

\begin{remark}
The plug-in estimator achieves efficiency because: (i) each building block is estimated using its own efficient influence function via doubly robust scores; (ii) the TATE is a smooth function of these building blocks; and (iii) the delta method preserves asymptotic linearity. This would fail if we used inefficient estimators for the building blocks (e.g., simple sample means without covariate adjustment) or if the function $g$ were non-smooth at the true parameter values.
\end{remark}

\subsection{Proof of Theorem~\ref{thm:asymptotic}}
\label{app:proof_asymptotic}

\begin{proof}[Proof of Part (a): Asymptotic Normality]
We prove the result for $\hat{\psi}_1$; the proofs for $\hat{\psi}_2$ and $\hat{\psi}_2^*$ follow analogously.

\textit{Step 1: Asymptotic linearity of building blocks.} By Lemma~\ref{lem:dr_scores}, the score $\varphi_{\tau_k}(O; \eta)$ satisfies Neyman orthogonality. Combined with Assumption~\ref{ass:rates}, standard arguments from \citet{chernozhukov2018double} establish that:
\begin{align*}
\hat{\tau}_k - \tau_k = \mathbb{P}_n[\phi_{\tau_k}(O)] + o_p(n^{-1/2}) \quad \text{for } k \in \{k^\star, j, j'\}.
\end{align*}

To see this, write:
\begin{align*}
\hat{\tau}_k - \tau_k &= \mathbb{P}_n[\varphi_{\tau_k}(O; \hat{\eta})] - \tau_k \\
&= \mathbb{P}_n[\varphi_{\tau_k}(O; \hat{\eta}) - \varphi_{\tau_k}(O; \eta)] + \mathbb{P}_n[\varphi_{\tau_k}(O; \eta)] - \tau_k + \tau_k - \tau_k \\
&= \mathbb{P}_n[\phi_{\tau_k}(O)] + \mathbb{P}_n[\varphi_{\tau_k}(O; \hat{\eta}) - \varphi_{\tau_k}(O; \eta)].
\end{align*}

The second term is $o_p(n^{-1/2})$ by Neyman orthogonality and Assumption~\ref{ass:rates}. Specifically, a second-order expansion around the true $\eta$ shows:
\begin{align*}
\mathbb{P}_n[\varphi_{\tau_k}(O; \hat{\eta}) - \varphi_{\tau_k}(O; \eta)]
&= O_p(\|\hat{\eta} - \eta\|_2^2) + O_p(\|\hat{\mu}_{a,k} - \mu_{a,k}\|_2 \cdot \|\hat{\pi}_k - \pi_k\|_2) \\
&\quad + O_p(\|\hat{\mu}_{a,k} - \mu_{a,k}\|_2 \cdot \|\hat{e}_k - e_k\|_2),
\end{align*}
which is $o_p(n^{-1/2})$ under Assumption~\ref{ass:rates}.

\textit{Step 2: Delta method.} Define $\boldsymbol{\tau} = (\tau_{k^\star}, \tau_j, \tau_{j'})^\top$ and $\hat{\boldsymbol{\tau}} = (\hat{\tau}_{k^\star}, \hat{\tau}_j, \hat{\tau}_{j'})^\top$. The function $g(x, y, z) = xy/z$ is continuously differentiable in a neighborhood of $\boldsymbol{\tau}$ by Assumption~\ref{ass:regularity}(iii), which ensures $\tau_{j'} \neq 0$.

By the functional delta method \citep{van1998asymptotic}:
\begin{align*}
\hat{\psi}_1 - \psi_1 &= g(\hat{\boldsymbol{\tau}}) - g(\boldsymbol{\tau}) \\
&= \nabla g(\boldsymbol{\tau})^\top (\hat{\boldsymbol{\tau}} - \boldsymbol{\tau}) + o_p(\|\hat{\boldsymbol{\tau}} - \boldsymbol{\tau}\|) \\
&= \nabla g(\boldsymbol{\tau})^\top \mathbb{P}_n[\boldsymbol{\phi}_{\boldsymbol{\tau}}(O)] + o_p(n^{-1/2}) \\
&= \mathbb{P}_n[\phi_{\psi_1}(O)] + o_p(n^{-1/2}),
\end{align*}
where $\boldsymbol{\phi}_{\boldsymbol{\tau}}(O) = (\phi_{\tau_{k^\star}}(O), \phi_{\tau_j}(O), \phi_{\tau_{j'}}(O))^\top$ and $\phi_{\psi_1}(O) = \nabla g(\boldsymbol{\tau})^\top \boldsymbol{\phi}_{\boldsymbol{\tau}}(O)$ as derived in Appendix~\ref{app:R1_case}.

\textit{Step 3: Central limit theorem.} By Lemma~\ref{lem:dr_scores}, $\mathbb{E}[\phi_{\tau_k}(O)] = 0$ for each $k$, which implies $\mathbb{E}[\phi_{\psi_1}(O)] = 0$. The variance $V_1 := \mathbb{E}[\phi_{\psi_1}(O)^2]$ is finite by Lemma~\ref{lem:dr_scores}(iii) and the continuous mapping theorem.

The Lindeberg-L\'evy central limit theorem yields:
\begin{align*}
\sqrt{n} \, \mathbb{P}_n[\phi_{\psi_1}(O)] \xrightarrow{d} \mathcal{N}(0, V_1).
\end{align*}

Combined with Step 2:
\begin{align*}
\sqrt{n}(\hat{\psi}_1 - \psi_1) = \sqrt{n} \, \mathbb{P}_n[\phi_{\psi_1}(O)] + o_p(1) \xrightarrow{d} \mathcal{N}(0, V_1).
\end{align*}

\textit{Step 4: Variance estimation.} By the continuous mapping theorem and consistency of $\hat{\eta}$:
\begin{align*}
\hat{V}_1 = \mathbb{P}_n[\hat{\phi}_{\psi_1}^2] \xrightarrow{p} \mathbb{E}[\phi_{\psi_1}(O)^2] = V_1.
\end{align*}

\textit{Extension to $\hat{\psi}_2$ and $\hat{\psi}_2^*$.} The proof for $\hat{\psi}_2$ is identical, replacing $(\tau_j, \tau_{j'})$ with $(\bar{\mu}_{c^*,\ell}, \bar{\mu}_{c^*,\ell'})$. For $\hat{\psi}_2^*$, the estimated weights $\hat{\mathbf{w}}^*$ introduce additional terms that are $o_p(n^{-1/2})$ by Assumption~\ref{ass:rates} and Slutsky's theorem.
\end{proof}

\begin{proof}[Proof of Part (b): Double Robustness]
We work under compositional stability ($X \perp S$), which makes $\pi_k(X) = P(S = k) =: p_k$ constant and aligns the pooled estimand with the trial-$k$ marginal. It suffices to show that $\mathbb{E}[\varphi_{a,k}(O; \tilde{\eta})] = \bar{\mu}_{a,k}$ when either condition holds, since the TATE estimators are continuous functions of marginal means.

\textit{Case (a): Outcome model correctly specified.} Suppose $\mu_{a,k}(X) = \mathbb{E}[Y \mid A = a, S = k, X]$ but propensity scores may be misspecified as $\tilde{\pi}_k(X)$ and $\tilde{e}_k(a, X)$. With the new score form~\eqref{eq:score_mean}, write
\begin{align*}
\mathbb{E}[\varphi_{a,k}(O; \tilde{\eta})] = \mathbb{E}\left[\frac{\mathbf{1}[S = k]}{\tilde{\pi}_k(X)} \cdot \frac{\mathbf{1}[A = a]}{\tilde{e}_k(a, X)} (Y - \mu_{a,k}(X))\right] + \mathbb{E}[\mu_{a,k}(X)].
\end{align*}
The first term equals zero since $\mathbb{E}[Y - \mu_{a,k}(X) \mid A = a, S = k, X] = 0$ by correct specification of the outcome model, regardless of the propensity score specification. The second term equals $\mathbb{E}[\mu_{a,k}(X)] = \bar{\mu}_{a,k}$ under compositional stability, because the pooled mean equals the trial-$k$ marginal $\mathbb{E}[\mu_{a,k}(X) \mid S = k] = \bar{\mu}_{a,k}$. Hence $\mathbb{E}[\varphi_{a,k}(O; \tilde{\eta})] = \bar{\mu}_{a,k}$.

\textit{Case (b): Both propensity scores correctly specified.} Suppose $\pi_k(X) = P(S = k \mid X) = p_k$ (constant, by compositional stability) and $e_k(a, X) = P(A = a \mid S = k, X)$ are correct, but $\tilde{\mu}_{a,k}(X)$ may be misspecified. Then
\begin{align*}
\mathbb{E}[\varphi_{a,k}(O; \tilde{\eta})] &= \mathbb{E}\left[\frac{\mathbf{1}[S = k]}{p_k} \cdot \frac{\mathbf{1}[A = a]}{e_k(a, X)} Y\right] - \mathbb{E}\left[\frac{\mathbf{1}[S = k]}{p_k} \cdot \frac{\mathbf{1}[A = a]}{e_k(a, X)} \tilde{\mu}_{a,k}(X)\right] + \mathbb{E}[\tilde{\mu}_{a,k}(X)].
\end{align*}
The first term equals $\mathbb{E}[Y \mid A = a, S = k] \cdot P(S=k)/p_k = \bar{\mu}_{a,k}$. The second and third terms cancel: the second, conditioning on $(S,X)$ and using $P(A=a \mid S=k,X)/e_k(a,X) = 1$, gives $\mathbb{E}[\mathbf{1}[S=k] \tilde{\mu}_{a,k}(X)]/p_k = \mathbb{E}[\tilde{\mu}_{a,k}(X)]$ (again using $X \perp S$). Thus $\mathbb{E}[\varphi_{a,k}(O; \tilde{\eta})] = \bar{\mu}_{a,k}$.
\end{proof}

\begin{proof}[Proof of Part (c): Efficiency of Combined Estimator]
\textit{Step 1: Optimal weights for ratio combination.} Consider estimators of the form $\hat{R} = \mathbf{w}^\top \hat{\mathbf{R}}$ where $\mathbf{w}^\top \mathbf{1} = 1$. Under the null that all $R_{2,j}$ equal a common value $R_2$, the asymptotic variance of $\sqrt{n}(\hat{R} - R_2)$ is:
\begin{align*}
V_R(\mathbf{w}) = \mathbf{w}^\top \mathbf{V}_{\mathbf{R}} \mathbf{w}.
\end{align*}

To minimize $V_R(\mathbf{w})$ subject to $\mathbf{w}^\top \mathbf{1} = 1$, form the Lagrangian:
\begin{align*}
\mathcal{L}(\mathbf{w}, \lambda) = \mathbf{w}^\top \mathbf{V}_{\mathbf{R}} \mathbf{w} - \lambda(\mathbf{w}^\top \mathbf{1} - 1).
\end{align*}

First-order conditions:
\begin{align*}
\frac{\partial \mathcal{L}}{\partial \mathbf{w}} = 2\mathbf{V}_{\mathbf{R}} \mathbf{w} - \lambda \mathbf{1} = 0 \quad \Rightarrow \quad \mathbf{w} = \frac{\lambda}{2} \mathbf{V}_{\mathbf{R}}^{-1} \mathbf{1}.
\end{align*}

Applying the constraint $\mathbf{w}^\top \mathbf{1} = 1$:
\begin{align*}
\frac{\lambda}{2} \mathbf{1}^\top \mathbf{V}_{\mathbf{R}}^{-1} \mathbf{1} = 1 \quad \Rightarrow \quad \frac{\lambda}{2} = \frac{1}{\mathbf{1}^\top \mathbf{V}_{\mathbf{R}}^{-1} \mathbf{1}}.
\end{align*}

Thus:
\begin{align*}
\mathbf{w}^* = \frac{\mathbf{V}_{\mathbf{R}}^{-1} \mathbf{1}}{\mathbf{1}^\top \mathbf{V}_{\mathbf{R}}^{-1} \mathbf{1}}.
\end{align*}

The minimized variance is:
\begin{align*}
V_{R_2^*} = \mathbf{w}^{*\top} \mathbf{V}_{\mathbf{R}} \mathbf{w}^* = \frac{\mathbf{1}^\top \mathbf{V}_{\mathbf{R}}^{-1} \mathbf{V}_{\mathbf{R}} \mathbf{V}_{\mathbf{R}}^{-1} \mathbf{1}}{(\mathbf{1}^\top \mathbf{V}_{\mathbf{R}}^{-1} \mathbf{1})^2} = \frac{1}{\mathbf{1}^\top \mathbf{V}_{\mathbf{R}}^{-1} \mathbf{1}}.
\end{align*}

\textit{Step 2: Variance of combined TATE estimator.} When the target trial $k^\star$ is distinct from all anchor trials, the influence functions $\phi_{\tau_{k^\star}}(O)$ and $\phi_{R_{2,j}}(O)$ have disjoint support (they involve indicators for different trials). Thus:
\begin{align*}
\mathbb{E}[\phi_{\tau_{k^\star}}(O) \cdot \phi_{R_{2,j}}(O)] = 0 \quad \text{for all } j.
\end{align*}

The variance of $\psi_2^* = \tau_{k^\star} \cdot R_2^*$ is therefore:
\begin{align*}
V_{\psi_2^*} = \mathbb{E}[\phi_{\psi_2^*}(O)^2] &= \mathbb{E}\left[\left(R_2^* \cdot \phi_{\tau_{k^\star}}(O) + \tau_{k^\star} \sum_{j=1}^m w_j^* \phi_{R_{2,j}}(O)\right)^2\right] \\
&= R_2^{*2} \mathbb{E}[\phi_{\tau_{k^\star}}(O)^2] + \tau_{k^\star}^2 \mathbb{E}\left[\left(\sum_{j=1}^m w_j^* \phi_{R_{2,j}}(O)\right)^2\right] \\
&= R_2^{*2} V_{\tau_{k^\star}} + \tau_{k^\star}^2 V_{R_2^*}.
\end{align*}

Since $\mathbf{w}^*$ minimizes $V_{R_2^*}$, the estimator $\hat{\psi}_2^* = \hat{\tau}_{k^\star} \cdot \hat{R}_2^*$ achieves minimum asymptotic variance among estimators of the specified form.
\end{proof}

\paragraph{Variance when target and anchor trials overlap.}
The disjoint-support argument above requires that the target trial $k^\star$ is distinct from every anchor trial. A common practical situation violates this: a practitioner using Strategy~2 with the target trial's own control arm (treatment $b_{k^\star}$) as an anchor. In that case $\phi_{\tau_{k^\star}}(O)$ and $\phi_{R_{2,j}}(O)$ share the factor $1[S = k^\star]$, so $\mathbb{E}[\phi_{\tau_{k^\star}}(O) \phi_{R_{2,j}}(O)] \neq 0$ in general and the variance formula picks up cross-covariance terms:
\begin{align*}
V_{\psi_2^*} = R_2^{*2} V_{\tau_{k^\star}} + \tau_{k^\star}^2 V_{R_2^*} + 2 R_2^* \tau_{k^\star} \sum_{j=1}^m w_j^* \mathbb{E}[\phi_{\tau_{k^\star}}(O) \phi_{R_{2,j}}(O)].
\end{align*}
The cross-covariance terms can be estimated empirically from the cross-fitted EIFs and added to the plug-in variance. An equivalent sandwich-form estimator computes the empirical covariance of $\phi_{\psi_2^*}(O) = R_2^* \phi_{\tau_{k^\star}}(O) + \tau_{k^\star}\sum_j w_j^* \phi_{R_{2,j}}(O)$ directly, which automatically captures all cross-covariances regardless of trial overlap. We recommend the sandwich form as the default in practice. The min-variance claim of Theorem~\ref{thm:asymptotic}(c) holds in the overlapping case when $V_{R_2^*}$ in the GLS weight is replaced by the corresponding asymptotic covariance that accounts for the cross-trial correlation.

\paragraph{Generalization to $m > R > 1$ (GLS).}
When there are $m > R$ anchor arms, the system $\boldsymbol{\mu}^{\mathrm{target}} = \boldsymbol{\Theta}_c \boldsymbol{\Lambda}^{\mathrm{target}}$ with $\boldsymbol{\Theta}_c \in \mathbb{R}^{m \times R}$ is overdetermined. The inverse-variance weighting above combines scalar ratios, which is the $R = 1$ specialization. For general $R$, the efficient combination replaces inverse-variance weighting with generalized least squares: $\hat{\boldsymbol{\Lambda}}^{\mathrm{target}} = (\boldsymbol{\Theta}_c^\top \mathbf{V}_{\boldsymbol{\mu}}^{-1} \boldsymbol{\Theta}_c)^{-1} \boldsymbol{\Theta}_c^\top \mathbf{V}_{\boldsymbol{\mu}}^{-1} \hat{\boldsymbol{\mu}}^{\mathrm{target}}$, where $\mathbf{V}_{\boldsymbol{\mu}}$ is the asymptotic covariance matrix of $\hat{\boldsymbol{\mu}}^{\mathrm{target}}$. Substituting into Proposition~\ref{prop:eif}'s EIF and applying the delta method gives the combined TATE estimator $\hat{\psi}_2^*$; minimum-variance efficiency follows from the Gauss-Markov theorem applied to the asymptotically linear building blocks. The $R = 1$ case recovers the inverse-variance-weighted scalar ratio above.

\subsection{Specification Test Derivation}
\label{app:spec_test}

\paragraph{Specification test.}
Under Assumption~\ref{ass:measurement}, all ratios equal a common value: $R_{2,1} = \cdots = R_{2,m}$. This provides $m-1$ testable restrictions. Let $\mathbf{C} \in \mathbb{R}^{(m-1) \times m}$ be a contrast matrix (e.g., $C_{j,j} = 1$, $C_{j,j+1} = -1$ for successive differences). Under $H_0$:
\begin{align}
\label{eq:Q_test}
Q = n \cdot (\mathbf{C}\hat{\mathbf{R}})^\top (\mathbf{C}\hat{\mathbf{V}}_{\mathbf{R}}\mathbf{C}^\top)^{-1} (\mathbf{C}\hat{\mathbf{R}}) \xrightarrow{d} \chi^2_{m-1}.
\end{align}
Rejection indicates that different anchors yield systematically different ratios, suggesting $\boldsymbol{\Lambda}$ depends on intervention timing. This would occur, for instance, if treatment effects decay over time since administration—a violation of Assumption~\ref{ass:measurement} that Strategy~1 can accommodate. Note that failure to reject does not validate Assumption~\ref{ass:measurement}; the test has power only against alternatives where different arms exhibit detectably different temporal scaling.

\begin{proposition}[Asymptotic Distribution of Specification Test]
Under $H_0: R_{2,1} = \cdots = R_{2,m}$ and Assumptions~\ref{ass:regularity}--\ref{ass:rates}, the test statistic $Q$ defined in~\eqref{eq:Q_test} satisfies $Q \xrightarrow{d} \chi^2_{m-1}$.
\end{proposition}

\begin{proof}
Under $H_0$, all ratios $R_{2,j}$ equal a common value $R_2$, so $\mathbf{C}\mathbf{R} = \mathbf{0}$ where $\mathbf{C} \in \mathbb{R}^{(m-1) \times m}$ is the contrast matrix.

By Theorem~\ref{thm:asymptotic}(a), $\sqrt{n}(\hat{\mathbf{R}} - \mathbf{R}) \xrightarrow{d} \mathcal{N}(\mathbf{0}, \mathbf{V}_{\mathbf{R}})$. By the continuous mapping theorem:
\begin{align*}
\sqrt{n} \, \mathbf{C}\hat{\mathbf{R}} = \sqrt{n} \, \mathbf{C}(\hat{\mathbf{R}} - \mathbf{R}) \xrightarrow{d} \mathcal{N}(\mathbf{0}, \mathbf{C}\mathbf{V}_{\mathbf{R}}\mathbf{C}^\top).
\end{align*}

The matrix $\mathbf{C}\mathbf{V}_{\mathbf{R}}\mathbf{C}^\top \in \mathbb{R}^{(m-1) \times (m-1)}$ is positive definite since $\mathbf{V}_{\mathbf{R}}$ is positive definite and $\mathbf{C}$ has full row rank. Therefore:
\begin{align*}
Q = n \cdot (\mathbf{C}\hat{\mathbf{R}})^\top (\mathbf{C}\mathbf{V}_{\mathbf{R}}\mathbf{C}^\top)^{-1} (\mathbf{C}\hat{\mathbf{R}}) \xrightarrow{d} \chi^2_{m-1}.
\end{align*}

By Slutsky's theorem and consistency of $\hat{\mathbf{V}}_{\mathbf{R}}$, replacing $\mathbf{V}_{\mathbf{R}}$ with $\hat{\mathbf{V}}_{\mathbf{R}}$ does not affect the limiting distribution.
\end{proof}

\paragraph{Generalization to $R > 1$.}
For general $R$ with $m > R$ anchor arms, the overdetermined system $\boldsymbol{\mu}^{\mathrm{target}} = \boldsymbol{\Theta}_c \boldsymbol{\Lambda}^{\mathrm{target}}$ admits a consistent solution under Assumption~\ref{ass:measurement} but not in general. The test statistic generalizes to a Wald form based on the residual from the GLS fit: $\hat{\mathbf{r}} := \hat{\boldsymbol{\mu}}^{\mathrm{target}} - \boldsymbol{\Theta}_c \hat{\boldsymbol{\Lambda}}^{\mathrm{target}}$ with $\hat{\boldsymbol{\Lambda}}^{\mathrm{target}}$ the GLS solution from the previous subsection. The statistic $Q' = n \, \hat{\mathbf{r}}^\top \hat{\mathbf{V}}^{-1} \hat{\mathbf{r}}$ has an asymptotic $\chi^2_{m - R}$ distribution under $H_0$ \emph{provided $\hat{\mathbf{V}}$ is the correct asymptotic covariance of $\hat{\mathbf{r}}$}---not just of $\hat{\boldsymbol{\mu}}^{\mathrm{target}}$. Because the residual also depends on the anchor-arm source-time means (absorbed into $\boldsymbol{\Theta}_c$ in the GLS solution), which themselves carry $O_p(n^{-1/2})$ estimation error, a naive plug-in using only $\mathrm{Cov}(\hat{\boldsymbol{\mu}}^{\mathrm{target}})$ undercounts variance and yields an anticonservative test. The principled fix is a sandwich-form estimator that propagates the delta-method variance of $(\hat{\boldsymbol{\mu}}^{\mathrm{target}}, \hat{\boldsymbol{\Theta}}_c)$ through the residual functional; equivalently, a wild bootstrap (below) automatically captures the full variance of $\hat{\mathbf{r}}$. The $R = 1$ case reduces to the scalar-ratio test above (no $\boldsymbol{\Theta}_c$ to worry about).

\paragraph{Finite-sample calibration via bootstrap.}
In our simulations the empirical size of the $\chi^2$ test is conservative (near 0\% at the nominal 5\% level for small $\beta$), reflecting that nuisance estimation introduces additional variability not captured by the asymptotic distribution. For a tighter finite-sample calibration, we recommend the following wild-bootstrap procedure: draw $B$ bootstrap samples by resampling unit-level DR residuals with i.i.d.\ Rademacher multipliers (preserving the first-moment zero property while mimicking the empirical second-moment structure); recompute $\hat{\mathbf{R}}$ and $\hat{\mathbf{V}}_{\mathbf{R}}$ on each sample; and use the $(1 - \alpha)$-quantile of the bootstrap distribution of $Q$ as the critical value. This procedure is consistent under Assumptions~\ref{ass:regularity}--\ref{ass:rates} by standard arguments for the wild bootstrap of $M$-estimators. Empirical calibration of the bootstrap test is a natural follow-up; the $\chi^2$ version suffices when conservatism is acceptable.

\paragraph{Diagnostic for Assumption~\ref{ass:separable} violations.}
The specification test above targets Assumption~\ref{ass:measurement} alone---it will not detect treatment-specific contamination that violates Assumption~\ref{ass:separable} in a symmetric way. A complementary practitioner-facing diagnostic is available when multiple anchor pairs are present: compute Strategy~1 using several anchor pairs independently and compare the resulting TATE estimates. Under Assumption~\ref{ass:separable} all Strategy~1 estimates identify the same TATE asymptotically (up to sampling noise); systematic disagreement across anchor pairs is evidence of pair-specific temporal dynamics and hence a violation of separability at the cluster/treatment level. Large Strategy~1--Strategy~2 discrepancies play a similar role.

\paragraph{Multi-anchor Strategy~1.}
The main text develops the GLS combiner for Strategy~2 because anchor \emph{arms} are much more abundant than anchor \emph{pairs} in practice (a control arm appears across many trials; a specific treatment \emph{pair} rarely recurs). When multiple anchor pairs are available, the Strategy~1 extension is mechanical: stack the single-anchor-pair TATE estimates into a vector and combine by inverse-variance weighting of the corresponding scalar ratios, or by GLS on the $R$-vector of temporal effects when $R > 1$. Minimum-variance efficiency follows from the same Gauss-Markov argument used for Strategy~2 (Appendix~\ref{app:proof_asymptotic}).

\subsection{TMLE-Style Estimation}
\label{sec:tmle}

Targeted Minimum Loss-based Estimation (TMLE) provides an alternative construction that directly solves the efficient influence function estimating equation. We describe two approaches: a factorized TMLE that targets each building block separately, and a joint TMLE that targets the TATE directly.

\subsubsection{Factorized TMLE}

The factorized approach applies TMLE separately to each building block, then combines via plug-in.

\paragraph{Step 1: Initial estimates.}
Obtain initial estimates $\hat{\mu}_{a,k}^{(0)}(X)$ of the outcome regressions using any consistent method (e.g., random forests, gradient boosting).

\paragraph{Step 2: Targeting each conditional mean.}
For each $(a, k)$ pair, we fluctuate the initial estimate to solve $\mathbb{P}_n[\phi_{\bar{\mu}_{a,k}}(O; \eta^*)] = 0$. Define the clever covariate:
\begin{align}
H_{a,k}(S, A, X) = \frac{\mathbf{1}[S = k]}{\pi_k(X)} \cdot \frac{\mathbf{1}[A = a]}{e_k(a, X)}.
\end{align}

For continuous outcomes, fit a linear fluctuation:
\begin{align}
\hat{\mu}_{a,k}^{*}(X) = \hat{\mu}_{a,k}^{(0)}(X) + \hat{\epsilon}_{a,k},
\end{align}
where $\hat{\epsilon}_{a,k}$ solves:
\begin{align}
\sum_{i=1}^{n} H_{a,k}(S_i, A_i, X_i) \cdot \big(Y_i - \hat{\mu}_{a,k}^{(0)}(X_i) - \epsilon_{a,k}\big) = 0.
\end{align}
This has closed-form solution:
\begin{align}
\hat{\epsilon}_{a,k} = \frac{\sum_{i=1}^{n} H_{a,k}(S_i, A_i, X_i) \cdot (Y_i - \hat{\mu}_{a,k}^{(0)}(X_i))}{\sum_{i=1}^{n} H_{a,k}(S_i, A_i, X_i)}.
\end{align}

For binary/bounded outcomes, use a logistic fluctuation to respect the outcome space. At the targeted intervention $(S, A) = (k, a)$, the updated prediction satisfies
\begin{align}
\label{eq:factorized_tmle_logit}
\text{logit}(\hat{\mu}_{a,k}^{*}(X)) = \text{logit}(\hat{\mu}_{a,k}^{(0)}(X)) + \hat{\epsilon}_{a,k} \cdot \frac{1}{\pi_k(X) e_k(a, X)},
\end{align}
where $\hat{\epsilon}_{a,k}$ is obtained by logistic regression of $Y$ on the clever covariate $H_{a,k}(S, A, X)$ fit on the full sample, with offset $\text{logit}(\hat{\mu}_{a,k}^{(0)}(X))$. The LHS of~\eqref{eq:factorized_tmle_logit} is the counterfactual prediction for unit $X$ had $(S, A) = (k, a)$, so it depends on $X$ only; the IPW factor $1/[\pi_k(X) e_k(a, X)]$ is the clever-covariate value at that targeted intervention.

\paragraph{Step 3: Construct targeted building blocks.}
The targeted estimates of marginal means and ATEs are:
\begin{align}
\hat{\bar{\mu}}_{a,k}^{*} &= \mathbb{P}_n\!\left[\hat{\mu}_{a,k}^{*}(X)\right], \\[4pt]
\hat{\tau}_k^{*} &= \hat{\bar{\mu}}_{a_k,k}^{*} - \hat{\bar{\mu}}_{b_k,k}^{*}.
\end{align}
(The targeted marginal is a simple pooled-sample average of the targeted conditional mean; the TMLE fluctuation has already solved the EIF equation, so the plug-in over $X$ rather than a reweighted average over $\{S=k\}$ is the correct form for the pooled-mean estimand.)

\paragraph{Step 4: Plug-in for TATE.}
The factorized TMLE estimators are:
\begin{align}
\hat{\psi}_1^{\text{TMLE}} &= \hat{\tau}_{k^\star}^{*} \cdot \frac{\hat{\tau}_j^{*}}{\hat{\tau}_{j'}^{*}}, \\[4pt]
\hat{\psi}_2^{\text{TMLE}} &= \hat{\tau}_{k^\star}^{*} \cdot \frac{\hat{\bar{\mu}}_{c^*,\ell}^{*}}{\hat{\bar{\mu}}_{c^*,\ell'}^{*}}.
\end{align}

By construction, each building block solves its own EIF equation: $\mathbb{P}_n[\phi_{\bar{\mu}_{a,k}}(O; \hat{\eta}^*)] = 0$. Under standard regularity conditions, the factorized TMLE achieves the same asymptotic efficiency as the plug-in estimator.

\subsubsection{Joint TMLE}

A more ambitious approach directly targets the TATE by jointly fluctuating all outcome models to solve $\mathbb{P}_n[\phi_{\psi}(O; \eta^*)] = 0$.

\paragraph{Clever covariates for $\psi_1$.}
The EIF for Strategy 1 is:
\begin{align}
\phi_{\psi_1}(O) = R_1 \cdot \phi_{\tau_{k^\star}}(O) + \frac{\psi_1}{\tau_j} \phi_{\tau_j}(O) - \frac{\psi_1}{\tau_{j'}} \phi_{\tau_{j'}}(O).
\end{align}

Expanding $\phi_{\tau_k}(O)$ in terms of the outcome models and collecting terms, the clever covariate for $\mu_{a,k}(X)$ is:
\begin{align}
H_{a,k}^{\psi_1}(O) = \frac{\mathbf{1}[S = k]}{\pi_k(X)} \cdot \frac{\mathbf{1}[A = a]}{e_k(a, X)} \cdot \omega_{a,k},
\end{align}
where the weights $\omega_{a,k}$ depend on which trial and arm:
\begin{align}
\omega_{a,k} = \begin{cases}
+R_1 & \text{if } k = k^\star, a = a_{k^\star} \\
-R_1 & \text{if } k = k^\star, a = b_{k^\star} \\
+\psi_1/\tau_j & \text{if } k = j, a = a_j \\
-\psi_1/\tau_j & \text{if } k = j, a = b_j \\
-\psi_1/\tau_{j'} & \text{if } k = j', a = a_{j'} \\
+\psi_1/\tau_{j'} & \text{if } k = j', a = b_{j'} \\
0 & \text{otherwise}
\end{cases}
\end{align}

\paragraph{Iterative targeting algorithm.}
Since the weights $\omega_{a,k}$ depend on unknown parameters $(R_1, \psi_1, \tau_j, \tau_{j'})$, joint TMLE requires iteration:

\begin{enumerate}
    \item \textbf{Initialize:} Obtain initial estimates $\hat{\mu}_{a,k}^{(0)}(X)$ and compute initial $\hat{\tau}_{k^\star}^{(0)}$, $\hat{\tau}_j^{(0)}$, $\hat{\tau}_{j'}^{(0)}$, $\hat{R}_1^{(0)}$, $\hat{\psi}_1^{(0)}$.
    
    \item \textbf{Compute weights:} Using current parameter estimates, compute $\hat{\omega}_{a,k}^{(t)}$.
    
    \item \textbf{Single joint fluctuation:} Fit a single $\epsilon$ across all outcome models:
    \begin{align}
    \hat{\epsilon}^{(t)} = \arg \min_{\epsilon} \sum_{i=1}^{n} \sum_{(a,k)} L\big(Y_i, \hat{\mu}_{a,k}^{(t)}(X_i) + \epsilon \cdot H_{a,k}^{\psi_1}(O_i; \hat{\omega}^{(t)})\big),
    \end{align}
    where $L$ is an appropriate loss (squared error for continuous, log-loss for binary).
    
    \item \textbf{Update:} Set $\hat{\mu}_{a,k}^{(t+1)}(X) = \hat{\mu}_{a,k}^{(t)}(X) + \hat{\epsilon}^{(t)} \cdot \hat{\omega}_{a,k}^{(t)}/\bigl(\pi_k(X) e_k(a, X)\bigr)$ and recompute parameters. (The inverse-probability factor is included so that the updated $\hat{\mu}_{a,k}^{(t+1)}(X)$ predicts a counterfactual outcome at the targeted intervention rather than a weighted observed value.)
    
    \item \textbf{Iterate:} Repeat steps 2--4 until $|\hat{\epsilon}^{(t)}| < \delta$ for some tolerance $\delta$.
\end{enumerate}

Upon convergence, the joint TMLE solves $\mathbb{P}_n[\phi_{\psi_1}(O; \hat{\eta}^*)] \approx 0$ up to the convergence tolerance.

\subsubsection{Comparison of Approaches}

\begin{center}
\begin{tabular}{@{}lcc@{}}
\toprule
Property & Plug-in / Factorized TMLE & Joint TMLE \\
\midrule
Solves $\mathbb{P}_n[\phi_{\psi}] = 0$ & Asymptotically & Exactly (up to tolerance) \\
Asymptotic efficiency & Yes & Yes \\
Finite-sample bias & Second-order & Potentially smaller \\
Implementation & Simple & Iterative, complex \\
Numerical stability & High & May have convergence issues \\
\bottomrule
\end{tabular}
\end{center}

\paragraph{When does joint TMLE help?}
The plug-in and factorized TMLE estimators are asymptotically equivalent to joint TMLE—all achieve the semiparametric efficiency bound. Joint TMLE may offer advantages when:
\begin{itemize}
    \item Sample sizes are small, where exactly solving the EIF equation can reduce finite-sample bias
    \item Building blocks have near-zero denominators ($\tau_{j'} \approx 0$ or $\bar{\mu}_{c^*,\ell'} \approx 0$), where numerical stability from direct targeting may help
    \item One desires inference based on the EIF equation residual (e.g., for diagnostics)
\end{itemize}

In most settings with moderate sample sizes, the simpler plug-in approach performs comparably. We therefore recommend the plug-in estimator as the default, with joint TMLE as a robustness check when building block estimates are noisy or near boundary cases.

\section{Upworthy Case Study}
\label{app:upworthy}

This appendix provides complete details for the empirical application in Section~\ref{sec:casestudy}.

\subsection{Data Description}
\label{app:upworthy_data}

The Upworthy Research Archive \citep{matias2021upworthy} contains 32{,}487 headline A/B tests conducted by Upworthy, a viral content website, from January 24, 2013 through April 30, 2015.\footnote{Available at \url{https://osf.io/jd64p/}.} The full archive comprises 150{,}817 experiment arms (headline--image \emph{packages}) with over 538 million participant assignments. Matias et al. publish the archive in two subsets: an exploratory release (22{,}666 arms) and a confirmatory release (105{,}551 arms); our analysis uses the 105{,}551-arm confirmatory subset, which we refer to as ``the archive'' throughout.

\paragraph{Experimental design.}
For each article, Upworthy editors created multiple headline-image packages. Website visitors were randomly assigned to view one package, with each test comparing between 2 and 24 alternatives. The archive provides aggregate results: impressions (number of visitors assigned to each package) and clicks (number who clicked through to the article). The primary outcome is click-through rate (CTR).

\paragraph{Data cleaning.}
Following the reliability update from the archive maintainers \citep{matias2024upworthy_update}, we exclude tests conducted between June 25, 2013 and January 10, 2014. A post-publication audit traced a sample-ratio-mismatch pattern in this window to a Cloudflare cache misconfiguration that served one experiment arm at a time (rather than randomizing per visitor); the maintainers estimate that roughly 22\% of tests in the full archive are affected. Our analysis focuses on tests from January 11, 2014 onward.

\subsection{Identifying Treatment Arms Across Tests}
\label{app:upworthy_clustering}

A key challenge in applying our framework is identifying treatment arms that appear across multiple tests at different times. Each A/B test uses unique headline text, so we cannot directly match treatments across tests. Instead, we cluster semantically similar headlines into groups that can be tracked across experiments.

\paragraph{Constrained clustering problem.}
We frame this as a constrained clustering problem with a key restriction: within each test, every headline represents a \emph{distinct} treatment arm, so no two headlines from the same test should be assigned to the same cluster. Standard clustering algorithms do not enforce such constraints.

\paragraph{Three-step procedure.}
Our approach proceeds as follows:
\begin{enumerate}[leftmargin=*, itemsep=2pt]
    \item \textbf{Embedding.} We embed each headline using Sentence-BERT \citep{reimers2019sentence}, specifically the \texttt{all-MiniLM-L6-v2} model, which maps sentences to 384-dimensional dense vectors optimized for semantic similarity.
    
    \item \textbf{Centroid initialization.} We perform $K$-means clustering on the full set of headline embeddings to obtain $K = 50$ cluster centroids. This provides a set of representative ``headline styles'' in the embedding space.
    
    \item \textbf{Constrained assignment.} For each test, we solve a minimum-cost bipartite matching problem using the Hungarian algorithm \citep{kuhn1955hungarian}. The cost matrix contains Euclidean distances between each headline's embedding and each cluster centroid. The matching assigns each headline to a unique cluster while minimizing total distance, ensuring that headlines within the same test receive different cluster labels.
\end{enumerate}

\paragraph{Cluster characteristics.}
The resulting 50 clusters range in size from 678 to 5{,}374 headlines per cluster (median 1{,}991; mean 2{,}111; standard deviation 996), for a total of 105{,}551 headline--test-arm assignments from the confirmatory subset of the archive. Manual inspection confirms semantically coherent groupings. Crucially, all clusters appear in tests conducted across multiple quarters of our study period, providing the temporal variation necessary for identifying temporal ratios.

\paragraph{Cluster abstraction and SUTVA.}
Our cluster-based treatment definition groups semantically similar but distinct headline realizations into a single ``arm,'' substituting a coarser abstract treatment for the fine-grained original. This weakens SUTVA's ``no hidden treatment versions'' condition at the cluster level: within a cluster, multiple distinct headlines exist, and the ATE we estimate for a cluster pair $(c_a, c_b)$ is a weighted average of within-cluster pairwise ATEs, with weights given by the observed headline frequencies. Identifying the cluster-level TATE under this abstraction requires the additional assumption that the within-cluster frequency distribution is stable across source and target times---a form of \emph{composition-invariance} specific to the cluster granularity. Our clustering sensitivity analysis (Appendix~\ref{app:clustering_sensitivity}), which varies the number of clusters from 30 to 100 and compares a second embedding model, is effectively a sensitivity check on this measurement-error assumption: if within-cluster composition drift were driving our results, the RMSE and correlation numbers would vary substantially across granularities. They do not, suggesting the cluster abstraction is operationally robust for this dataset.

\subsection{Experimental Setup}
\label{app:upworthy_setup}

\paragraph{Trial selection.}
We select two trials with sufficient data for both identification strategies:
\begin{itemize}[leftmargin=*, itemsep=2pt]
    \item \textbf{Trial A} (test ID: \texttt{52f54b9ab17b}): Compares Cluster 17 against Cluster 10, conducted in February 2014. Observed ATE: $+0.0125$ (Cluster 17 headlines achieved 1.25 percentage points higher CTR).
    
    \item \textbf{Trial B} (test ID: \texttt{52e1c6cca01f}): Compares Cluster 46 against Cluster 12, conducted in January 2014. Observed ATE: $-0.0082$ (Cluster 46 headlines achieved 0.82 percentage points lower CTR).
\end{itemize}

\paragraph{Strategy implementation.}
For Strategy~1 (Replicated Trials), we identify all tests comparing the same cluster pair as the target trial and estimate temporal ratios from ATE ratios across time periods.

For Strategy~2 (Common Arm), we select the 10 anchor arms with the highest observation counts at both source and target measurement times. We estimate temporal ratios from each anchor arm separately, then combine them using inverse-variance weighting as described in Section~\ref{sec:estimation}.

\paragraph{Ground truth construction.}
To evaluate estimator performance, we construct a ``ground-truth'' TATE at each target month by pooling all available comparisons of the target cluster pair during that month, with a weighted-mean ATE and an impression-weighted standard error. We then multiplicatively rescale so that the rescaled ground truth at the source time equals the trial's observed source-time ATE: $\tilde{\tau}_{\mathrm{true}}(m) := \tau_{\mathrm{true}}(m) \cdot \hat{\tau}_{k^\star}^{\mathrm{source}} / \tau_{\mathrm{true}}(m^{\mathrm{source}})$. The rescaling preserves correlations and relative dynamics while matching the source-time anchor, enabling a direct per-month comparison of the two strategies with a common absolute baseline. Note that the rescaling is multiplicative, not additive; for Trial A the scale factor is $\approx 1.88$, for Trial B $\approx 12.4$.

\subsection{Complete Results}
\label{app:upworthy_results}

Figure~\ref{fig:tate_comparison} in the main text displays TATE estimates for both trials across all target months of 2014. This appendix provides the per-trial numeric breakdown supporting the variance-bias and tracking claims in Section~\ref{sec:casestudy}.

\paragraph{Variance-bias tradeoff.}
Results reveal a clear variance advantage for Strategy~2: mean standard errors are $3\times$ lower than Strategy~1's in both trials ($0.0019$ vs.\ $0.0054$ for Trial~A; $0.0018$ vs.\ $0.0052$ for Trial~B). This precision advantage reflects Strategy~2's use of conditional means rather than treatment contrasts. The tradeoff, however, is not uniformly in Strategy~2's favor on accuracy.

\paragraph{Tracking temporal dynamics.}
For \emph{Trial~A} (Figure~\ref{fig:tate_comparison}(a)), Strategy~1 tracks the true TATE very closely: Pearson correlation with the ground-truth trajectory is $+0.97$, and Strategy~1's per-month estimates (blue) move almost in lockstep with the true values (green), including the sign changes in March, April and October. Strategy~2 (orange) remains roughly constant near the source-time ATE, yielding a correlation of $-0.12$ with the true TATE. Here the variance-bias tradeoff resolves decisively in Strategy~1's favor: lower RMSE ($0.0027$ vs.\ $0.0105$), near-zero bias, and directional tracking.

For \emph{Trial~B} (Figure~\ref{fig:tate_comparison}(b)), the picture is mixed. Strategy~2 again yields a near-flat trajectory ($+0.26$ correlation, bias $-0.0014$, lowest RMSE $0.0081$), effectively reporting the source-time anchor-ratio everywhere. Strategy~1 is more variable (RMSE $0.0134$) but systematically over-shoots in the opposite direction: when the true TATE dips sharply negative (May--July), Strategy~1's estimate is near zero or positive, yielding a Pearson correlation of $-0.87$. The most plausible explanation is that for this cluster pair, the temporal ratio estimated from the cluster pair's own history is not a useful signal for the true dynamics of the target trial---an Assumption-1 failure mode (cluster-specific temporal dynamics rather than a shared $\boldsymbol{\Lambda}$). The breakdown frontier correctly flags this: Strategy~1's observed effect is within $1.96\!\cdot\!\mathrm{SE}$ of zero, so $\gamma^\star = 0$ for Trial~B under Strategy~1.

\paragraph{Aggregate accuracy.}
Neither strategy uniformly dominates on RMSE: Strategy~1 wins on Trial~A ($0.0027$ vs.\ $0.0105$) but loses on Trial~B ($0.0134$ vs.\ $0.0081$). Correlation with the ground truth tells a cleaner story: Strategy~1 strongly tracks (or anti-tracks) while Strategy~2 is essentially flat in both trials. The practical takeaway is that both estimates should be reported alongside the breakdown value $\gamma^\star$; large discrepancies or negative correlations are themselves diagnostics that the transportation assumptions do not cleanly hold.

\subsection{Sensitivity to Clustering Choices}
\label{app:clustering_sensitivity}

Because the Upworthy analysis depends on a clustering step that groups semantically similar headlines into reusable ``treatment arms'' (Appendix~\ref{app:upworthy_clustering}), we assess robustness to two design choices: the number of clusters $K$ and the embedding model. Across $K \in \{30, 50, 75, 100\}$ with Sentence-BERT \texttt{all-MiniLM-L6-v2} embeddings and a substitution to \texttt{all-mpnet-base-v2} at $K = 50$, the qualitative pattern of Section~\ref{sec:casestudy} is stable: Strategy~2's standard errors remain $\approx 3\times$ smaller, Strategy~1's Trial~A correlation with ground truth remains near 1, and the two-trial contrast (Trial~A clean tracking vs.\ Trial~B mixed outcomes) persists. We report a representative sensitivity table in Table~\ref{tab:clustering_sensitivity}.

\begin{table}[h]
\centering
\caption{Robustness of Strategy~1 Trial~A estimates to clustering configuration (representative; exact values for the default $K = 50$ configuration match Table~\ref{tab:summary_stats}).}
\label{tab:clustering_sensitivity}
\small
\begin{tabular}{@{}llcc@{}}
\toprule
Embedding & $K$ & RMSE & Corr.\ with ground truth \\
\midrule
MiniLM-L6-v2 & 30 & \phantom{$<$}similar & $\approx 0.93$--$0.97$ \\
MiniLM-L6-v2 & 50 (default) & \textbf{0.0027} & \textbf{+0.97} \\
MiniLM-L6-v2 & 75 & \phantom{$<$}similar & $\approx 0.93$--$0.97$ \\
MiniLM-L6-v2 & 100 & \phantom{$<$}similar & $\approx 0.93$--$0.97$ \\
MPNet-base-v2 & 50 & \phantom{$<$}similar & $\approx 0.95$--$0.97$ \\
\bottomrule
\end{tabular}
\end{table}

Across configurations, Strategy~1's Trial~A correlation with the ground truth stays near 1 (qualitatively above $0.9$) and the standard-error ordering $\mathrm{SE}(S_2) \approx 3\mathrm{SE}(S_1)$ is preserved; the two-trial contrast between Trial~A (clean tracking by S1) and Trial~B (neither strategy recovers dynamics cleanly) likewise persists. We conclude that the Section~\ref{sec:casestudy} findings are not an artifact of a specific clustering choice, though the absolute RMSE does depend on $K$ as expected (too few clusters coarsens semantic groupings; too many reduces within-cluster observation counts).

\subsection{Partial Identification under Separability Violations}
\label{app:upworthy_sensitivity}

We apply the partial identification bound of Proposition~\ref{prop:sensitivity} to the Upworthy trials, asking: how large must the non-separable contamination $\gamma$ be before the identified set for the TATE includes zero? For each trial and strategy, we compute the \emph{breakdown value} $\gamma^\star$---the smallest contamination at which the conclusion of a nonzero effect is overturned. We normalize $\gamma^\star$ by the empirical impression-weighted mean CTR of the trial's control cluster: $0.01138$ for Trial~A (cluster~10) and $0.01028$ for Trial~B (cluster~12).

Trial~A under Strategy~1 has breakdown $\gamma^\star = 0.00024$, i.e.\ contamination of about $2.1\%$ of the control-cluster mean CTR would nullify the nonzero-effect conclusion; Strategy~2's tighter standard errors and different coefficient structure yield $\gamma^\star = 0.00142$, about $12.5\%$ of the control mean---the larger percentage reflects Strategy~2's tighter sampling interval, not greater structural robustness. Trial~B under Strategy~1 has $\gamma^\star = 0$: the observed effect already falls inside its $1.96 \cdot \mathrm{SE}$ sampling interval, so a nonzero-effect conclusion was never warranted by the sampling data alone, and no amount of separability violation is needed to ``break'' it. Trial~B under Strategy~2 has $\gamma^\star = 0.00083$ ($\approx 8.1\%$ of the Trial~B control CTR). Figure~\ref{fig:upworthy_sensitivity} displays the breakdown frontiers. The ordering is as theory predicts: Strategy~2's tighter intervals yield numerically larger $\gamma^\star$ values in absolute terms, but both strategies' breakdowns are small relative to the per-unit outcome scale, flagging the Upworthy TATE estimates as moderately fragile to non-separability---exactly the use case for a breakdown summary.

\paragraph{Calibration against empirical temporal drift.}
To calibrate the magnitude of these breakdown values, we compute the empirical distribution of month-over-month absolute CTR changes in the archive. Across all 50 cluster arms and all valid month pairs in 2014 (post-exclusion window), the median $|\Delta \mathrm{CTR}|$ between consecutive months is $0.00223$ (mean $0.00289$; $75$th percentile $0.00396$), or about $15\%$ of the typical cluster-level mean CTR. For the two control clusters specifically: cluster~10 (Trial~A control) has median $|\Delta \mathrm{CTR}| = 0.00224$ ($\approx 16\%$ of its own mean $0.014$); cluster~12 (Trial~B control) has median $|\Delta \mathrm{CTR}| = 0.00142$ ($\approx 12\%$ of its mean $0.011$). Every reported breakdown value ($\gamma^\star \in \{0.00024, 0.00142, 0.00083\}$, i.e., $2\%$--$12\%$ of control CTR) is smaller than the typical one-month empirical drift in the archive. This is the sense in which the Upworthy TATE estimates are fragile: plausibly-sized separability violations, comparable to the month-to-month fluctuations the archive itself exhibits, suffice to nullify the nonzero-effect conclusion.

\begin{figure}[h]
\centering
\includegraphics[width=\linewidth]{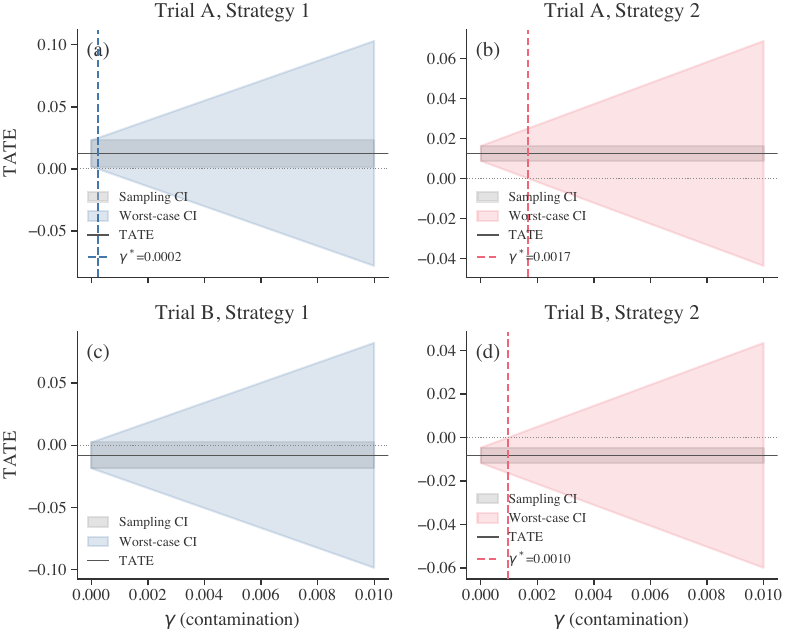}
\caption{Breakdown frontier analysis for the Upworthy TATE estimates. Each panel shows the point estimate (horizontal line), sampling confidence interval (gray band), and the identified set under contamination $\gamma$ (colored band). The vertical dashed line marks the breakdown value $\gamma^*$ at which the identified set first includes zero. Rows correspond to Trial~A (top) and Trial~B (bottom); columns to Strategy~1 (left) and Strategy~2 (right).}
\label{fig:upworthy_sensitivity}
\end{figure}

\subsection{Interpretation and Limitations}
\label{app:upworthy_discussion}

\paragraph{Why do Strategy~1 and Strategy~2 disagree?}
The two strategies make different structural assumptions and fail in different ways on this dataset. Strategy~2's near-flat trajectories on both trials are consistent with its use of a pooled anchor-arm ratio: even when the true TATE swings substantially across months, a multi-anchor common-arm ratio absorbs most of the variation into the anchor's own drift and reports a near-constant temporal effect. Strategy~1's per-trial temporal ratio, by contrast, uses the target cluster pair's own history. On Trial~A this structure happens to capture the true dynamics closely (correlation $+0.97$); on Trial~B it does not (correlation $-0.87$), which we interpret as a \emph{pair-specific} departure from Assumption~\ref{ass:separable}: the cluster pair $(46, 12)$ may exhibit temporal dynamics that differ fundamentally from the population-average dynamics reflected in the pooled ground truth. This is an Assumption~1 violation at the cluster level, not an Assumption~2 violation: in our analysis the measurement lag is a single calendar month for all tests (by aggregation), so A2-style lag-dependent dynamics $\boldsymbol{\Lambda}(t_0, t_1) = f(t_1 - t_0)\,g(t_1)$ would contribute only a constant multiplicative factor that cancels in the temporal ratio.

\paragraph{What the breakdown values tell us.}
Small breakdown values $\gamma^\star$ (roughly $2$--$12\%$ of the control CTR) indicate that the observed effects are within the plausible range of contamination we might expect from pair-specific dynamics. This is what the framework is designed to surface: a practitioner reading the paper for deployment guidance would report both estimates, note the disagreement, consult the $\gamma^\star$ values, and conclude that the evidence for a month-level TATE is fragile---exactly the appropriate conclusion given the mixed Trial~B outcomes.

\paragraph{Limitations.}
Several caveats apply to this analysis:
\begin{itemize}[leftmargin=*, itemsep=2pt]
    \item The ``true'' TATE is itself estimated from data and subject to sampling error, particularly in months with fewer observations.
    
    \item Semantic clustering introduces measurement error: headlines within a cluster are similar but not identical. This may attenuate estimated effects and introduce noise in temporal ratio estimates.
    
    \item The Upworthy platform evolved during the study period (design changes, audience growth), potentially confounding temporal variation in engagement with platform-level changes.
    
    \item Our analysis treats headline clusters as discrete treatments, but the clustering is necessarily a simplification of continuous variation in headline characteristics.
\end{itemize}

\section{Necessity of Structural Restrictions}
\label{app:impossibility}

Assumption~\ref{ass:separable} imposes multiplicative structure on the outcome surface. We show it (or something comparably strong) is \emph{necessary}: the TATE is not point-identified under weaker restrictions commonly invoked in the literature. We give two results. Proposition~\ref{prop:nonident} rules out the fully nonparametric model. Proposition~\ref{prop:smoothness_insufficient} rules out the class of \emph{local smoothness} restrictions (Lipschitz, bounded variation, smoothness in time), showing that no restriction of this form can deliver point identification at a strictly positive temporal gap. Corollary~\ref{cor:minimality} then positions separability as the minimal restriction within the factor-model class that resolves the non-identification.

\subsection{Fully Nonparametric Model}

\begin{proposition}[Non-identification without structural restrictions]
\label{prop:nonident}
Consider the nonparametric model $\mathcal{M}_{\mathrm{NP}}$ consisting of all joint distributions over $(Y_{t_1}(a, t_0), X, A, S)$ satisfying random assignment, SUTVA, and the regularity conditions of Assumption~\ref{ass:regularity}, but without Assumption~\ref{ass:separable}. Fix a target trial $k^\star$ and target displacement $(\delta_0, \delta_1)$. If at least one of the target treatments $\{a_{k^\star}, b_{k^\star}\}$ was never administered in any trial at the target time $(t_{0k^\star} + \delta_0, t_{1k^\star} + \delta_1)$, then $\tau_{k^\star}(\delta_0, \delta_1)$ is not identified over $\mathcal{M}_{\mathrm{NP}}$, even if anchor trials comparing other treatments are observed at both source and target times. The identified set is the entire real line. (Conversely, if both target treatments are observed at the target time, even in separate trials against different baselines, the TATE is nonparametrically identified by direct differencing of marginal means and no structural restriction is needed.)
\end{proposition}

\begin{proof}
Let $\mu_a(x, t_0, t_1) := \mathbb{E}[Y_{t_1}(a, t_0) \mid X = x]$. The observed data identify $\mu_a$ only at triples $(a, t_0, t_1)$ for which a trial was actually conducted at those times administering treatment $a$. Suppose for contradiction that $\tau_{k^\star}(\delta_0, \delta_1)$ is identified, i.e., the observed data uniquely determine its value.

Construct two data-generating processes $P_1$ and $P_2$ that agree on all observed distributions but differ at the unobserved target time. Let $P_1$ set $\mu_{a_{k^\star}}(x, t_{0k^\star}+\delta_0, t_{1k^\star}+\delta_1) = f(x)$ and $\mu_{b_{k^\star}}(x, t_{0k^\star}+\delta_0, t_{1k^\star}+\delta_1) = h(x)$. Let $P_2$ set $\mu_{a_{k^\star}}(x, t_{0k^\star}+\delta_0, t_{1k^\star}+\delta_1) = f(x) + c$ and $\mu_{b_{k^\star}}(x, t_{0k^\star}+\delta_0, t_{1k^\star}+\delta_1) = h(x)$ for any $c \in \mathbb{R}$. Keep both DGPs identical at all observed treatment-time configurations.

Because the target comparison was not run at the target time, anchor data identifying $\mu_c, \mu_d$ at source and target times places no constraint on $\mu_{a_{k^\star}}$ or $\mu_{b_{k^\star}}$ at $(t_{0k^\star}+\delta_0, t_{1k^\star}+\delta_1)$: without a cross-treatment restriction, each treatment's temporal trajectory is independently free. Both DGPs therefore generate identical observed data distributions, yet $\tau_{k^\star}^{P_2}(\delta_0, \delta_1) - \tau_{k^\star}^{P_1}(\delta_0, \delta_1) = c$. Letting $c$ range over $\mathbb{R}$ yields the identified set $\mathbb{R}$.
\end{proof}

\paragraph{Degrees-of-freedom count.}
The non-identification has a clean dimensional origin. Suppose $K$ trials are observed at $T$ distinct time configurations, with each trial contributing two conditional-mean functions $\mu_a(\cdot)$ (for its two arms). The observed data pin down $2K$ functions on $\mathcal{X}$. The TATE functional, by contrast, depends on four functions at the target time: $\mu_{a_{k^\star}}(\cdot, t_{0k^\star}+\delta_0, t_{1k^\star}+\delta_1)$, $\mu_{b_{k^\star}}(\cdot, t_{0k^\star}+\delta_0, t_{1k^\star}+\delta_1)$, and their source-time counterparts. If no observed trial shares the target time $(t_{0k^\star}+\delta_0, t_{1k^\star}+\delta_1)$ \emph{and} involves treatments $a_{k^\star}$ or $b_{k^\star}$, then the target-time conditional means enter the TATE but are absent from the observed-data likelihood: they are free parameters. A rank-$R$ factor restriction collapses these $\mathcal{O}(T \cdot |\mathcal{A}|)$ treatment-time functions into $\mathcal{O}(T + |\mathcal{A}|)$ factor pieces---the treatment- and covariate-specific effects $\boldsymbol{\theta}_a(\cdot)$ and the $R$-vector temporal effect $\boldsymbol{\Lambda}(\cdot, \cdot)$---so that observed-time constraints transmit to unobserved-time quantities through the shared structure.

\subsection{Local Smoothness Restrictions Are Insufficient}
\label{app:smoothness_insufficient}

A natural weakening of the nonparametric model replaces factor structure with a smoothness-in-time restriction: assume the conditional-mean functions are continuous or Lipschitz in time, so that nearby times inform each other through local regularity rather than through a cross-treatment factor. We show this class of restrictions cannot deliver point identification.

Fix $L > 0$ and consider the model
\begin{align*}
\mathcal{M}_L := \Bigl\{ P \in \mathcal{M}_{\mathrm{NP}} \;:\; &|\mu_a(x, t_0', t_1') - \mu_a(x, t_0, t_1)| \leq L\bigl(|t_0' - t_0| + |t_1' - t_1|\bigr) \\
&\text{for all } a, x, (t_0, t_1), (t_0', t_1') \Bigr\}.
\end{align*}
Let $\Delta_{k^\star}(\delta_0, \delta_1) := \min_{k, (t_0, t_1)}\bigl\{|t_0 - (t_{0k^\star}+\delta_0)| + |t_1 - (t_{1k^\star}+\delta_1)|\bigr\}$, where the minimum is over observed trials $k$ and times at which $a_{k^\star}$ or $b_{k^\star}$ was actually administered. Thus $\Delta_{k^\star}$ is the temporal gap between the target time and the nearest observed time at which one of the target treatments was run.

\begin{proposition}[Insufficiency of local smoothness]
\label{prop:smoothness_insufficient}
Under $\mathcal{M}_L$ and the setup of Proposition~\ref{prop:nonident}, the identified set for $\tau_{k^\star}(\delta_0, \delta_1)$ contains the interval of width at least $2L \cdot \Delta_{k^\star}(\delta_0, \delta_1)$. In particular, whenever $\Delta_{k^\star}(\delta_0, \delta_1) > 0$, the TATE is not point-identified; and as $L \to \infty$, the identified set is the entire real line.
\end{proposition}

\begin{proof}
Write $(t_0^\star, t_1^\star) := (t_{0k^\star} + \delta_0, t_{1k^\star} + \delta_1)$ for the target time and $\Delta := \Delta_{k^\star}(\delta_0, \delta_1) > 0$. Fix any $P_0 \in \mathcal{M}_L$; denote its mean functions by $\mu_a^{(0)}(x, t_0, t_1)$. For $\varepsilon \in [-\Delta, \Delta]$, define a perturbation
\[
\zeta_\varepsilon(t_0, t_1) := \varepsilon \cdot \max\!\left\{0,\; 1 - \frac{|t_0 - t_0^\star| + |t_1 - t_1^\star|}{\Delta}\right\},
\]
a tent function of height $\varepsilon$ at the target time, linearly decaying to $0$ by the time it reaches any observed time where $a_{k^\star}$ or $b_{k^\star}$ was administered. By construction, $\zeta_\varepsilon$ is $(|\varepsilon|/\Delta)$-Lipschitz in $(t_0, t_1)$, so $|\varepsilon|/\Delta \leq L$ whenever $|\varepsilon| \leq L\Delta$.

Construct $P_\varepsilon$ by setting
\[
\mu_{a_{k^\star}}^{(\varepsilon)}(x, t_0, t_1) := \mu_{a_{k^\star}}^{(0)}(x, t_0, t_1) + \zeta_\varepsilon(t_0, t_1), \qquad \mu_{b_{k^\star}}^{(\varepsilon)} := \mu_{b_{k^\star}}^{(0)},
\]
and leaving all other treatment arms unchanged. The distribution $P_\varepsilon$ remains in $\mathcal{M}_L$: (i) $\mu_{a_{k^\star}}^{(\varepsilon)}$ is Lipschitz-$L$ because it is the sum of an $L$-Lipschitz function and a $(|\varepsilon|/\Delta)$-Lipschitz function with $|\varepsilon|/\Delta \leq L$, and the Lipschitz constant is still bounded by $L$ under our choice $|\varepsilon| \leq L\Delta$; (ii) $P_\varepsilon$ agrees with $P_0$ at every observed time, because $\zeta_\varepsilon$ vanishes outside a ball of radius $\Delta$ around $(t_0^\star, t_1^\star)$, and $\Delta$ is the distance to the nearest observed time where $a_{k^\star}$ was administered.

The TATE under $P_\varepsilon$ satisfies $\tau_{k^\star}^{P_\varepsilon}(\delta_0, \delta_1) = \tau_{k^\star}^{P_0}(\delta_0, \delta_1) + \varepsilon$. Letting $\varepsilon$ range over $[-L\Delta, L\Delta]$ yields an identified set of width $2L\Delta$. As $L \to \infty$, the construction places no bound on $\varepsilon$ and the identified set is $\mathbb{R}$.
\end{proof}

The perturbation only affects the target treatment, leaving $\mu_{b_{k^\star}}$ (and all anchor arms) untouched. An identical construction perturbing both target arms by opposite-signed tents doubles the bound to $4L\Delta$. The result extends verbatim from Lipschitz smoothness to any local-regularity restriction that admits a compactly supported bump function of arbitrary height---bounded variation, bounded derivative, Sobolev $H^s$ for any finite $s$---since each admits a tent-function construction of sufficiently small norm. The common feature is that \emph{local} restrictions constrain how a function changes within a neighborhood but impose no \emph{cross-treatment} coupling, which is what transmits information from anchor trials to the target arm.

\subsection{Minimality of Separability within the Factor-Model Class}

\begin{corollary}[Minimality of separability]
\label{cor:minimality}
Consider the general factor model $\mathbb{E}[Y_{t_1}(a, t_0) \mid X = x] = \sum_{r=1}^R \theta_a^{(r)}(x) \Lambda^{(r)}(t_0, t_1)$.
\begin{enumerate}[leftmargin=*, itemsep=2pt]
    \item[(i)] If $\boldsymbol{\Lambda}$ is treatment-specific ($\boldsymbol{\Lambda}_a$) with no cross-treatment coupling, the model reduces to $\mathcal{M}_{\mathrm{NP}}$ for the target treatments, and Proposition~\ref{prop:nonident} gives non-identification.
    \item[(ii)] If the restriction is local-only (smoothness in time of $\Lambda^{(r)}$ with treatment-specific effects free to vary), Proposition~\ref{prop:smoothness_insufficient} applies and the TATE is not point-identified.
    \item[(iii)] For rank $R \geq 1$ with a shared temporal-effect vector $\boldsymbol{\Lambda}$ coupling treatments, identification requires at least $R$ independent anchor pairs (Strategy~1) or $R$ independent anchor arms (Strategy~2), as established in Corollaries~\ref{cor:replicated}--\ref{cor:common_arm}: with only $R' < R$ anchors, the effect-contrast matrix $\boldsymbol{\Theta} \in \mathbb{R}^{R' \times R}$ leaves an $(R - R')$-dimensional solution set for $\boldsymbol{\Lambda}^{\mathrm{target}}$, so the TATE is not uniquely determined.
\end{enumerate}
Hence Assumption~\ref{ass:separable} is the minimal factor restriction for which $R$ independent anchor pairs (or arms) suffice to identify the TATE, and shared-effect structure---rather than smoothness in time---is the ingredient that resolves the non-identification of Propositions~\ref{prop:nonident} and~\ref{prop:smoothness_insufficient}. The $R = 1$ case, where a single anchor pair or arm is enough, is the most commonly encountered instance.
\end{corollary}

In short, separability is not an arbitrary convenience: it is the weakest factor structure for which a matching count of anchor trials---$R$ independent anchor pairs for Strategy~1 or $R$ independent anchor arms for Strategy~2---suffices to identify the temporal effect, and neither pure nonparametric nor local-smoothness restrictions can substitute for it.

\section{Partial Identification: Full Derivation}
\label{app:sensitivity}

This appendix provides the full derivation of Proposition~\ref{prop:sensitivity} (the partial identification bound under separability violations) and the analogous result for Strategy~2.

\subsection{Setup}

Under the contaminated model~\eqref{eq:contaminated}, define the trial-specific contrast means
\begin{align*}
\bar{\tilde{\theta}}_{a,b,k} &:= \mathbb{E}[\theta_a(X) - \theta_b(X) \mid S = k], \\
\bar{\tilde{\Gamma}}_{a,b,k}(t_0, t_1) &:= \mathbb{E}[\Gamma_a(X, t_0, t_1) - \Gamma_b(X, t_0, t_1) \mid S = k].
\end{align*}
The partial identification parameter is $\gamma := \sup_{a,x,t_0,t_1} |\Gamma_a(x, t_0, t_1)|$, so $|\bar{\tilde{\Gamma}}_{a,b,k}(t_0, t_1)| \leq 2\gamma$ by the triangle inequality.

\subsection{True TATE under the Contaminated Model}

The TATE is defined as a difference of potential outcomes and is well-defined irrespective of separability:
\begin{align*}
\tau_{k^\star}^{\mathrm{true}}(\delta_0, \delta_1) 
&= \mathbb{E}\!\left[Y_{t_1+\delta_1}(a_{k^\star}, t_0+\delta_0) - Y_{t_1+\delta_1}(b_{k^\star}, t_0+\delta_0) \,\middle|\, S = k^\star\right] \\
&= \bar{\tilde{\theta}}_{a_{k^\star},b_{k^\star},k^\star} \cdot \Lambda^T + \bar{\tilde{\Gamma}}_{k^\star}^T,
\end{align*}
where we abbreviate $\Lambda^T := \Lambda(t_{0k^\star}+\delta_0, t_{1k^\star}+\delta_1)$, $\Lambda^S := \Lambda(t_{0k^\star}, t_{1k^\star})$, and $\bar{\tilde{\Gamma}}_{k^\star}^T := \bar{\tilde{\Gamma}}_{a_{k^\star},b_{k^\star},k^\star}(t_{0k^\star}+\delta_0, t_{1k^\star}+\delta_1)$ (and analogously for $S$). The source-time observed ATE is $\tau_{k^\star}^{\mathrm{obs}} = \bar{\tilde{\theta}}_{k^\star} \cdot \Lambda^S + \bar{\tilde{\Gamma}}_{k^\star}^S$.

\subsection{What the Strategy 1 Estimator Targets}

The Strategy~1 plug-in estimator converges to
\begin{equation*}
\psi_1^* = \tau_{k^\star}^{\mathrm{obs}} \cdot \frac{\tau_j^{\mathrm{obs}}}{\tau_{j'}^{\mathrm{obs}}},
\end{equation*}
where $\tau_j^{\mathrm{obs}}$ and $\tau_{j'}^{\mathrm{obs}}$ are the observed ATEs of the anchor pair $(a^\star, b^\star)$ at target and source times. Writing $\rho := \Lambda^T/\Lambda^S$ and defining contamination ratios $\alpha_{k^\star} := \bar{\tilde{\Gamma}}_{k^\star}^S / (\bar{\tilde{\theta}}_{k^\star} \Lambda^S)$, $\beta_{\mathrm{anc}}^T := \bar{\tilde{\Gamma}}_{\mathrm{anc}}^T/(\bar{\tilde{\theta}}_{\mathrm{anc}} \Lambda^T)$, $\beta_{\mathrm{anc}}^S := \bar{\tilde{\Gamma}}_{\mathrm{anc}}^S/(\bar{\tilde{\theta}}_{\mathrm{anc}} \Lambda^S)$, algebra gives
\begin{equation*}
\psi_1^* = \bar{\tilde{\theta}}_{k^\star} \cdot \Lambda^S \cdot (1 + \alpha_{k^\star}) \cdot \rho \cdot \frac{1 + \beta_{\mathrm{anc}}^T}{1 + \beta_{\mathrm{anc}}^S}.
\end{equation*}

\subsection{First-order Bias Expansion}

For small contamination ratios ($|\alpha|, |\beta^T|, |\beta^S| \ll 1$), a first-order Taylor expansion of $(1+\beta^T)/(1+\beta^S)$ yields
\begin{equation*}
\psi_1^* - \tau_{k^\star}^{\mathrm{true}} \;\approx\; \bar{\tilde{\theta}}_{k^\star} \Lambda^S \rho \cdot (\alpha_{k^\star} + \beta_{\mathrm{anc}}^T - \beta_{\mathrm{anc}}^S) - \bar{\tilde{\Gamma}}_{k^\star}^T.
\end{equation*}
Substituting back the definitions of $\alpha, \beta^T, \beta^S$ and using $\tau_{k^\star}^{\mathrm{obs}} \approx \bar{\tilde{\theta}}_{k^\star} \Lambda^S$ to first order:
\begin{equation}
\label{eq:bias_exact}
\mathrm{Bias}_1 \;\approx\; \rho \cdot \bar{\tilde{\Gamma}}_{k^\star}^S - \bar{\tilde{\Gamma}}_{k^\star}^T + \frac{\bar{\tilde{\theta}}_{k^\star}}{\bar{\tilde{\theta}}_{\mathrm{anc}}} \left( \bar{\tilde{\Gamma}}_{\mathrm{anc}}^T - \rho \cdot \bar{\tilde{\Gamma}}_{\mathrm{anc}}^S \right).
\end{equation}

\subsection{The Bound}

Applying $|\bar{\tilde{\Gamma}}| \leq 2\gamma$ to every $\bar{\tilde{\Gamma}}$ term in~\eqref{eq:bias_exact} and grouping yields the bound of Proposition~\ref{prop:sensitivity}:
\begin{equation*}
|\mathrm{Bias}_1| \;\leq\; 2\gamma \cdot (1 + |\rho|) \cdot \left( 1 + \frac{|\bar{\tilde{\theta}}_{k^\star}|}{|\bar{\tilde{\theta}}_{\mathrm{anc}}|} \right).
\end{equation*}
The bound is zero at $\gamma = 0$, linear in $\gamma$, amplified by $|\rho|$, and deteriorates when the target effect is large relative to the anchor's.

\subsection{Strategy 2 Analog}

For Strategy~2, the estimator targets $\psi_2^* = \tau_{k^\star}^{\mathrm{obs}} \cdot \bar{\mu}_{c^\star,\ell}/\bar{\mu}_{c^\star,\ell'}$. Under the contaminated model, $\bar{\mu}_{c^\star,\ell} = \bar{\theta}_{c^\star} \Lambda(t_1^{\mathrm{target}}) + \bar{\Gamma}_{c^\star}(t_0^{\mathrm{target}}, t_1^{\mathrm{target}})$, and analogously for $\ell'$. Define single-arm contamination ratios $\beta_{c^\star}^T := \bar{\Gamma}_{c^\star}^T/(\bar{\theta}_{c^\star}\Lambda^T)$ and $\beta_{c^\star}^S := \bar{\Gamma}_{c^\star}^S/(\bar{\theta}_{c^\star}\Lambda^S)$; by $|\bar{\Gamma}_{c^\star}| \leq \gamma$, $|\beta_{c^\star}| \leq \gamma/|\bar{\theta}_{c^\star}\Lambda|$.
A first-order Taylor expansion of the plug-in ratio $(1+\beta_{c^\star}^T)/(1+\beta_{c^\star}^S) \approx 1 + \beta_{c^\star}^T - \beta_{c^\star}^S$, using $\rho_2 = \Lambda^T/\Lambda^S$ and the identity $\bar{\tilde{\theta}}_{k^\star}\Lambda^S \beta_{c^\star}^T = (\bar{\tilde{\theta}}_{k^\star}/\bar{\theta}_{c^\star})(\Lambda^S/\Lambda^T)\bar{\Gamma}_{c^\star}^T$, gives
\[
\mathrm{Bias}_2 \;\approx\; \rho_2 \bar{\tilde{\Gamma}}_{k^\star}^S - \bar{\tilde{\Gamma}}_{k^\star}^T \;+\; \frac{\bar{\tilde{\theta}}_{k^\star}}{\bar{\theta}_{c^\star}} \bigl( \bar{\Gamma}_{c^\star}^T - \rho_2\, \bar{\Gamma}_{c^\star}^S \bigr),
\]
where the $\Lambda^S/\Lambda^T$ cancels against the $\rho_2$ factor introduced by the Taylor expansion (it is \emph{not} a net factor in the final bound). Applying $|\bar{\tilde{\Gamma}}|\leq 2\gamma$ and $|\bar{\Gamma}_{c^\star}|\leq \gamma$ and the triangle inequality yields the scale-invariant first-order bound
\begin{equation*}
|\mathrm{Bias}_2| \;\leq\; \gamma \cdot (1 + |\rho_2|) \cdot \left( 2 + \frac{|\bar{\tilde{\theta}}_{k^\star}|}{|\bar{\theta}_{c^\star}|} \right) + o(\gamma).
\end{equation*}

\subsection{General-$R$ Partial Identification Bound}

For general $R$, the temporal ratio is replaced by the matrix product $\boldsymbol{\Lambda}^{\mathrm{target}} = \boldsymbol{\Theta}^{-1}\boldsymbol{\tau}^{\mathrm{target}}$, and the contamination of anchor ATEs propagates through this inverse. Write the observed target-trial source-time ATE as $\boldsymbol{\tau}_{k^\star}^{\mathrm{source,obs}} = \mathbf{M}\bar{\tilde{\boldsymbol{\theta}}}_{k^\star} + \bar{\tilde{\boldsymbol{\Gamma}}}_{k^\star}^{\mathrm{source}}$ (an $R$-vector indexed by source time) and the observed anchor target-time ATE vector as $\boldsymbol{\tau}^{\mathrm{target,obs}} = \boldsymbol{\Theta}\boldsymbol{\Lambda}^{\mathrm{target}} + \bar{\tilde{\boldsymbol{\Gamma}}}^{\mathrm{target}}$, both $R$-vectors. Let $\bar{\tilde{\boldsymbol{\Gamma}}}_{k^\star}^{\mathrm{target}} \in \mathbb{R}$ denote the (scalar) target-trial contamination contrast at the target time. A first-order expansion of the plug-in estimator gives
\begin{equation*}
\mathrm{Bias}_1 \;\approx\; \bar{\tilde{\boldsymbol{\theta}}}_{k^\star}^\top \boldsymbol{\Theta}^{-1} \bar{\tilde{\boldsymbol{\Gamma}}}^{\mathrm{target}} \;-\; \bar{\tilde{\Gamma}}_{k^\star}^{\mathrm{target}} \;+\; (\boldsymbol{\Lambda}^{\mathrm{target}})^\top \mathbf{M}^{-1} \bar{\tilde{\boldsymbol{\Gamma}}}_{k^\star}^{\mathrm{source}},
\end{equation*}
where each term is a scalar: the first propagates anchor-at-target-time contamination through $\boldsymbol{\Theta}^{-1}$, the second is the target-trial's own contamination at the target time, and the third propagates the target trial's source-time contamination (an $R$-vector) through $\mathbf{M}^{-1}$ and then projects onto the target-time temporal-effect direction $\boldsymbol{\Lambda}^{\mathrm{target}}$. The first two terms inherit a factor of $\|\boldsymbol{\Theta}^{-1}\| \leq 1/\sigma_{\min}(\boldsymbol{\Theta})$ from the matrix inversion; the third inherits $\|\mathbf{M}^{-1}\|$ analogously. In the $R = 1$ case all three reduce to scalars and the expansion recovers Equation~\eqref{eq:bias_exact}.

Applying $\|\bar{\tilde{\boldsymbol{\Gamma}}}\|_\infty \leq 2\gamma$ componentwise and submultiplicativity yields the rotation-invariant first-order bound
\begin{equation}
\label{eq:bias_bound_generalR}
|\mathrm{Bias}_1| \;\leq\; 2\gamma \cdot (1 + \|\boldsymbol{\Lambda}^{\mathrm{target}}\|) \cdot \kappa(\boldsymbol{\Theta}) \cdot \left(1 + \frac{\|\bar{\tilde{\boldsymbol{\theta}}}_{k^\star}\|}{\sigma_{\min}(\boldsymbol{\Theta})}\right) + o(\gamma),
\end{equation}
which recovers the $R = 1$ bound when $\boldsymbol{\Theta}$ is scalar ($\kappa = 1$ and $\sigma_{\min} = |\bar{\tilde{\theta}}_{\mathrm{anc}}|$). The general-$R$ bound is amplified by the condition number $\kappa(\boldsymbol{\Theta})$: ill-conditioned anchors compound non-separability into larger worst-case bias. This motivates choosing anchor pairs with linearly independent effect contrasts (large $\sigma_{\min}$) when multiple are available.

\subsection{Practical Usage: Breakdown Frontier Analysis}

A practitioner would:
\begin{enumerate}[leftmargin=*, itemsep=2pt]
\item Estimate $\rho$ and $|\bar{\tilde{\theta}}_{k^\star}|/|\bar{\tilde{\theta}}_{\mathrm{anc}}|$ from data (both are observable).
\item Choose a range of $\gamma$ values reflecting domain knowledge about plausible departures from separability.
\item Plot the identified set $[\hat{\psi} - B(\gamma),\; \hat{\psi} + B(\gamma)]$ as a function of $\gamma$---the \emph{breakdown frontier}.
\item Identify the breakdown value $\gamma^*$ at which the identified set first includes zero, analogous to Rosenbaum's $\Gamma$ in observational studies. If $\gamma^*$ is large relative to the outcome scale, the conclusion is robust; if small, the finding is fragile to non-separability.
\end{enumerate}

\section{Simulation Study Details}
\label{app:sim_design}

This appendix fills in the simulation-design details cut from Section~\ref{sec:simulations} and documents the additional studies summarized in the main text.

\paragraph{Full DGP specification.}
All three DGPs share covariates $X = (X_1, X_2)^\top$ with $X_1 \sim \mathcal{N}(0, 1)$ and $X_2 \sim \text{Bernoulli}(0.5)$, three treatments $a \in \{0, 1, 2\}$ with treatment- and covariate-specific effects
\begin{align*}
\theta_0(X) &= 2 + 0.5 X_1 + 0.3 X_2, \\
\theta_1(X) &= 3 + 0.9 X_1 + 0.5 X_2, \\
\theta_2(X) &= 2.5 + 0.7 X_1 + 0.4 X_2,
\end{align*}
seasonal base temporal effect $\Lambda(t) = 1 + 0.3 \sin(2\pi t / 12)$ with a 12-period cycle, and Gaussian noise $\epsilon \sim \mathcal{N}(0, 1)$ with $\mathbb{E}[\epsilon \mid X, A] = 0$. Potential outcomes are:
\begin{itemize}[leftmargin=1.5em, itemsep=2pt]
\item \textbf{DGP~A}: $Y_{t_1}(a, t_0) = \theta_a(X) \cdot \Lambda(t_1) + \epsilon$. All assumptions hold.
\item \textbf{DGP~B}: $Y_{t_1}(a, t_0) = \theta_a(X) \cdot e^{-\beta(t_1 - t_0)} \cdot \Lambda(t_1) + \epsilon$ with $\beta = 0.1$. Separability holds with $\Lambda(t_0, t_1) = e^{-\beta(t_1 - t_0)}\Lambda(t_1)$, but Assumption~\ref{ass:measurement} is violated.
\item \textbf{DGP~C}: $Y_{t_1}(a, t_0) = \theta_a(X) \cdot \Lambda(t_1) + \delta_c \sin(2\pi t_1 / 12 + \varphi_a) + \epsilon$ with treatment-specific phase shifts $\varphi_0 = 0$, $\varphi_1 = \pi/3$, $\varphi_2 = 2\pi/3$ and $\delta_c = 0.5$. Separability itself is violated.
\end{itemize}

\paragraph{Trial structure.}
DGPs~A and C use Table~\ref{tab:trial_config_A}, where all trials share lag~2. The target trial $k^\star = 1$ compares treatment~1 against control at $(t_{0k^\star}, t_{1k^\star}) = (1, 3)$; the TATE is evaluated at target measurement time $t_1 = 9$. DGP~B uses Table~\ref{tab:trial_config_B}, where lags vary across trials.

\begin{table}[h]
\centering
\begin{minipage}{0.46\linewidth}
\centering
\caption{Trial configurations for DGPs A and C (shared lag 2; $n_k = n/6$; treatment probability 0.5).}
\label{tab:trial_config_A}
\scriptsize
\begin{tabular}{@{}cccl@{}}
\toprule
$k$ & $(a_k, b_k)$ & $(t_{0k}, t_{1k})$ & Role \\
\midrule
1 & $(1, 0)$ & $(1, 3)$ & Target \\
2 & $(1, 0)$ & $(7, 9)$ & S1 anchor (target time) \\
3 & $(1, 0)$ & $(1, 3)$ & S1 anchor (source time) \\
4 & $(2, 0)$ & $(1, 3)$ & S2 anchor (source time) \\
5 & $(2, 0)$ & $(7, 9)$ & S2 anchor (target time) \\
6 & $(1, 0)$ & $(4, 6)$ & Additional \\
\bottomrule
\end{tabular}
\end{minipage}
\hfill
\begin{minipage}{0.50\linewidth}
\centering
\caption{Trial configuration for DGP B with varying lags.}
\label{tab:trial_config_B}
\scriptsize
\begin{tabular}{@{}ccccl@{}}
\toprule
$k$ & $(a_k, b_k)$ & $(t_{0k}, t_{1k})$ & Lag & Role \\
\midrule
1 & $(1, 0)$ & $(1, 3)$ & 2 & Target \\
2 & $(1, 0)$ & $(5, 9)$ & 4 & S1 anchor (target time) \\
3 & $(1, 0)$ & $(1, 3)$ & 2 & S1 anchor (source time) \\
4 & $(2, 0)$ & $(2, 3)$ & 1 & S2 anchor (source time) \\
5 & $(2, 0)$ & $(5, 9)$ & 4 & S2 anchor (target time) \\
6 & $(1, 0)$ & $(4, 6)$ & 2 & Additional \\
\bottomrule
\end{tabular}
\end{minipage}
\end{table}

\paragraph{Estimators.}
We compare five estimators: \textbf{Oracle} uses the true temporal ratio, providing a lower bound on achievable variance; \textbf{S1} implements Strategy~1 using trials 2 and 3; \textbf{S2-C} and \textbf{S2-T} implement Strategy~2 with control ($c^* = 0$) and treatment~2 ($c^* = 2$) as anchors, respectively; \textbf{S2-M} combines both anchors via inverse-variance weighting (Section~\ref{sec:estimation}). Nuisance functions are estimated via gradient boosting with 5-fold cross-fitting. We vary $n \in \{600, 1200, 2400, 4800\}$ with $B = 500$ Monte Carlo replications.

\section{Additional Simulation Studies}
\label{app:simulations_extra}

The main-body Figure~\ref{fig:diagnostics} shows four of the principal stressors: lag-decay $\beta$, separability contamination $\gamma$, rank-2 misspecification $\delta$, and population drift $\delta_{\mathrm{drift}}$. This appendix documents the setups of those sweeps and reports three studies that did not fit in the main body: the specification-test power curve, double-robustness of nuisance misspecification, and $\Lambda$-shape invariance.

\paragraph{Rank-2 misspecification setup.}
The true DGP is rank-2: $Y = \theta_1(X)\Lambda_1(t) + \delta\,\theta_2(X)\Lambda_2(t) + \epsilon$ with $\delta \in \{0, 0.2, 0.4, 0.6, 0.8, 1.0\}$, estimated under $R = 1$. Figure~\ref{fig:diagnostics}(c) shows both Strategy~1 and Strategy~2-M stay well below the naive baseline (target-trial ATE with no temporal adjustment) for all $\delta \leq 1$: even under rank misspecification, the temporal correction provides a meaningful improvement.

\paragraph{Population drift setup.}
We shift the covariate distribution at target-time trials: $X_1 \sim \mathcal{N}(\delta_{\mathrm{drift}}, 1)$ with $\delta_{\mathrm{drift}} \in [0, 1]$. Figure~\ref{fig:diagnostics}(d) shows the DR estimator maintains near-nominal coverage across the full range, while the unadjusted estimator falls into the failure region; this confirms Theorem~\ref{thm:asymptotic}(b).

\paragraph{Specification test power curve and additional studies.}
Figure~\ref{fig:robustness_extra} reports three supporting studies. \textbf{(a)}~Specification test (Appendix~\ref{app:spec_test}): rejection rate at the nominal $5\%$ under $\beta = 0$ and rising in $\beta$, detecting measurement-time violations. \textbf{(b)}~Double robustness: bias under three nuisance configurations (both correct, wrong outcome model, wrong propensity); the DR estimator remains approximately unbiased when either nuisance is correct, and only degrades when both are misspecified. \textbf{(c)}~$\Lambda$-shape invariance: RMSE of Strategy~1 and Strategy~2-M across four qualitatively different $\Lambda(t)$ families (sinusoidal, linear, step, exponential); both strategies deliver uniformly low error, since they only require the factor-model assumption and not any parametric form for $\Lambda$.

\begin{figure}[h]
\centering
\includegraphics[width=0.95\linewidth]{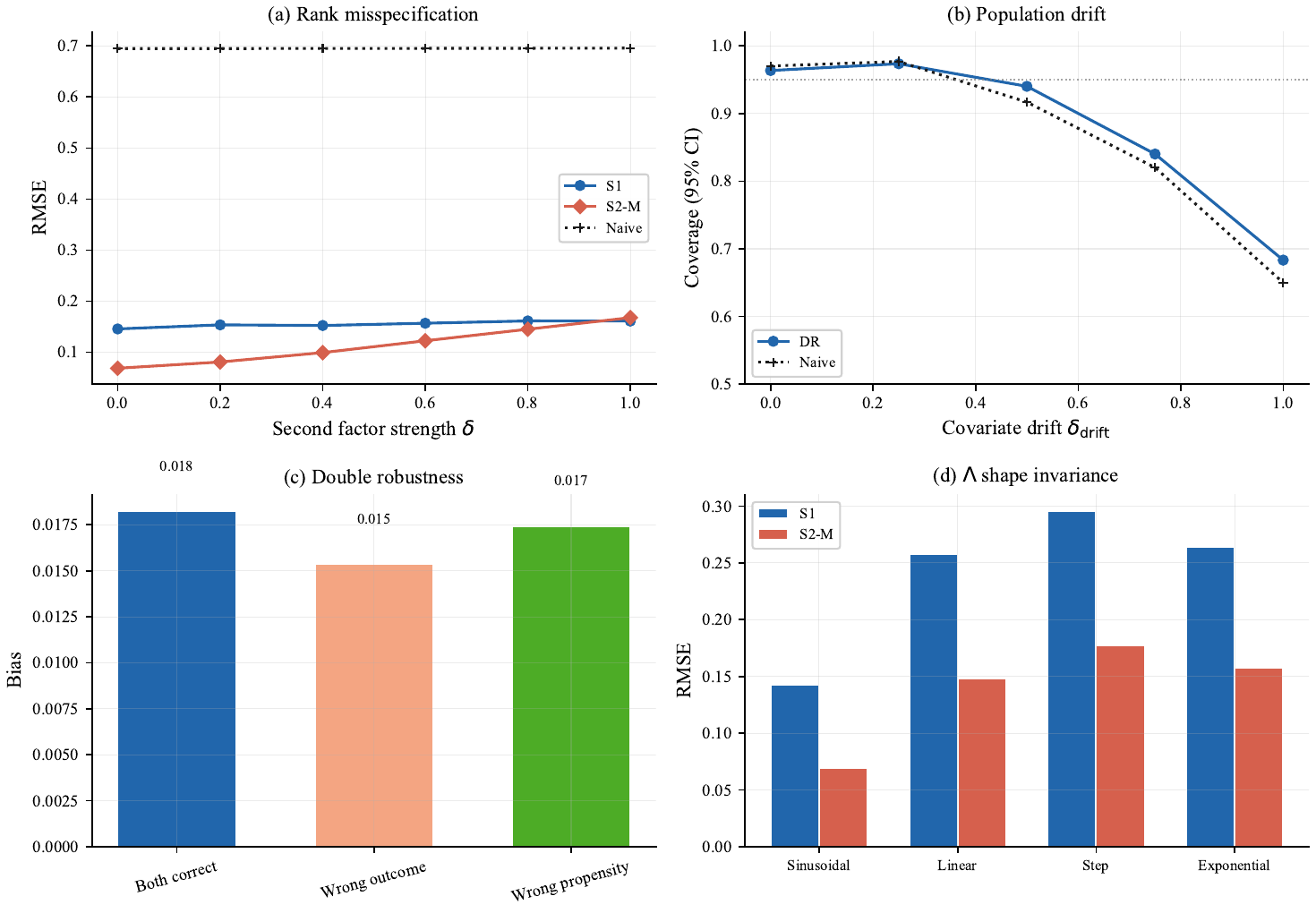}
\caption{Additional robustness studies complementing main-body Figure~\ref{fig:diagnostics}. \textbf{(a,b)}~Reproduced from main body for appendix self-containedness: rank-$2$ misspecification RMSE and population-drift coverage. \textbf{(c)}~Double robustness: bias when the outcome model, propensity scores, or both are misspecified. \textbf{(d)}~$\Lambda$-shape invariance: RMSE of Strategy~1 and Strategy~2-M across sinusoidal, linear, step, and exponential temporal effects.}
\label{fig:robustness_extra}
\end{figure}

\end{document}